\documentclass[10pt]{scrartcl}
\pdfoutput=1

\usepackage{a4wide}
\usepackage{color}

\usepackage[utf8]{inputenc}
\usepackage[T1]{fontenc}

\usepackage[english]{babel}

\usepackage{slashed}
\usepackage{caption}
\usepackage{subcaption}
\usepackage{amsmath, amssymb}
\usepackage{amsthm}
\usepackage{xspace}
\usepackage{graphicx}

\newcommand{\rnum}[1]{\uppercase\expandafter{\romannumeral #1\relax}}

\newcommand{\fer}[1]{\psi_{#1}}
\newcommand{\afer}[1]{\overline{\psi}_{#1}}
\newcommand{\gau}[1]{\lambda_{#1}}

\newcommand{\yuk}[2]{\ensuremath{\smash[t]{\Upsilon_{#1}^{\protect\phantom{#1}#2}}}\xspace}
\newcommand{\yuks}[2]{\ensuremath{\Upsilon_{#1}^{\phantom{#1}#2}}{}^*\xspace}

\newcommand{\yukw}[2]{\ensuremath{\widetilde{\Upsilon}_{#1}^{\protect\phantom{#1}#2}}\xspace}
\newcommand{\yukws}[2]{\ensuremath{\widetilde{\Upsilon}_{#1}^{\phantom{#1}#2}}{}^*\xspace}

\newcommand{\Om}[1]{\Omega_{#1}}
\newcommand{\Cw}[1]{\widetilde C_{#1}}
\def\bp{\beta'}
\def\bps{\beta'^*}

\newcommand{\yukp}[2]{\Upsilon'^{#2}_{#1}}
\newcommand{\yukps}[2]{\Upsilon'^{#2}_{#1}{}^*}

\def\sfer{\widetilde{\psi}}
\def\asfer{\overline{\widetilde{\psi}}}

\newcommand{\BBBB}[1]{\ensuremath{\mathcal{B}_{\textrm{maj}}}\xspace} 
\newcommand{\BBB}[1]{\ensuremath{\mathcal{B}_{#1}}\xspace} 
\newcommand{\BB}[1]{\ensuremath{\mathcal{B}_{#1}}\xspace} 
\newcommand{\B}[1]{\ensuremath{\mathcal{B}_{#1}}\xspace}

\newcommand{\Bc}[2]{\ensuremath{\mathcal{B}_{#1}^{#2}}\xspace}

\def\bas#1\eas{\begin{align*}#1\end{align*}}
\def\ba#1\ea{\begin{align}#1\end{align}}
\def\bi#1\ei{\begin{itemize}#1\end{itemize}}
\def\be#1\ee{\begin{enumerate}#1\end{enumerate}}
\def\nn{\nonumber}

\newcommand{\eR}[0]{\ensuremath{\epsilon_{\scalebox{.6}{R}}}\xspace}
\newcommand{\eL}[0]{\ensuremath{\epsilon_{\scalebox{.6}{L}}}\xspace}
\newcommand{\eLR}[0]{\ensuremath{\epsilon_{\scalebox{.6}{L,R}}}\xspace}

\def\szeta{\widetilde{\zeta}}

\def\maj{\Upsilon_{\mathrm{m}}}

\def\A{\mathcal{A}}
\def\E{\mathcal{E}}

\def\H{\mathcal{H}}
\def\K{\mathcal{K}}

\def\n{\mathcal{N}}

\def\P{\mathcal{P}}
\def\Q{\mathcal{Q}}
\def\com{\mathbb{C}}
\def\cS{\mathcal{S}}

\DeclareMathOperator{\End}{End}

\DeclareMathOperator{\diag}{diag}
\DeclareMathOperator{\tr}{tr}
\DeclareMathOperator{\id}{id}

\DeclareMathOperator{\ad}{ad}

\newcommand\act[1]{S_{#1}}

\newcommand{\rep}[2]{\ensuremath{\mathbf{N}_{#1} \otimes \mathbf{N}_{#2}^{o}}\xspace}
\newcommand{\repl}[2]{\ensuremath{\mathbf{#1} \otimes \mathbf{#2}^o}\xspace}
\newcommand{\srep}[1]{\ensuremath{\mathbf{N}_{#1}}\xspace}
\newcommand{\srepo}[1]{\ensuremath{\mathbf{N}_{#1}^o}\xspace}

\def\can{\ensuremath{\slashed{\partial}}\xspace}
\def\dirac{\ensuremath{\slashed{\partial}_M}\xspace}

\newcommand\inpr[2]{\langle #1, #2 \rangle}
\newcommand\rinpr[2]{( #1, #2 )}
\newcommand{\D}[2]{D_{#1}^{\phantom{#1}#2}}

\newcommand{\w}[1]{\ensuremath{\omega_{#1}}}

\def\sgnc{\epsilon}

\theoremstyle{plain}
\newtheorem{theorem}{Theorem}
\newtheorem{lem}[theorem]{Lemma}
\newtheorem{prop}[theorem]{Proposition}

\newtheorem{defin}[theorem]{Definition}
\newtheorem{rmk}[theorem]{Remark}
\newtheorem{exmpl}[theorem]{Example}
\newtheorem{cor}[theorem]{Corollary}

\usepackage{import}

\newtoks\svgpath
\svgpath={gfx/}

\graphicspath{{gfx/}}

\newcommand{\includesvg}[1]{
\subimport{\the\svgpath}{#1.pdf_tex}
}

\newcommand{\fig}[6]{\begin{figure}[#1] \centering \def\svgwidth{#2} \subimport{\the\svgpath}{#3.pdf_tex}\captionsetup{width=#4}\caption{#5}\label{fig:#6} \end{figure}}

\usepackage{hyperref}

\usepackage{tabularx}
\usepackage{booktabs}

\usepackage[affil-it]{authblk}

\author[a]{Wim Beenakker%
  \thanks{Electronic address: \texttt{W.Beenakker@science.ru.nl}}}
\author[a,b]{Thijs van den Broek%
  \thanks{Electronic address: \texttt{T.vandenBroek@science.ru.nl}}}
\author[a]{Walter D.~van Suijlekom%
  \thanks{Electronic address: \texttt{waltervs@math.ru.nl} (corresponding author)}}
\affil[a]{Radboud University Nijmegen, Institute for Mathematics, Astrophysics and Particle Physics,
Faculty of Science, PO Box 9010, 6500 GL, Nijmegen, The Netherlands}
\affil[b]{Nikhef, Science Park Amsterdam 105, 1098 XG Amsterdam}

\usepackage{parskip}
\begin{document}

\title{Supersymmetry and noncommutative geometry}
\subtitle{Part \rnum{1}: Supersymmetric almost-commutative geometries.}

\maketitle

\begin{abstract}
	Noncommutative geometry has seen remarkable applications for high energy physics, viz.~the geometrical interpretation of the Standard Model. The question whether it also allows for supersymmetric theories has so far not been answered in a conclusive way. In this first of three papers we do a systematic analysis of the possibilities for almost-commutative geometries on a $4$-dimensional, flat background to exhibit not only a particle content that is eligible for supersymmetry but also have a supersymmetric action. We come up with an approach in which we identify the basic `building blocks' of potentially supersymmetric theories and the demands for their action to be supersymmetric. Examples that satisfy these demands turn out to be sparse.
\end{abstract}

\tableofcontents

\section{Introduction}

The Standard Model of elementary particles (SM) is one of the most successful and best tested theories ever created. Yet, only few truly believe that with the SM we have reached the end of the story. Many will point at the possibility of enormous corrections (e.g.~\cite[\S 1.2]{DGR04}) that the Higgs boson mass receives from loop contributions, or at the existence of dark matter (DM, \cite{Colafrancesco2010}). Some of the more mathematically inclined even feel uneasy with quantum field theory itself.\\

The prime application of the framework of noncommutative geometry (NCG, \cite{C94}), put forward by Connes and others, is the interpretation of the SM as a geometrical theory, regarding it in some sense as a generalization of Einsteins theory of General Relativity. This line of thought, that started with the Connes-Lott model \cite{CL89}, culminated in \cite{CCM07} with the full SM, including a prediction of the Higgs boson mass. One of its essential features is that a natural notion of an action functional is associated with something that is called a \emph{noncommutative geometry}, describing a physical theory. This not only allows one to come up with a geometrical derivation of the SM particle content and action, but also to make a prediction \cite{CCM07, DS12} for the Higgs mass.\footnote{This prediction was seen to be different from the afterwards observed \cite{Consonni2013} value, but the main point here is that the approach allows one to come up with a prediction for the Higgs mass in the first place, and that the specific value depends on the particle content, as illustrated by \cite{CC12}.} This makes it a promising candidate not only for model building, but also to gain deeper insights into the realm of high energy physics. \\

The physics community, on the other hand, has made an overwhelming collective effort both in finding possible extensions of the SM, and entirely new paradigms in which the SM should appear as a low energy approximation. One of the best known and studied extensions is called \emph{supersymmetry} (e.g.~\cite{wessbagger1992}). Its key feature is that particles are accompanied by one or more superpartners; particles that carry the same quantum numbers, but differ in spin by $\tfrac{1}{2}$. Theories with such a characteristic then can have an action that is invariant under transformations that link the various particles to their superpartners. A direct application of such a theory to the SM is called the Minimal Supersymmetric Standard Model (MSSM, \cite{DGR04}), in which each of the SM particles has one ---yet unseen--- superpartner.\footnote{To be more precise: strictly speaking, the MSSM predicts a doubling of the Higgs degrees of freedom as compared to the SM.} Supersymmetry was originally devised to fully exploit the symmetries of space-time, but the MSSM has turned out to both protect the Higgs boson mass against loop corrections and provide us with a dark matter candidate. It is believed and hoped for that the MSSM can be either confirmed or convincingly falsified in the current generation of particle detectors. \\

With NCG providing us with a successful geometrical derivation of the SM Lagrangian on the one hand, and the MSSM being a potentially successful extension of the SM on the other, it is interesting to see to what extent the two can be combined. Does noncommutative geometry allow for supersymmetric theories, in particular a supersymmetric extension of the SM? Although this question has been around for some time, no one has ever come up with a conclusive answer.\\

The (predictive) power of the noncommutative method relies heavily on the principle of the spectral action that provides the link between a noncommutative geometry and its associated action \cite{CC97}. Because of this success, we ask ourselves: ``for what noncommutative geometries is the action supersymmetric?'', or ``what are supersymmetric noncommutative geometries?''. This is in contrast to the question ``what actions are supersymmetric?'' that one typically tries to answer in supersymmetry using the superfield formalism \cite{SS74}. Note the crucial difference here; the intimate connection between the noncommutative geometry and the action forbids us to manually add terms to the latter.\\

We will identify so-called \emph{building blocks}; parts of an almost-commutative geometry (an example of a noncommutative geometry that is used in obtaining particle models) that yield a supersymmetric particle content. In total five such building blocks exist, four of which require others to be defined first. In obtaining these results we translate the unimodularity condition that is commonly used to reduce the bosonic degrees of freedom to be applicable to the fermions too. Along the way a number of obstructions to a supersymmetric theory are found, that are often caused by kinetic terms of $u(1)$-particles appearing where they should not, or vice versa. These dictate the form of the finite algebra. Other obstructions lie in the impossibility of the action functional to be rewritten off shell. We set up a list of sufficient demands on the contents of an almost-commutative geometry for having a supersymmetric action. \\

The building blocks are found to be very much analogous to the ingredients of the common superfield method for supersymmetry. One of the most striking differences is that the action that corresponds to a single chiral superfield cannot be obtained in this context. This is due to the fact that noncommutative geometry describes gauge theories by nature.\\
   
This paper is organised as follows. In the upcoming section we cover those parts of noncommutative geometry that we will need later on. In Section \ref{sec:r-parity} we introduce the concept of $R$-parity ---a key notion in supersymmetry--- to NCG. In Section \ref{sec:susy-st} we give in full detail a classification of all geometries that have a supersymmetric particle content. In Section \ref{sec:4s-aux} we combine the demands that we have encountered along the way for almost-commutative geometries to also have a supersymmetric action.\\ 

We must add that this is a fairly general account. The full MSSM, with three generations of particles and all other bells and whistles will be covered in an upcoming paper.

\subsection{Noncommutative geometry and the spectral triple}\label{ch:prel}

The basic device in noncommutative geometry \cite{C94} is a {\it spectral triple} $(\A,\H,D)$ consisting of a $*$-algebra $\A$ of bounded operators on a Hilbert space $\H$, and an unbounded self-adjoint operator $D$ on $\H$, such that
\begin{enumerate}
\item the commutator $[D,a]$ is a bounded operator for all $a \in \A$;
\item the resolvent $(i+D)^{-1}$ of $D$ is a compact operator.
\end{enumerate}
One may further enrich this set of data by a \emph{grading} and a \emph{real structure}. The first is a self-adjoint operator $\gamma$ on $\H$ that commutes with all elements of $\A$, anticommutes with $D$ and is such that $\gamma^2=1$. The second is an anti-unitary operator $J$ on $\H$ implementing a right action of $\A$ on $\H$ via $J a^* J^*$, $a \in \A$. It should be such that the compatibility conditions 
\begin{align}
 &[[D, a], J b J^{-1}] = 0 \qquad\forall\ a,b \in \A, \label{eq:order_one} \intertext{and} &[a, J b J^{-1}] =0;\label{eq:left-right} \qquad\forall\ a,b \in \A
\end{align}
are satisfied. These conditions are called the {\it first-order condition} and the {\it commutant property}, respectively. The $\pm$-signs as in Table \ref{tab:ko_dimensions} for the commutation relations between $J$, $\gamma$ and $D$ determine the so-called KO-dimension of a spectral triple. A spectral triple that has a grading $\gamma$ defined on it, receives the adjective \emph{even}, whereas one on which a $J$ is defined, is called \emph{real}. We will simply write $(\A, \H, D; J, \gamma)$ for a real and even spectral triple.\\

\begin{table}[h!]
\begin{tabularx}{\textwidth}{X cccc X}
  \toprule 
  & KO-dimension & $J^2 = \epsilon$ & $JD = \epsilon' DJ$ & $J\gamma = \epsilon''\gamma J$ & 
\\
         \midrule
        & 0 & + & + & +&\\ 
        & 2 & $-$ & + & $-$ &\\
        & 4 & $-$ & + & + &\\
        & 6 & + & + & $-$&\\
        \bottomrule
\end{tabularx}
\caption{The signs of $\epsilon$, $\epsilon'$ and $\epsilon''$ for the even KO-dimensions \cite[\S 9.5]{GVF00}.}
\label{tab:ko_dimensions}
\end{table} 

The notion of a spectral triple generalizes Riemannian spin geometry to the noncommutative world, in the following way. 
\begin{exmpl}{(Canonical spectral triple \cite[Ch 6]{C94})}\label{ex:canon}
	The triple 
	\begin{align}
		(\A, \H, D) = (C^{\infty}(M), L^2(M, S), \dirac := i\slashed{\nabla}^S)\label{eq:canon}
	\end{align}
	serves as the motivating example of a spectral triple. Here $M$ is a compact Riemannian manifold that has a spin structure, $C^{\infty}(M)$ is the (commutative) algebra of smooth, complex-valued functions on $M$ and $L^2(M, S)$ denotes the square-integrable sections of the corresponding spinor bundle $S \to M$. The operator \dirac comes from the unique spin connection which in turn is derived from the Levi-Civita connection on $M$. 
This spectral triple can be dressed with a real structure $J_M$ (`charge conjugation') and ---when $\dim M$ is even--- a grading $\gamma_M \equiv \gamma^{\dim M + 1}$ (`chirality'). The KO-dimension of a canonical spectral triple is equal to the dimension of $M$.\\

In the physics parlance the canonical spectral triple roughly determines a physical \emph{system}: the algebra encodes space(-time), the Hilbert space contains spinors `living' on that space(-time) and \dirac determines how the corresponding fermions propagate.\\
\end{exmpl}

A second important example is that of a \emph{finite spectral triple}:

\begin{exmpl}{(Finite spectral triple \cite{PS96, KR97})}\label{ex:finite}
	For a finite-dimensional algebra $\A_F$, a finite-dimensional left module $\H_F$ of $\A_F$ and a symmetric matrix $D_F : \H_F \to \H_F$, we call $(\A_F, \H_F, D_F)$ a \emph{finite spectral triple}. 
\end{exmpl}

As in the general case a finite spectral triple is called real and/or even if there exists a $J_F$ (implementing a bimodule structure of $\H_F$) and/or $\gamma_F$ (that acts as a grading on $\H_F$) respectively. We will go into more detail on finite spectral triples in Section \ref{sec:finite_krajewski}.\\

We can combine Examples \ref{ex:canon} and \ref{ex:finite} to construct another important class of spectral triples:

\begin{defin}[Real, even almost-commutative geometry \cite{ISS03}]\label{def:acg}
Taking the tensor product of a real, even canonical spectral triple and a real, even finite spectral triple yields another spectral triple, called a (real, even) \emph{almost-commutative geometry}:
\begin{align*}
	(C^{\infty}(M, \A_F), L^2(M, S \otimes \H_F), \dirac + \gamma_M \otimes D_F; J_{\otimes}, \gamma_M \otimes \gamma_F),
\end{align*}
with 
\begin{align*}
	J_{\otimes} =
	\begin{cases} 
		J_M\gamma_M\otimes J_F & \text{ if } n = n_1 + n_2 \in \{1,5\}\\
		J_M \otimes J_F\gamma_F &\text{ if } n_1 \in \{2,6\} \text{ and } n_2 \text{ even}\\
		J_M \otimes J_F &\text{ otherwise}.
	\end{cases}
\end{align*}
see \cite{Dabrowski2010, Vanhecke2007}
\end{defin}
\emph{All spectral triples considered in this article will be of the above form, where the freedom lies in varying the finite spectral triple that is part of it. It is seen to describe the internal structure of the various fermion fields.}\\

\begin{defin}[Unitary equivalence of spectral triples (cf.~\cite{V06}, \S 7.1)]
  Two real and even spectral triples $(\A, \H, D; J, \gamma)$ and $(\A, \H, D'; J', \gamma')$ are said to be \emph{unitarily equivalent} if there exists a unitary operator $U : \H \to \H$ such that
\begin{align*}
    D' &= UDU^*,&
    J' &= UJU^*,&
    \gamma' &= U\gamma U^*,&
    U\pi(a)U^* &= \pi(\sigma(a))\ \forall\ a \in \A.  
\end{align*}
 Here, with $\pi$ we explicated the representation of $\A$ on $\H$ and $\sigma : \A \to \A$ is an automorphism of $\A$. 
\end{defin}

\begin{exmpl}\label{ex:gauge_group}
As an important example of such a unitary equivalence, we can form the \emph{gauge group}
\begin{align*}
	U(\A) := \{u \in \A, uu^* = u^*u = 1\}
\end{align*}
and ---in the case of a real spectral triple--- take $U := uJuJ^*$ for $u \in U(\A)$, i.e.~$\H \ni \psi \to U\psi = u\psi u^*$. Using \eqref{eq:left-right} it is then seen that $J' = J$, $\gamma' = \gamma$, 
\begin{align}
	 \sigma(a) &= uau^*\quad \forall\ u \in U(\A)\label{eq:algebra_transf},& \intertext{and}
D' &= D + u[D, u^*] + J(u[D, u^*])^*J^*.
\end{align}
\end{exmpl}

In the presence of a determinant on $\H$ we can restrict $U(\A)$ to 
\begin{align}
	SU(\A) := \{ u \in U(\A), \det{}_{\H}(u) = 1\}.\label{eq:gauge_group}
\end{align}

\subsection{Gauge fields as inner fluctuations}\label{sec:infs}

Rather than isomorphisms of algebras, a natural notion of equivalence for noncommutative\linebreak ($C^*$-)algebras is Morita equivalence \cite{Rie74}. Given a spectral triple $(\A,\H,D)$ and an algebra $\mathcal{B}$ that is Morita equivalent to $\A$, one can define \cite{C96}, \cite[\S \MakeUppercase{\romannumeral 11}]{C00} a spectral triple $(\mathcal{B},\H',D')$ with $\mathcal{B}$. This is found to be of the form 
\bas
	(\B,\ \E \otimes_{\A} \H, \nabla \otimes 1 + 1\otimes D),
\eas
where $\E$ is the $\mathcal{B}-\A$ bimodule implementing the Morita equivalence of the algebras and $\nabla$ is a connection $\nabla : \E \to \E \otimes \Omega^1_{D}(\A)$, with
\begin{align}
 \Omega^1_D(\A)  :=\big\{ \sum_i a_i[D, b_i]: a_i, b_i \in \A \big\}\label{eq:innerfluctuationform}.
\end{align}
Interestingly, upon taking $\mathcal{B}$ to be $\A$, also $\E$ is equal to $\A$, $\H' = \A \otimes_{\A} \H \simeq \H$ and $D' = D + \nabla(1)$, where the latter term means $\nabla$ acting on the identity of the algebra $\A$. This leads to a whole family of Morita equivalent spectral triples $(\A, \H, D_A)$ where $ D_A := D + A$ with self-adjoint $A \in \Omega_D^1(\A)$. The bounded operators $A$ are generally referred to as the \emph{inner fluctuations} of $D$.\\

When considering a real spectral triple $(\A, \H, D; J)$, we have the additional restriction that the real structure $J'$ of the spectral triple $(\A, \H', D'; J')$ on the Morita equivalent algebra should be compatible with the relation $J'D' = \epsilon'D'J'$. Upon taking $\mathcal{B}$ to be $\A$ again in such a case, the resulting spectral triple is of the form
$
  (\A, \H, D_A; J)
$, but now with
\begin{align}
D_A := D + A + \epsilon' JAJ^*,\qquad A \in \Omega^1_D(\A)\label{eq:inner_flucts}
\end{align}
For a real canonical spectral triple with $J_M\dirac = \dirac J_M$ these inner fluctuations vanish, due to the commutativity of the algebra.\\

The action of the gauge group (Example \ref{ex:gauge_group}) on $D_A \mapsto UD_AU^*$ induces one on the inner fluctuations: 
\begin{align}
	A \mapsto A^u := uAu^* + u[D, u^*]\label{eq:A_gauge_trans},
\end{align}
an expression that is reminiscent of the way gauge fields transform in quantum field theory.\\

Both components $\dirac$ and $D_F$ of the Dirac operator of an almost-commutative geometry (Definition \ref{def:acg}) generate inner fluctuations. 
For these we will write 
\ba\label{eq:fluctDfull}
	D_A &:= \can_A + \gamma_M \otimes \Phi,
\ea	
where $\can_A = i\gamma^\mu D_\mu$, $D_\mu = (\nabla^S + \mathbb{A})_\mu$, with 
\ba
	\mathbb{A}_\mu &= \sum_n \Big(a_n[\partial_\mu, b_n] -\epsilon' J a_n[\partial_\mu, b_n] J^*\Big)\qquad a_n, b_n \in C^{\infty}(M, \A_F),\label{eq:param_A}
\intertext{skew-Hermitian and}
	 \Phi &= D_F + \sum_{n}\Big( a_n [D_F, b_n] + \epsilon' Ja_n[D_F, b_n]J^*\Big),\qquad a_n, b_n \in C^{\infty}(M, \A_F). \nn
\ea
The relative minus sign between the two terms in $\mathbb{A}_\mu$ comes from the identity $J_M \gamma^\mu J_M^* = - \gamma^\mu$. The terms will later be seen to contain all gauge fields of the theory. The inner fluctuations of the finite Dirac operator $D_F$ (see also \eqref{eq:F-innerfl}) are seen to parametrize all scalar fields.

\subsection{The spectral action}
The above suggests that a (real) spectral triple defines a gauge theory, with the gauge fields arising as the inner fluctuations of the Dirac operator and with the gauge group given by the unitary elements in the algebra. One seeks for gauge invariant functionals of $A \in \Omega^1_D(\A)$. The so-called spectral action \cite{CC97} is the most natural one. \\

Let $(\A,\H,D; J, \gamma)$ be a real, even spectral triple. Given the operator $D_A$ of \eqref{eq:inner_flucts}, a \emph{cut-off scale} $\Lambda$ and some positive, even function $f$ one can define (cf.~\cite{C96,CC97}) the gauge invariant \emph{spectral action}:
\begin{align}
S_b[A]&:= \tr f(D_A/\Lambda),\qquad A \in \Omega^1_D(\A).\label{eq:spectral_action}
\intertext{The cut-off parameter $\Lambda$ is used to obtain an asymptotic series for the spectral action (see below); the physically relevant terms then appear with a positive power of $\Lambda$ as a coefficient. Besides this bosonic action, one can define a fermionic action:}
S_f[\zeta, A]&:= \frac{1}{2}\langle J\zeta, D_A \zeta \rangle,\qquad \zeta \in \frac{1}{2}(1 + \gamma)\H \equiv \H^+,\ A \in \Omega^1_D(\A) \label{eq:ferm_action}.
\end{align}
Using that $J^2 = \epsilon$, $DJ = \epsilon' JD$ this expression is seen to satisfy
\begin{align}
	\langle J\xi, D_A \zeta\rangle = \epsilon\epsilon' \langle J\zeta, D_A \xi\rangle\qquad\forall\ \xi, \zeta \in \H,\label{eq:symmInnerProd}
\end{align}
i.e.~it is either symmetric or antisymmetric. In its original form, the expression for the fermionic action did not feature the real structure (nor the factor $\tfrac{1}{2}$) and did not have elements of only $\H^+$ as input. It was shown \cite{CC97} that for a suitable choice of a spectral triple it does yield the full fermionic part of the Standard Model Lagrangian, including the Yukawa interactions, but suffered from the fact that the fermionic degrees of freedom were twice what they should be, as pointed out in \cite{LMMS97}. Furthermore it does not allow a theory with massive right-handed neutrinos. Adding $J$ to the expression for the fermionic action and requiring $\{J, \gamma\} = 0$ allows restricting its input to $\H^+$ without vanishing altogether. The expression \eqref{eq:ferm_action} is seen to solve both problems at the same time \cite{CCM07} (see also \cite{CM07}). We will not further go into details but refer to the mentioned literature instead. \\

The full action is then given by the sum of \eqref{eq:spectral_action} and \eqref{eq:ferm_action}:
\begin{align*}
	S[\zeta, A] &= S_f[\zeta, A] + S_b[A].
\end{align*}
For an almost-commutative geometry we will write this in particular as
\bas
	S[\zeta, \mathbb{A}, \szeta] &= S_f[\zeta, \mathbb{A}, \szeta] + S_b[\mathbb{A}, \szeta],
\eas
where $\szeta$ is the generic notation for all scalar fields in the theory, all captured by $\Phi$.\\

In order to compare the spectral action with the actions of the physics that we know, the former is approximated by a \emph{heat kernel expansion} \cite{Gil84}. Let $V$ be a vector bundle on a compact Riemannian manifold $(M, g)$. For a second-order elliptic differential operator $P : C^{\infty}(V) \to C^{\infty}(V)$ of the form
\begin{equation}
 P = - \big(g^{\mu\nu}\partial_{\mu}\partial_{\nu} + K^{\mu}\partial_{\mu} + L)\label{eq:elliptic} 
\end{equation}
with $K^{\mu}, L \in \Gamma(\End(V))$, we can expand
\begin{equation}
\tr\,e^{-tP} \sim \sum_{n \geq 0}t^{(n-m)/2}a_n(P),\qquad a_n(P) := \int_{M}a_n(x, P)\sqrt{g}d^m x,\quad \text{as } t \to 0^+\label{eq:gilkey},
\end{equation}
where $m$ is the dimension of $M$, $\sqrt{g}\mathrm{d}^mx$ (with $g \equiv \det g$) its \emph{volume form} and the coefficients $a_n(x, P)$ are called the \emph{Seeley--DeWitt coefficients} \cite[\S 11.2]{Gil84}. For an almost-commutative geometry $D_A^2$ is of the form \eqref{eq:elliptic} and one finds (for $\dim M = 4$): 
  \begin{align}
    \tr f(D_A/\Lambda) &= 2\Lambda^4 f_4 a_0(D_A^2) + 2\Lambda^2 f_2 a_2(D_A^2) + a_4(D_A^2)f(0) + \mathcal{O}(\Lambda^{-2}),\label{eq:expansion_action_functional}
  \end{align}
where the $f_k$ are moments of the function $f$,
\begin{align*}
	f_{k} := \int_{0}^{\infty} f(w)w^{k-1}dw \qquad (k>0) \nonumber. 
\end{align*}
\emph{In all cases that we will consider, the manifold will be taken four-dimensional, flat and without boundary (so that all boundary terms vanish by Stokes' Theorem).} The expansion \eqref{eq:expansion_action_functional} of the spectral action is then seen to be
\begin{align}
\tr f\bigg(\frac{D_A}{\Lambda}\bigg) &\sim \int_M \bigg[\frac{f(0)}{8\pi^2}\Big( - \frac{1}{3}\tr_F\mathbb{F}_{\mu\nu}\mathbb{F}^{\mu\nu} + \tr_F \Phi^4 + \tr_F [D_\mu, \Phi]^2\Big) \nn	\\
		&\qquad\qquad + \frac{1}{2\pi^2}\Lambda^4 f_4\tr_{F}\id - \frac{1}{2\pi^2}\Lambda^2 f_2\tr_F\Phi^2\bigg]  + \mathcal{O}(\Lambda^{-2}),\label{eq:spectral_action_acg_flat}
\end{align}
where $\tr_F$ denotes the trace over the finite Hilbert space and $\mathbb{F}_{\mu\nu}$ is the (skew-Hermitian) field strength (or curvature) of $\mathbb{A}_\mu$, i.e.
\ba\label{eq:gauge_field_strength}
	\mathbb{F}_{\mu\nu} = [(\nabla^S + \mathbb{A})_\mu, (\nabla^S + \mathbb{A})_\nu].
\ea

%
An additional constraint is imposed on the spectral triple, namely the demand that the gauge fields be traceless, as is expressed by
	\ba
		\tr_F A_\mu  = 0\label{eq:unimod},
	\ea
where with $A_\mu$ we have denoted the first term on the RHS of \eqref{eq:param_A}. This is called the \emph{unimodularity condition} and applying it removes a $u(1)$ gauge field. In fact, it turns out to be closely related the demand $\det(u) = 1$ for the gauge group (cf.~\eqref{eq:gauge_group}). Applying it in the derivation of the Standard Model one both obtains the right gauge degrees of freedom and ensures that the quarks have the correct electromagnetic interaction \cite[\S 3.5]{CCM07}. 

\subsection{Finite spectral triples and Krajewski diagrams}\label{sec:finite_krajewski}

Since we will be using real finite spectral triples (Example \ref{ex:finite}) extensively later on, we cover them in more detail. They are characterized by the following properties:
\begin{itemize}
	\item The finite-dimensional algebra is (by Wedderburn's Theorem) a direct sum of matrix algebras:
	\begin{align}
		\A_F = \bigoplus_{i}^K M_{N_i}(\mathbb{F}_i)\qquad \mathbb{F}_i = \mathbb{R}, \mathbb{C}, \mathbb{H}\label{eq:finite_algebra}.
	\end{align}	
	
	\item The finite Hilbert space is an $\A_F^\com$-bimodule. More specifically, it is a direct sum of tensor products of irreducible representations $\srep{i} \equiv \com^{N_i}$ of $M_{N_i}(\mathbb{F}_i)$, for $\mathbb{F}_i = \com, \mathbb{R}$ and\footnote{For the case $\mathbb{F}_i = \mathbb{H}$, the irreducible representation of $M_{N_i}(\mathbb{F}_i)^\com$ is $\com^{2N_i}$.} a contragredient representation \srepo{j}. 
The latter can be identified with the dual of \srep{j} (by using the canonical inner product on the latter). Thus $\H_F$ is generically of the form
	\begin{align}
		\H_F = \bigoplus_{i \leq j\leq K} \big(\rep{i}{j}\big)^{\oplus M_{N_iN_j}} \oplus \big(\rep{j}{i}\big)^{\oplus M_{N_jN_i}} \oplus \big(\rep{i}{i}\big)^{M_{N_iN_i}}
\label{eq:Hilbertspace}.
	\end{align}
The non-negative integers $M_{N_iN_j}$ denote the \emph{multiplicity} of the representation \rep{i}{j}. When various multiplicities all have one particular value $M$, we speak of ($M$) \emph{generations} that are part of a \emph{family}. 

In the rest of this paper we will not consider representations such as the last part of \eqref{eq:Hilbertspace}, since these are incompatible with $J_F\gamma_F = - \gamma_FJ_F$, necessary for avoiding the fermion doubling problem.

	\item The right $\A_F$-module structure is implemented by a real structure \ba\label{eq:fin_real} J_F :\rep{i}{j} \to \rep{j}{i}\ea that takes the adjoint: $J_F(\eta \otimes \bar\zeta) = \zeta \otimes \bar\eta$, for $\eta \in \srep{i}$ and $\zeta \in \srep{j}$. To be explicit: let $a := (a_1, \ldots, a_K) \in \A_F$ and $\eta \otimes \bar\zeta \in \rep{i}{j}$, then
\ba\label{eq:def_right_mult}
	a^o := J_Fa^*J^*_F (\eta \otimes \bar\zeta) = J_Fa^*\zeta \otimes \bar\eta = J_F(a^*_j\zeta \otimes \bar\eta) = \eta \otimes \overline{a_j^*\zeta} \equiv \eta \otimes \bar\zeta a_j.
\ea
From this it is clear that \eqref{eq:left-right} entails the compatibility of the left and right action. For the Hilbert space the existence of a real structure \eqref{eq:fin_real} implies that $M_{N_iN_j} = M_{N_jN_i}$.

	\item For each component of the algebra for which $\mathbb{F}_i = \mathbb{C}$ we will a priori allow both the (complex) linear representation \srep{i} and the anti-linear representation $\overline{\mathbf{N}}_i$, given by: 
	\begin{align*}
		\pi(m)v &:= \overline{m}v,\qquad m \in M_{N_i}(\com), v \in \mathbb{C}^{N_i}.
	\end{align*}
\item The finite Dirac operator $D_F$ consists of components
\ba
	\D{ij}{kl} : \rep{k}{l} \to \rep{i}{j}\label{eq:order_one_finite}.
\ea
The first order condition \eqref{eq:order_one} implies that any component is either left- or right-linear with respect to the algebra \cite{KR97}. This means that $i = k$ or $j = l$.\footnote{An exception to this rule is when one component of the algebra acts in the same way on more than one different representations in $\H_F$.} In both cases it is parametrized by a matrix; in the first case it constitutes of right multiplication with some $\eta_{lj} \in \rep{l}{j}$, in the second case of left multiplication with some $\eta_{ik}\in \rep{i}{k}$. 
\end{itemize}

There exists a very useful graphical representation for finite spectral triples, called \emph{Krajewski diagrams} \cite{KR97}. Such a diagram consists of a two-dimensional grid, labeled by the various $N_i$ and $N_i^o$, representing (the irreducible representations of) the algebra. Any representation \rep{i}{j} that occurs in $\H_F$ then can be represented as a \emph{vertex} on the point $(i, j)$ in this grid. If the finite spectral triple is even, each such representation has a value $\pm$ for the grading $\gamma_F$. We represent it by putting the sign in the corresponding vertex. For real spectral triples, a diagram has to be symmetric with respect to reflection around the diagonal from the upper left to the lower right corner. This is due to the role of $J_F$. The reflection of a particular vertex has the same or an opposite value for the grading, depending on whether $J_F$ commutes or anticommutes with $\gamma_F$. \\

We can represent the component $\D{ij}{kl}$ of the Dirac operator in a Krajewski diagram by an \emph{edge} from $(k,l)$ to $(i, j)$. Since the Dirac operator is self-adjoint, this means that there is also an edge from $(i, j)$ to $(k,l)$ and since it (anti)commutes with $J_F$, this means that there must also be an edge from $(l, k)$ to $(j, i)$. From the first order condition it follows \cite{KR97} that these lines can only be horizontal or vertical. We provide a particularly simple example of a Krajewski diagram in Figure~\ref{fig:kraj}, in which there are two vertices (and their conjugates) between which there is an edge.\\

\begin{figure}[ht]
\begin{center}
	\def\svgwidth{.4\textwidth}
		\subimport{\the\svgpath}{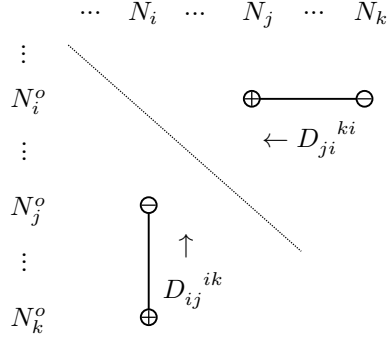}
		\caption{An example of a Krajewski diagram. Each circle in the grid stands for a representation in $\H_F$. A solid line represents a component of the Dirac operator. As can be seen from the signs, $\{J_F, \gamma_F\} = 0$ here.}
		\label{fig:kraj}
\end{center}
\end{figure}

Both as an example of the power of Krajewski diagrams and for future reference Figure \ref{fig:KrajSM} shows the diagram that fully determines (the internal structure of) the Standard Model. There, the finite algebra is taken to be 
\bas
	\A_{SM} = \com \oplus \mathbb{H} \oplus M_3(\com)
\eas 
for which we consider the representations $\mathbf{1}$, $\overline{\mathbf{1}}$, $\mathbf{2}$ and $\mathbf{3}$, determining the grid in Figure \ref{fig:KrajSM}. The particles that the SM contains are then represented as the vertices in the grid. On each point there are in fact three vertices, corresponding to the three generations of particles. Employing all demands on the finite Dirac operator it is seen \cite[\S 2.6]{CCM07} to be parametrized by the fermion mass mixing matrices $\Upsilon_{\nu,e,u,d} \in M_{3}(\com)$. Their inner fluctuations generate scalars that are interpreted as the Higgs boson doublet (solid lines), connecting the left- and right-handed representations. Furthermore we have the possibility of adding a Majorana mass $\Upsilon_R$ for the right handed neutrino (dotted line). Note that there are in principle extra components of $D_F$ possible (e.g.~from $\repl{\bar 1}{1}$ to $\repl{3}{1}$) but they are all forbidden by the additional demand
\bas
	[D_F, (\lambda, \diag(\lambda, \bar\lambda), 0)] = 0\qquad \forall\ (\lambda, \diag(\lambda, \bar\lambda), 0) \in \A_F, \lambda \in \com,
\eas
required to keep the photon massless \cite[\S 2.6]{CCM07}. \\

\begin{figure}
	\begin{center}
		\def\svgwidth{.4\textwidth}
		\subimport{\the\svgpath}{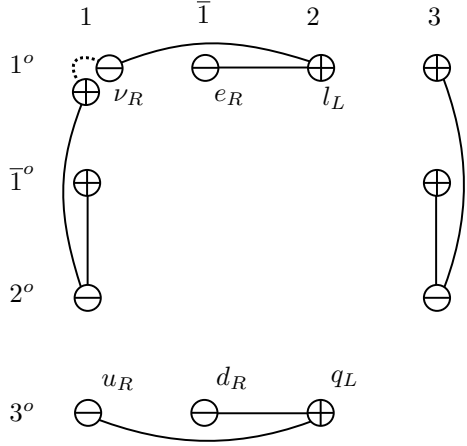}
		\caption{The Krajewski diagram representing the Standard Model.}
		\label{fig:KrajSM}	
	\end{center}
\end{figure}

The important result of \cite{KR97} is that all properties of a finite spectral triple can be read off from a Krajewski diagram. Although Krajewski diagrams were thus developed as a tool to characterize or classify finite spectral triples, they have turned out to have an applicability beyond that, e.g.~\cite{S12}. Here, we will use them also to determine the value of the trace of the second and fourth powers of the finite Dirac operator $D_F$ (or $\Phi$ after fluctuations), appearing in the action functional \eqref{eq:spectral_action_acg_flat}. We notice \cite[\S 5.4]{KR97} that 
\begin{itemize}
\item all contributions to the trace of the $n$th power of $D_F$ are given by continuous, closed paths that are comprised of $n$ edges in the Krajewski diagram. 
\item such paths can go back and forth along an edge.
\item a step in the horizontal direction corresponds to a component $\D{ij}{kl}$ of $D_F$ acting on the left of the bimodule $\H_F$, whereas a vertical step corresponds to a component $\D{ij}{kl}$ acting on the right via $J(\D{ij}{kl})^*J^*$. Due to the tensor product structure, the trace that corresponds to a certain closed path is therefore the product of the horizontal and vertical contributions.
\item if a closed path extends in only one direction, this means that the operator acts trivially on either the right or the left of the representation \rep{i}{j} at which the path started. The trace then yields an extra factor $N_i$ or $N_j$, depending on the direction of the path.
\end{itemize}
As an example we have depicted in Figure \ref{fig:KrajPaths} all possible contributions to the trace of the fourth power of a $D_F$. This is the highest power that we shall encounter, as we are interested in the action \eqref{eq:spectral_action_acg_flat}. We introduce the notation $|X|^2 := \tr_N X^*X$, for $X^*X \in M_{N}(\com)$. As an illustration of the factors appearing; in the second case a path can start at any of the three vertices, but when it starts in the middle one, it can either go first to the left or to the right. In addition, for a real spectral triple, each path appears in the same way in both directions, giving an extra factor $2$. This last argument does not hold for the last case when $k = i$ and $l = j$, however.\\

\begin{figure}
\centering
	\def\svgwidth{\textwidth}
		\subimport{\the\svgpath}{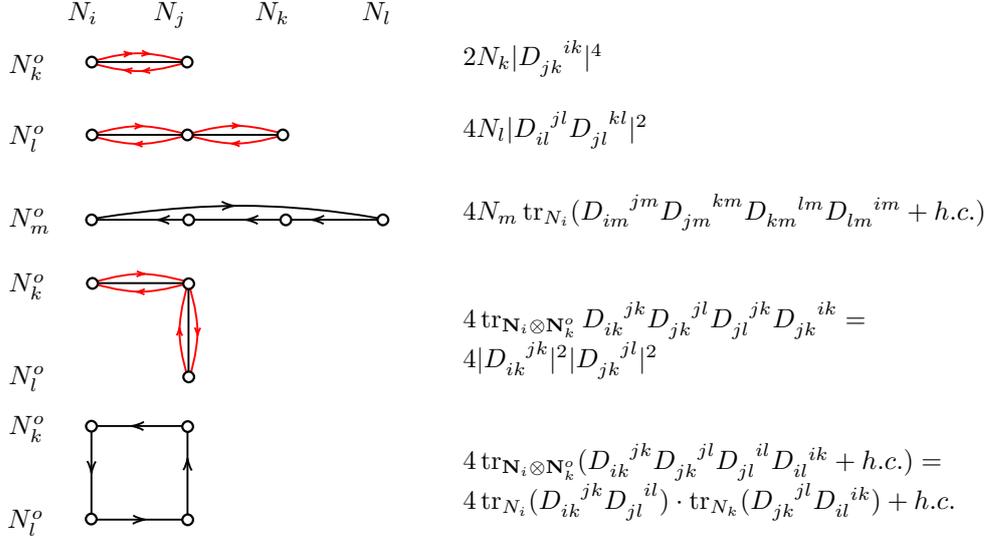}
		\caption{All types of paths contributing to the fourth power of a finite Dirac operator. The last two only occur when it is part of a real spectral triple, and we have used the notation \eqref{eq:def_right_mult} and that left-linear and right-linear components of the finite Dirac operator commute with each other. The Hermitian conjugates come from traversing a certain path in the opposite direction.}
		\label{fig:KrajPaths}	

\end{figure}

A component $\D{ij}{kj}$ of the finite Dirac operator will develop inner fluctuations \eqref{eq:inner_flucts} that are of the form
\begin{align}
	\D{ij}{kj} &\to \D{ij}{kj} + \sum_n a_n[\D{ij}{kj}, b_n],\qquad a_n, b_n \in \A\nonumber\\
							&=  \D{ij}{kj} + \sum_n (a_n)_i(\D{ij}{kj}(b_n)_k - (b_n)_i\D{ij}{kj}),\label{eq:F-innerfl}
\end{align}
where $(a_n)_i$ denotes the $i$th component of the algebra element $a_n$. It describes a scalar $\Phi_{ik}$ in the representation \rep{i}{k}. Note that the gauge group (Example \ref{ex:gauge_group}) acts on such a component in the following way:
\bas
	\D{ij}{kj} \to uJuJ^* \D{ij}{kj} u^*Ju^*J^* = u_iu_j^{*o}\D{ij}{kj}u_k^*u_j^o = u_i\D{ij}{kj}u_k^*,
\eas
whereas on an element $\psi_{ij} \in \rep{i}{j} \subset \H_F$ it acts as
\ba
	\psi_{ij} \to uJuJ^*\psi_{ij} = u_i\psi_{ij} u_j^*.\label{eq:fermion_transf}
\ea
Finally, we find for the commutator of $D_\mu$ with a component $\D{ij}{ik}$ (appearing in the action \eqref{eq:spectral_action_acg_flat}), by applying it to an element $\zeta_{kj} \in L^2(M, S\otimes \rep{k}{j})$ that:
\begin{align}
	[D_\mu, \D{ij}{kj}] \zeta_{kj} &= \partial_\mu(\Phi_{ik}\zeta_{kj}) -i g_iA_{i \mu} \Phi_{ik}\zeta_{kj} + ig_j\Phi_{ik}\zeta_{kj}A_{j \mu} -\Phi_{ik}\partial_\mu(\zeta_{kj}) \nn\\
		&\qquad +i g_k \Phi_{ik}A_{k \mu}\zeta_{kj} - ig_j\Phi_{ik}\zeta_{kj}A_{j \mu}\nonumber\\
	&= \big( \partial_\mu(\Phi_{ik}) -i g_i A_{i\mu} \Phi_{ik}  + i g_k \Phi_{ik}A_{k\mu} \big)\zeta_{kj}\nonumber\\
	&\equiv D_\mu(\Phi_{ik})\zeta_{kj}.\label{eq:commutatorExpr}
\end{align}
Here we have preliminarily introduced coupling constants $g_{i,k} \in \mathbb{R}$ and wrote $\mathbb{A}_\mu = - i g_iA_{i\mu} + ig_k A_{k\mu}^o$ (with $A_{i\mu}, A_{k\mu}$ Hermitian) to connect with the physics notation.

\subsection{NCG and R-parity}\label{sec:r-parity}

One of the key features of many supersymmetric theories is the notion of \emph{$R$-parity}; particles and their superpartners are not only characterized by the fact that they are in the same representation of the gauge group and differ in spin by $\tfrac{1}{2}$, but in addition they have opposite $R$-parity values (cf.~\cite[\S 4.5]{DGR04}). As an illustration of this fact for the MSSM, see Table \ref{tab:rpar}.\\

\begin{table}[h!]
\begin{tabularx}{\textwidth}{X lllll X}
  \toprule 
& \textbf{Fermions}& \textbf{R-parity} & \textbf{Bosons}& \textbf{R-parity} & \textbf{Multiplicity} &\\
  \midrule
& gauginos			& $-1$ 							& gauge bosons 	& $+1$ & 1 &\\
& SM fermions 	& $+1$							& sfermions 		& $-1$ & 3 &\\
& higgsinos 		& $-1$							&	Higgs(es) 		& $+1$ & 1 &\\
    \bottomrule
\end{tabularx}
\caption{The $R$-parity values for the various particles in the MSSM. In the left column are the fermions, in the right column the bosons. The SM fermions and their superpartners come in three generations each, whereas there is only one copy of the other particles. This statement presupposes that we view the up- and downtype Higgses and higgsinos as being distinct.}
\label{tab:rpar}
\end{table} 

In this section we try to mimic such properties, providing an implementation of this concept in the language of noncommutative geometry: 

\begin{defin}
	\emph{An $R$-extended, real, even spectral triple} is a real and even spectral triple $(\A, \H, D; \gamma, J)$ that is dressed with a grading 
$R : \H \to \H$ satisfying
\begin{align*}
	[R, \gamma] = [R, J] = [R, a] = 0\ \forall\ a \in \A.
\end{align*}
We will simply write $(\A, \H, D; \gamma, J, R)$ for such an $R$-extended spectral triple.
\end{defin}
Note that, as with any grading, $R$ allows us to split the Hilbert space into an \emph{$R$-even} and \emph{$R$-odd} part:
\begin{align*}
	\H = \H_{R = +} \oplus \H_{R = -}, \qquad \H_{R = \pm} = \frac{1}{2}(1 \pm R)\H.
\end{align*}
Consequently the Dirac operator splits in parts that (anti-)commute with $R$: $D = D_+ + D_-$ with $\{D_-, R\} = [D_+, R] = 0$. We anticipate what is coming in the next section by mentioning that in applying this notion to (the Hilbert space of) the MSSM, elements of $\H_{R = +}$ should coincide with the SM particles and those of $\H_{R=-1}$ with the gauginos and higgsinos. 


\begin{rmk}
	In Krajewski diagrams we will distinguish between objects on which $R = 1$ and on which $R = -1$ in the following way:
	\begin{itemize}
		\item Representations in $\H_F$ on which $R = -1$ get a black fill, whereas those on which $R = +1$ get a white fill with a black stroke.
		\item Scalars (i.e.~components of the Dirac operator) that commute with $R$ are represented by a dashed line, whereas scalars that anti-commute with $R$ get a solid line.	\end{itemize} 
\end{rmk}

We immediately use the $R$-parity operator to make a refinement to the unimodularity condition \eqref{eq:unimod}. Instead of taking the trace over the full (finite) Hilbert space, we only take it over the part on which $R$ equals 1, i.e.~it now reads
\ba
	\tr_{\H_{R =+}}A_\mu &= 0.\label{eq:unimod_new}
\ea 
Analogously, the definition \eqref{eq:gauge_group} of the gauge group must then be modified to 
\begin{align}
	SU(\A) := \{ u \in U(\A), \det{}_{\H_{R = +}}(u) = 1\}.\label{eq:gauge_group_new}
\end{align}
We will justify this choice later, after Lemma \ref{lem:gobinogo}.

Note that adjusting the unimodularity condition has no effect when applying it to the case of the NCSM, since all SM-fermions have $R$-parity $+1$ (Table \ref{tab:rpar}).

\section{Supersymmetric spectral triples}\label{sec:susy-st}

This section forms the heart of the paper. We give a classification of all almost-commutative geometries whose particle content are supersymmetric. The canonical part \eqref{ex:canon} of the almost-commutative geometries is sometimes only implicitly there. Throughout this section we characterize the finite spectral triples / almost-commutative geometries by their Krajewski diagrams as presented in Section \ref{sec:finite_krajewski}. Since gravity is known to break global supersymmetry, we shall from the outset restrict ourselves to a canonical spectral triple on a flat background, i.e.~all Christoffel symbols and consequently the Riemann tensor vanish. Unless stated otherwise we will restrict ourselves to finite algebras $\A_F$ whose components are matrix algebras over $\com$:
	\begin{align}
		\A_F = \bigoplus_i^K M_{N_i}(\com).\label{eq:alg}
	\end{align}		

\emph{For a given algebra of this form}, we look for supersymmetric `building blocks' ---made out of representations \rep{i}{j} ($i, j \in \{1, \ldots, K\}$) in the Hilbert space (fermions) and components of the finite Dirac operator (scalars)--- that give a particle content and interactions eligible for supersymmetry. In particular, these building blocks should be `irreducible'; they are the smallest extensions to a spectral triple that are necessary to retain a supersymmetric action. We underline that we do not require that the extra action associated to a building block is supersymmetric in itself. Rather, the building blocks will be defined such that the total action can remain supersymmetric, or can become it again.\\

We will start by considering all possibilities for a finite algebra consisting of one component.

\subsection{First building block: the adjoint representation}\label{sec:bb1}

For a finite algebra $\A_F = M_{N_j}(\com)$ that consists of one component, the finite Hilbert space can be taken to be $\rep{j}{j} \simeq M_{N_j}(\com)$, the bimodule of the component $M_{N_j}(\com)$ of the algebra. In order to reduce the fermionic degrees of freedom in the same way as in the NCSM, we need a finite spectral triple of KO-dimension $6$, i.e.~one that satisfies $\{J, \gamma\} = 0$. This requires at least two copies of this bimodule, both having a different value of the finite grading\footnote{We will distinguish the copies by giving them subscripts $L$ and $R$.} and a finite real structure $J_F$ that interchanges these copies (and simultaneously takes their adjoint):
\begin{align*}
	J_F(m, n) := (n^*, m^*).
\end{align*}
We call this 
\begin{defin}\label{def:bb1}
A \emph{building block of the first type} \B{j} ($j \in \{1, \ldots, K\}$) consists of two copies of an adjoint representation $M_{N_j}(\com)$ in the finite Hilbert space, having opposite values for the grading. It is denoted by
\begin{align*}
	\mathcal{B}_j &= { (m, m', 0) \in M_{N_j}(\com)_L \oplus M_{N_j}(\com)_R \oplus \End(\H_F) } \subset H_F \oplus \End(\H_F). 
\end{align*}
\end{defin}

As for the $R$-parity operator, we put $R|_{M_{N_j}(\com)} = - 1$. Since $D_A$ maps between $R = -1$ representations the gauge field has $R = 1$, indeed opposite to the fermions. The Krajewski diagram that corresponds to this spectral triple is depicted in Figure \ref{fig:bb1}.

	\begin{figure}[ht]
		\begin{center}
		\def\svgwidth{.3\textwidth}
		\subimport{\the\svgpath}{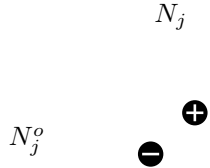}
	\captionsetup{width=.7\textwidth}
		\caption{The first building block consists of two copies in the adjoint representation $M_{N_j}(\com)$, having opposite grading. The solid fill means that they have $R = -1$.}
		\label{fig:bb1}
		\end{center}
	\end{figure}

Via the inner fluctuations \eqref{eq:inner_flucts} of the canonical Dirac operator \dirac \eqref{eq:param_A} we obtain gauge fields that act on the $M_{N_j}(\com)$ in the adjoint representation. If we write 
\begin{align*}
	(\gau{jL}', \gau{jR}') \in \H^+ = L^2(S_+ \otimes M_{N_j}(\com)_L) \oplus  L^2(S_- \otimes M_{N_j}(\com)_R)
\end{align*}
for the elements of the Hilbert space as they would appear in the inner product, we find for the fluctuated canonical Dirac operator \eqref{eq:param_A} that:
\begin{align*}
	\can_A(\gau{jL}', \gau{jR}') = i \gamma^\mu (\partial_\mu + \mathbb{A}_\mu)(\gau{jL}', \gau{jR}'),
\end{align*}
with $\mathbb{A}_\mu = - i g_j \ad A_{\mu j}'$. Here we have written $\ad(A_{\mu j}')\gau{L,R}' := A_{\mu j}'\gau{L,R}' - \gau{L,R}'A_{\mu j}'$ with $A_{\mu j}' \in \End(\Gamma(\cS) \otimes u(N_j))$ self-adjoint and we have introduced a coupling constant $g_j$.

\subsubsection{Matching degrees of freedom}\label{sec:equalizing}

In order for the gauginos to have the same number of finite degrees of freedom as the gauge bosons ---an absolute necessity for supersymmetry--- we can simply reduce their finite part $\gau{jL,R}'$ to $u(N_j)$, as described in \cite[\S 4]{BS10}. However, as is also explained in loc.~cit., even though the finite part of the gauge field $A_{\mu j}'$ is initially also in $u(N_j)$, the trace part is invisible in the action since it acts on the fermions in the adjoint representation. To be explicit, writing $A_{\mu j}' = A_{\mu j} + \tfrac{1}{N_j}B_{\mu j}\id_{N_j}$, with $A_{\mu j}(x) \in su(N_j)$, $B_{\mu j}(x) \in u(1)$ (for conciseness we have left out coupling constants for the moment), we have
\bas
	\ad(A_{\mu j}') = \ad(A_{\mu j}).
\eas
This fact spoils the equality between the number of fermionic and bosonic degrees of freedom again. We observe however that upon splitting the fermions into a traceless and trace part, i.e.~
\ba
	\gau{jL,R}' = \gau{jL,R} + \gau{jL,R}^0 \id_{N_j},\label{eq:bb1-gaugino}
\ea
the latter part is seen to fully decouple from the rest in the fermionic action \eqref{eq:ferm_action}:
\bas
	\inpr{J_M\gau{jL}'}{D_A\gau{jR}'} = \inpr{J_M\gau{jL}}{\can_A\gau{jR}} + \inpr{J_M\gau{jL}^0}{\dirac \gau{jR}^0}.
\eas
We discard the trace part from the theory.

\begin{rmk}\label{rmk:bb1-obstr}
	In particular, a building block of the first type with $N_j = 1$ does not yield an action since the bosonic interactions automatically vanish and all fermionic ones are discarded. This is remedied again in a set up such as in the next section. 
\end{rmk} 

Note that applying the unimodularity condition \eqref{eq:unimod_new} does not teach us anything here, for $\H_{R = +}$ is trivial.

One last aspect is hampering a theory with equal fermionic and bosonic degrees of freedom. There is a mismatch between the number of degrees of freedom for the theory \emph{off shell}; the equations of motion for the gauge field and gaugino constrain a different number of degrees of freedom. This is a common issue in supersymmetry and is fixed by means of a non-propagating \emph{auxiliary field}. We mimic this procedure by introducing a variable $G_j := G^a_j T^a_j \in C^{\infty}(M, su(N_j))$ ---with $T^a_j$ the generators of $su(N_j)$--- which appears in the action via:\footnote{This auxiliary field is commonly denoted by $D$. Since this letter already appears frequently in NCG, we instead take $G$ to avoid confusion.}
\begin{align}
- \frac{1}{2n_j}\int_M \tr_{N_j} G^2_j \sqrt{g}\mathrm{d}^4x.\label{eq:auxfieldG}
\end{align}
The factor $n_j$ stems from the normalization of the $T^a_j$, $\tr T^a_jT^b_j = n_j\delta^{ab}$, and is introduced so that in the action $(G^a)^2$ has coefficient $1/2$, as is customary. Typically $n_j = \frac{1}{2}$. Using the Euler-Lagrange equations we obtain $G_j = 0$, i.e.~the auxiliary field does not propagate. This means that on shell the action corresponds to what the spectral action yields us. In proving the supersymmetry of the action, however, we will work with the off shell counterpart of the spectral action.

The action of the spectral triple associated to \B{j} has been determined before (e.g.~\cite{CC96}, \cite{CC97}, \cite{Cha94}) and is given by
	\begin{align}\label{eq:SYM-0}
				\act{j}[\gau{}, \mathbb{A}] := \inpr{J_M\gau{jR}'}{\can_A \gau{jL}'} - \frac{f(0)}{24\pi^2}\int_M \tr_{\H_F} \mathbb{F}^j_{\mu\nu} \mathbb{F}^{j,\mu\nu}
 + \mathcal{O}(\Lambda^{-2}),
	\end{align}
where we have written the fermionic terms as they would appear in the path integral (cf.~\cite[\S 16.3]{CM07}).\footnote{It might seem that there are too many independent spinor degrees of freedom, but this is a characteristic feature for a theory on a Euclidean background, see e.g.~\cite{OS1,OS2,NW96} for details.} Using the notation introduced in \eqref{eq:def_right_mult} we write $\mathbb{A}_\mu = - ig_j (A_{\mu j} - A_{\mu j}^o)$ and find for the corresponding field strength \eqref{eq:gauge_field_strength}
\bas
	\mathbb{F}_{\mu\nu} &= - i g_j \big( F_{\mu\nu}^j - (F_{\mu\nu}^j)^o\big),\\
	&\qquad \text{with}\quad F_{\mu\nu}^j = \partial_\mu(A_{\nu j}) - \partial_\nu (A_{\mu j}) - ig_j [A_{\mu j}, A_{\nu j}]
\eas
Hermitian. Consequently we have in the action
\begin{align}
- \frac{f(0)}{24\pi^2}\int_M \tr_{\H_F} \mathbb{F}^j_{\mu\nu} \mathbb{F}^{j,\mu\nu} &=  \frac{1}{4}\frac{\mathcal{K}_j}{n_j}\int_M \tr_{N_j} F^j_{\mu\nu} F^{j,\mu\nu}, \nn\\
		&\qquad\qquad \text{with } \mathcal{K}_j = \frac{f(0)}{3\pi^2}n_jg_j^2 (2N_j)\label{eq:expressionK1},
\end{align}
Here we have used that for $X \in M_{N_j}(\com)$ traceless, $\tr_{M_{N_j}(\com)}(X - X^o)^2 = 2N_j\tr_{N_j}X^2$ and there is an additional factor $2$ since there are two copies of $M_{N_j}(\com)$ in $\H_F$. The expression for $\K_j$ gets a contribution from each representation on which the gauge field $A_{\mu j}$ acts, see Remark \ref{rmk:bb2-rmk} ahead. The factor $n_j^{-1}$ in front of the gauge bosons' kinetic term anticipates the same factor arising when performing the trace over the generators of the gauge group. The same thing happens for the gauginos and since we want $\gau{j}^{a}$, rather than $\gau{j}$, to have a normalized kinetic term, we scale these according to
\ba\label{eq:bb1-scaling}
	\gau{j} \to \frac{1}{\sqrt{n_{j}}}\gau{j}, \quad\text{where } \tr T^a_{j}T^b_{j} = n_{j}\delta_{ab}.
\ea
Discarding the trace part of the fermion, scaling the gauginos, introducing the auxiliary field $G_j$ and working out the second term of \eqref{eq:SYM-0} then gives us for the action
\begin{align}
	\act{j}[\gau{}, \mathbb{A}, G_j] := \frac{1}{n_j}\inpr{J_M\gau{jL}}{\can_A\gau{jR}} + \frac{1}{4}\frac{\K_j}{n_j}\int_M \tr_{N_j} F^j_{\mu\nu} F^{j,\mu\nu} - \frac{1}{2n_j}\int_M \tr_{N_j} G^2_j \label{eq:SYM}
\end{align}
with $\gau{jL,R} \in L^2(M, S_{\pm} \otimes su(N_j)_{L,R})$, $A_j \in \End(\Gamma(S) \otimes su(N_j))$ and $G_j \in C^{\infty}(M, su(N_j))$.\\ 

For this action we have:
\begin{theorem}\label{thm:bb1}
	The action \eqref{eq:SYM} of an $R$-extended almost-commutative geometry that consists of a building block \B{j} of the first type (Definition \ref{def:bb1}, with $N_j \geq 2$) is supersymmetric under the transformations 
\begin{subequations}\label{eq:bb1-transforms}
\begin{align}
	\delta A_j &= c_{j}\gamma^\mu \big[(J_M\eR, \gamma_\mu\gau{jL})_\cS + (J_M\eL, \gamma_\mu\gau{jR})_\cS\big],  \label{eq:bb1-transforms1}\\
	\delta\gau{jL,R} &= c_{j}'\gamma^\mu\gamma^\nu F^j_{\mu\nu}\eLR + c_{G_j}'G_j\eLR, \label{eq:bb1-transforms2}\\
	\delta G_j &= c_{G_j}\big[(J_M\eR, \can_A\gau{jL})_{\cS} + (J_M\eL, \can_A\gau{jR})_{\cS}\big]\label{eq:bb1-transforms3}
\end{align}
\end{subequations}
	with $c_{j}, c_{j}', c_{G_j}, c_{G_j}' \in \mathbb{C}$ iff
\ba
			2ic_{j}' &= - c_j\K_j, & c_{G_j} = - c_{G_j}'.\label{eq:bb1-constr-final}
\ea
\end{theorem}
\begin{proof}
	The entire proof, together with the explanation of the notation, is given in the Appendix \ref{sec:SYM}.
\end{proof}

We have now established that the building block of Definition \ref{def:bb1} 
 gives the super Yang-Mills action, which is supersymmetric under the transformations \eqref{eq:bb1-transforms}.\footnote{A similar result, without taking two copies of the adjoint representation, was obtained in \cite{BS10}.} This building block is the NCG-analogue of a single vector superfield in the superfield formalism.

Note that we cannot define multiple copies of the same building block of the first type without explicitly breaking supersymmetry, since this would add new fermionic degrees of freedom but not bosonic ones. This exhausts all possibilities for a finite algebra that consists of one component. 

\subsection{Second building block: adding non-adjoint representations}\label{sec:bb2}

If the algebra \eqref{eq:alg} contains two summands, we can first of all have two \emph{different} building blocks of the first type and find that the action is simply the sum of actions of the form \eqref{eq:SYM} and thus still supersymmetric.

 We have a second go at supersymmetry by adding the representation \rep{i}{j} to the finite Hilbert space, corresponding to an off-diagonal vertex in a Krajewski diagram. This introduces non-gaugino fermions to the theory. A real spectral triple then requires us to also add its conjugate \rep{j}{i}. To keep the spectral triple of KO-dimension $6$, both representations should have opposite values of the finite grading $\gamma_F$. For concreteness we choose \rep{i}{j} to have value $+$ in this section, but the opposite sign works equally well with only minor changes in the various expressions. With only this content, the action corresponding to this spectral triple can never be supersymmetric for two reasons. First, it lacks the degrees of freedom of a bosonic (scalar) superpartner. Second, it exhibits interactions with gauge fields (via the inner fluctuations of \dirac) without having the necessary gaugino degrees to make the particle content supersymmetric. However, if we also add the building blocks \B{i} and \B{j} of the first type to the spectral triple, both the gauginos are present and a finite Dirac operator is possible, that might remedy this.

\begin{lem}\label{lem:bb2-components}
	For a finite Hilbert space consisting of two building blocks \B{i} and \B{j} together with the representation \rep{i}{j} and its conjugate the most general finite Dirac operator on the basis 
		\begin{align}
			\rep{i}{j} \oplus M_{N_i}(\com)_L \oplus M_{N_i}(\com)_R \oplus M_{N_j}(\com)_L \oplus M_{N_j}(\com)_R \oplus \rep{j}{i}\label{eq:bb2-basis}. 
	\end{align}
is given by
	\begin{align}
	D_F &= 
	\begin{pmatrix}
			0 & 0 & A & 0 & B & 0\\
				0 & 0 & M_i & 0 & 0 & JA^*J^*\\
			A^* & M^*_i & 0 & 0 & 0 & 0 \\
				0 & 0 & 0 & 0 & M_j & JB^*J^* \\
			B^* & 0 & 0 & M_j^* & 0 & 0 \\
				0 & JAJ^* & 0 & JBJ^* & 0 & 0
		\end{pmatrix}\label{eq:bb2-DF}
	\end{align}
with $A : M_{N_i}(\com)_R \to \rep{i}{j}$ and $B : M_{N_j}(\com)_R \to \rep{i}{j}$.	
\end{lem}
\begin{proof}
	We start with a general $6\times 6$ matrix for $D_F$. Demanding that $\{D_F, \gamma_F\} = 0$ already sets half of its components to zero, leaving $18$ to fill. The first order condition \eqref{eq:order_one} requires all components on the upper-right to lower-left diagonal of \eqref{eq:bb2-DF} to be zero, so $12$ components are left. Furthermore, $D_F$ must be self-adjoint, reducing the degrees of freedom by a factor two. The last demand $J_FD_F = D_FJ_F$ links the remaining half components to the other half, but not for the components that map between the gauginos: because of the particular set up they were already linked via the demand of self-adjointness. This leaves the four independent components $A$, $B$, $M_i$ and $M_j$.
\end{proof}
In this paper we will set $M_i = M_j = 0$ since these components describe supersymmetry breaking gaugino masses. This will be the subject of a forthcoming paper. 

\begin{lem}\label{lem:samescalar}
	If the components $A$ and $B$ of \eqref{eq:bb2-DF} differ by only a complex number, then they generate a scalar field $\sfer_{ij}$ in the same representation of the gauge group as the fermion.
\end{lem}
\begin{proof}
We write $\D{ij}{ii} \equiv A$ and $\D{ij}{jj} \equiv B$ in the notation of \eqref{eq:order_one_finite}. First of all, recall that 
$
	\D{ij}{jj} : M_{N_j}(\com) \to \rep{i}{j}
$ 
is given by \emph{left} multiplication with an element $C_{ijj}\, \eta_{ij}$, where $\eta_{ij} \in \rep{i}{j}$ and $C_{ijj} \in \com$. Similarly,
$
	\D{ij}{ii} : M_{N_i}(\com) \to \rep{i}{j}
$ 
is given by \emph{right} multiplication with an element in \rep{i}{j}. If this differs from $\D{ij}{jj}$ by only a complex factor, it is of the form $C_{iij}\eta_{ij}$, with $C_{iij} \in \com$.

Then the inner fluctuations \eqref{eq:F-innerfl} that $\D{ij}{jj}$ develops, are of the form
\begin{align}
	\D{ij}{jj}&\to \D{ij}{jj}	+ \sum_n (a_n)_i\big(\D{ij}{jj}(b_n)_j	- (b_n)_i	\D{ij}{jj}\big) \equiv C_{ijj}\sfer_{ij}\label{eq:bb2-innfs2},
\end{align}
with which we mean left multiplication by the element 
\begin{align*}
\sfer_{ij} \equiv \eta_{ij} + \sum_n (a_n)_i[\eta_{ij} (b_n)_j - (b_n)_i \eta_{ij}] 
\end{align*}
 times the coupling constant $C_{ijj}$. The demand $JD_F = D_FJ$ (cf.~Table \ref{tab:ko_dimensions}) on $D_F$ means that $\D{ki}{ji} = J\D{ik}{ij}J^* = J(\D{ij}{ik})^*J^*$, from which we infer that the component $\D{ii}{ji}$ constitutes of left multiplication with $C_{iij}\eta_{ij}$. Its inner fluctuations are of the form
\begin{align*}
	\D{ii}{ji}&\to \D{ii}{ji}	+ \sum_n (a_n)_i\big(\D{ii}{ji}(b_n)_j	- (b_n)_i	\D{ii}{ji}\big) \equiv C_{iij}\sfer_{ij},
\end{align*}
which coincides with \eqref{eq:bb2-innfs2}. Furthermore, for $U = uJuJ^*$ with $u \in U(\A)$ we find for these components (together with the inner fluctuations) that
\begin{align*}
	U\D{ij}{ii}U &= u_i\D{ij}{ii}u_j^*, & U\D{ij}{jj}U &= u_i\D{ij}{jj}u_j^*,
\end{align*} 
establishing the result.
\end{proof}

Since the diagonal vertices have an $R$-value of $-1$, the scalar field $\sfer_{ij}$ generated by $D_F$ will always have an eigenvalue of $R$ opposite to that of the representation $\rep{i}{j} \in \H_F$. This makes the off-diagonal vertices and these scalars indeed each other's superpartners, hence allowing us to call $\sfer_{ij}$ a sfermion. The Dirac operator \eqref{eq:bb2-DF} (together with the finite Hilbert space) is visualized by means of a Krajewski diagram in Figure \ref{fig:bb2}. Note that we can easily find explicit constructions for $R \in \A_F \otimes \A^o_F$. Requiring that the diagonal representations have an $R$-value of $-1$, we have the implementations $(1_{N_i}, - 1_{N_j}) \otimes (- 1_{N_i}, 1_{N_j})^o$ and $(1_{N_i}, 1_{N_j}) \otimes (-1_{N_i}, - 1_{N_j})^o \in \A_F \otimes \A_F^o$, corresponding to the two possibilities of Figure \ref{fig:bb2}.

\begin{figure}[ht]
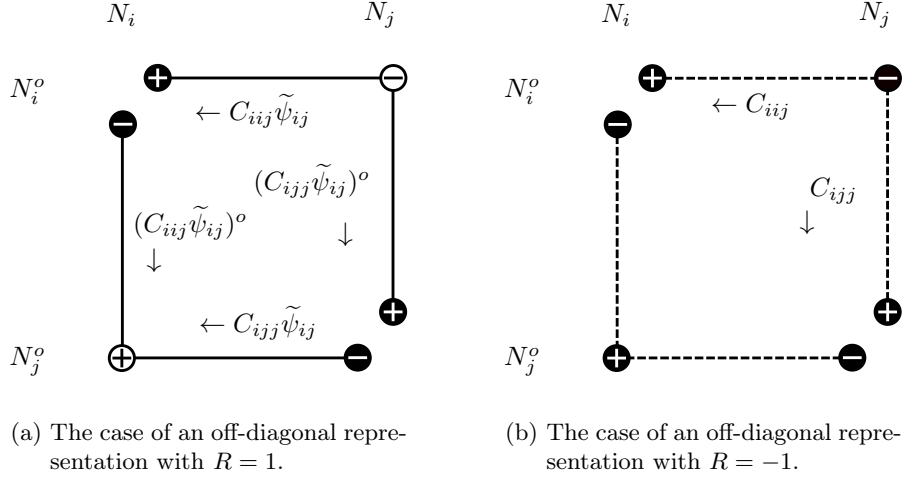

		\captionsetup{width=.9\textwidth}
	\centering
	\begin{subfigure}{.35\textwidth}
		\centering
		\def\svgwidth{\textwidth}
		\subimport{\the\svgpath}{bb2_Rplus.pdf_tex}
		\caption{The case of an off-diagonal representation with $R = 1$.}
		\label{fig:bb2_Rplus}
	\end{subfigure}
	\hspace{30pt}
	\begin{subfigure}{.35\textwidth}
		\centering
		\def\svgwidth{\textwidth}
		\subimport{\the\svgpath}{bb2_Rminus.pdf_tex}	
		\caption{The case of an off-diagonal representation with $R = -1$.}
		\label{fig:bb2_Rminus}
	\end{subfigure}
\caption{After allowing for \emph{off diagonal} representations we need a finite Dirac operator in order to have a chance at supersymmetry. The component $A$ of \eqref{eq:bb2-DF} corresponds to the upper and left lines, whereas the component $B$ corresponds to the lower and right lines. The off-diagonal vertex can have either $R = 1$ (left image) or $R = -1$ (right image). The $R$-value of the components of the finite Dirac operator changes accordingly, as is represented by the (solid/dashed) stroke of the edges.}
\label{fig:bb2}
\end{figure}

We capture this set up with the following definition:
\begin{defin}\label{def:bb2}
The \emph{building block of the second type} \Bc{ij}{\pm} consists of adding the representation \rep{i}{j} (having $\gamma_F$-eigenvalue $\pm$) and its conjugate to a finite Hilbert space containing \B{i} and \B{j}, together with maps between the representations \rep{i}{j} and \rep{j}{i} and the adjoint representations that satisfy the prerequisites of Lemma \ref{lem:samescalar}. Symbolically it is denoted by 
\begin{align*}
	\Bc{ij}{\pm} = (e_i\otimes \bar e_j, e_j'\otimes \bar e_i', \D{ii}{ji} + \D{ij}{jj}) &\in \rep{i}{j} \oplus \rep{j}{i} \oplus \End(\H_F) \\	
		&\qquad \subset \H_F \oplus \End(\H_F) 
\end{align*}
\end{defin}
When necessary, we will denote the chirality of the representation \rep{i}{j} with a subscript $L,R$. Note that such a building block is always characterized by two indices and it can only be defined when \B{i} and \B{j} have previously been defined. In analogy with the building blocks of the first type and with the Higgses/higgsinos of the MSSM in the back of our minds we will require building blocks of the second type whose off-diagonal representation in $\H_F$ has $R = -1$ to have a maximal multiplicity of $1$. In contrast, when the off-diagonal representation in the Hilbert space has $R = 1$ we can take multiple copies (`generations') of the same representation in $\H_F$, all having the \emph{same} value of the grading $\gamma_F$. This also gives rise to an equal number of sfermions, keeping the number of fermionic and scalar degrees of freedom the same, which effectively entails giving the fermion/sfermion-pair a family structure. The $C_{iij}$ and $C_{ijj}$ are then promoted to $M\times M$ matrices acting on these copies. This situation is depicted in Figure \ref{fig:bb2-gen}. We will always allow such a family structure when the fermion has $R = 1$, unless explicitly stated otherwise. There can also be two copies of a building block \B{ij} that have \emph{opposite} values for the grading. We come back to this situation in Section \ref{sec:bb5}.

\begin{figure}[ht]
\centering
	\def\svgwidth{.4\textwidth}
		\subimport{\the\svgpath}{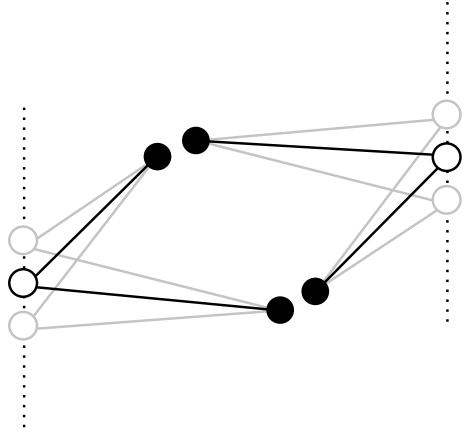}
\captionsetup{width=.9\textwidth}
\caption{An example of a building block of the second type for which the fermion has $R = 1$ and multiple generations.}
\label{fig:bb2-gen}
\end{figure}

Next, we compute the action corresponding to \B{ij}. For a generic element $\zeta$ on the finite basis \eqref{eq:bb2-basis} we will write
\begin{align*}
	\zeta = (\fer{ijL}, \gau{iL}', \gau{iR}', \gau{jL}', \gau{jR}', \afer{ijR}) \in \H^+,
\end{align*}
where the prime on the gauginos suggests that they still contain a trace-part (cf.~\eqref{eq:bb1-gaugino}). To avoid notational clutter, we will write $\fer{L} \equiv \fer{ijL}$, $\afer{R} \equiv \afer{ijR}$ and $\sfer\equiv \sfer_{ijL}$ throughout the rest of this section. The \emph{extra} action as a result of adding a building block \Bc{ij}{+} of the second type (i.e.~additional to that of \eqref{eq:SYM-0} for \B{i} and \B{j}) is given by
\begin{align}
\act{ij}[\gau{i}', \gau{j}', \fer{L}, \afer{R}, \mathbb{A}_i, \mathbb{A}_j, \sfer_{}, \asfer_{}] \equiv \act{ij}[\zeta,\mathbb{A}, \szeta] = \act{f, ij}[\zeta, \mathbb{A}, \szeta] + \act{b, ij}[\mathbb{A}, \szeta]\label{eq:bb2-action}.
\end{align}
The fermionic part of this action reads
\begin{align}
	\act{f, ij}[\zeta, \mathbb{A}, \szeta] &= \tfrac{1}{2}\inpr{J(\fer{L}, \afer{R})}{\can_A(\fer{L}, \afer{R})} \nonumber\\
	&\qquad + \tfrac{1}{2}\inpr{J(\fer{L}, \gau{iL}', \gau{iR}', \gau{jL}', \gau{jR}', \afer{R})}{\gamma^5\Phi(\fer{L}, \gau{i,L}', \gau{iR}', \gau{jL}', \gau{jR}', \afer{R})}\nn\\
&= \inpr{J_M\afer{R}}{D_A\fer{L}} + \inpr{J_M\afer{R}}{\gamma^5\gau{iR}'C_{iij}\sfer} + \inpr{J_M\afer{R}}{\gamma^5C_{ijj}\sfer\gau{jR}'}\nonumber \\
	&\qquad + \inpr{J_M\fer{L}}{\gamma^5\asfer C_{iij}^*\gau{iL}'} + \inpr{J_M\fer{L}}{\gamma^5\gau{jL}'\asfer C_{ijj}^*)},\label{eq:bb2-action-ferm}
\end{align}
prior to scaling the gauginos according to \eqref{eq:bb1-scaling}. Here we have employed \eqref{eq:bb2-innfs2} and the property \eqref{eq:symmInnerProd} of the inner product. The bosonic part of \eqref{eq:bb2-action} is given by 
\begin{align}
	\act{b, ij}[\mathbb{A}, \szeta] &= \int_M  |\n_{ij} D_\mu \sfer|^2 + \mathcal{M}_{ij}(\sfer, \asfer)\label{eq:bb2-action-term2}
\end{align}
(cf.~\eqref{eq:spectral_action_acg_flat}) with $\n_{ij} = \n_{ij}^*$ the square root of the positive semi-definite $M \times M$--matrix 
\begin{align}
\mathcal{N}_{ij}^2 &= \frac{f(0)}{2\pi^2}(N_i C_{iij}^*C_{iij} + N_jC_{ijj}^* C_{ijj}) \label{eq:exprN},\\
\intertext{where $M$ is the number of particle generations, and} 
\mathcal{M}_{ij}(\sfer, \asfer) &= \frac{f(0)}{2\pi^2}\Big[N_i|C_{iij}\sfer\asfer C_{iij}^*|^2 + N_j|\asfer C_{ijj}^*C_{ijj}\sfer|^2 + 2|C_{iij}\sfer|^2|C_{ijj}\sfer|^2\Big].\label{eq:exprM1}
\end{align}
The first term of this last equation corresponds to paths in the Krajewski diagram such as in the first example of Figure \ref{fig:KrajPaths}, involving the vertex at $(i, i)$. The second term corresponds to the same type of path but involving $(j, j)$ and the third term consists of paths going in two directions such as the fourth example of Figure \ref{fig:KrajPaths}.

\subsubsection{Matching degrees of freedom}

As far as the gauginos are concerned, there is a difference compared to the previous section; there the trace parts of the action fully decoupled from the rest of the action, but here this is not the case due to the fermion-sfermion-gaugino interactions in \eqref{eq:bb2-action}. At the same time, the gauge fields $A_{\mu i}'$ and $A_{\mu j}'$ do not act on \rep{i}{j} and \rep{j}{i} in the adjoint representation, causing their trace parts not to vanish either. We thus have fermionic and bosonic $u(1)$ fields, that are each other's potential superpartners.

We distinguish between two cases:
\begin{itemize}
\item In the left image of Figure \ref{fig:bb2} $\H_{R = +} = \rep{i}{j} \oplus \rep{j}{i}$ and thus we can employ the unimodularity condition \eqref{eq:unimod_new}. This yields\footnote{When having multiple copies of the representations \rep{i}{j} and \rep{j}{i} all expressions will be multiplied by the number of copies, since the gauge bosons act on each copy in the same way. This leaves the results unaffected, however.}

\bas
	0 &= \tr_{\rep{i}{j}}g_i'A_{i\mu}' + \tr_{\rep{j}{i}}g_j'A_{j\mu}' \\
		&= N_jg_{B_i}B_{i\mu} + N_ig_{B_j}B_{j\mu} \quad\Longrightarrow \quad B_{j\mu} = - (N_jg_{B_i}/N_ig_{B_j}) B_{i\mu}, 
\eas
	where we have first identified the independent gauge fields before introducing the coupling constants $g_{i,j}$, $g_{B_{i,j}}$ (cf.~\cite[\S 3.5.2]{CCM07}). Consequently the covariant derivative acting on the fermion $\fer{}$ and scalar $\sfer$ and their conjugates is equal to 
$\can_A = i\gamma^\mu D_\mu$ with
\bas
	 D_\mu &= \nabla^S_\mu - i\Big( g_iA_{i\mu} + \frac{g_{B_i}}{N_i}B_i\Big)  + i \Big(g_jA_{j\mu} + \frac{g_{B_j}}{N_j}B_{j}\Big)^o\\
	  	&= \nabla^S_\mu - ig_i A_{i\mu} + ig_j A_{j\mu}^o - 2ig_{B_i}\frac{B_i}{N_i}.
\eas
This also means that the kinetic terms of the $u(1)$ gauge field now appear in the action. After applying the unimodularity condition, the kinetic terms of the gauge bosons, as acting on \rep{i}{j}, are given by
\ba
	& - \tr_{\rep{i}{j}} \mathbb{F}_{\mu\nu}'\mathbb{F}'^{\mu\nu} \nn\\
	&\qquad = \tr_{\rep{i}{j}}\Big(g_iF_{\mu\nu}^i - g_jF_{\mu\nu}^{j\,o} + g_{B_i}\frac{2}{N_i}B^{i}_{\mu\nu}\Big)\Big(g_iF^{\mu\nu}_i - g_jF^{\mu\nu\,o}_{j} + g_{B_i}\frac{2}{N_i}B_{i}^{\mu\nu}\Big)\nn\\
			&\qquad = N_jg_i^2\tr_{N_i} F_{\mu\nu}^iF^{\mu\nu}_i + N_ig_j^2\tr_{N_j} F_{\mu\nu}^j F^{\mu\nu}_j + 4\frac{N_j}{N_i}g_{B_i}^2B_{\mu\nu}^iB^{\mu\nu}_i,\label{eq:bb2-gaugekinterms}
\ea
with $B_{i\mu\nu} = \partial_{[\mu} B_{i\nu]}$. The contribution from \rep{j}{i} is the same and those from \rep{i}{i} and \rep{j}{j} have been given in the previous section.

We can use the supersymmetry transformations to also reduce the fermionic degrees of freedom:
\begin{lem}\label{lem:gobinogo}
	Requiring the unimodularity condition \eqref{eq:unimod_new} also for the supersymmetry transformations of the gauge fields, makes the traces of the gauginos proportional to each other.
\end{lem}
\begin{proof}
	We introduce the notation $\gau{iL,R} = \gau{iL,R}^a \otimes T^a_i$, where $T^a_i$, $a = 0, 1, \ldots, N_i^2 - 1$ are the generators of $u(N_i) \simeq u(1) \oplus su(N_i)$. Writing out the unimodularity condition \eqref{eq:unimod_new} for the transformation \eqref{eq:bb1-transforms1} of the gauge field reads in this case
	\bas
		0 &=  N_j (g_i\tr \delta A_{i\mu} + g_{B_i}\delta B_{i\mu}) + N_i (g_j\tr \delta A_{j\mu} + g_{B_j}\delta B_{j\mu})
	\eas
	Putting in the expressions for the transformations and using that the $su(N_{i,j})$-parts of the gauginos are automatically traceless, we only retain the trace parts:
	\ba
	 0 &= N_jg_{B_i}\big[(J_M\eR, \gamma_\mu\gau{iL}^0) + (J_M\eL, \gamma_\mu\gau{iR}^0)\big] + N_ig_{B_j}\big[(J_M\eR, \gamma_\mu\gau{jL}^0) + (J_M\eL, \gamma_\mu\gau{jR}^0)\big]\nn\\
			&=\big(J_M\eR, \gamma_\mu(N_jg_{B_i}\gau{iL}^0 + N_ig_{B_j}\gau{jL}^0)\big) + (L \leftrightarrow R)\label{eq:bb2-terms-indep},
	\ea	
where with `$(L\leftrightarrow R)$' we mean the expression preceding it, but everywhere with $L$ and $R$ interchanged. Since $\epsilon = (\eL, \eR)$ can be any covariantly vanishing spinor, $(0, \eR)$ with $\nabla^S\eR = 0$ and $(\eL, 0)$ with $\nabla^S\eL = 0$ are valid solutions for which one of the terms in \eqref{eq:bb2-terms-indep} vanishes, but the other does not. The term with left-handed gauginos is thus independent from that of the right-handed gauginos. Hence, for any $\eR$, 
\bas
\big(J_M\eR, \gamma_\mu(N_jg_{B_i}\gau{iL}^0 + N_ig_{B_j}\gau{jL}^0)\big)
\eas
must vanish, establishing the result.
\end{proof}
Via the transformation \eqref{eq:bb1-transforms2} for the gaugino, we can also reduce one of the $u(1)$ parts of $G_{i,j}' = G_{i,j}^aT^a_{i,j} + H_{i,j} \in C^{\infty}(M, u(N_{i,j}))$.\\

This provides us a justification for the choice to take the trace in \eqref{eq:unimod_new} only over $\H_F$. For if we had not, we would have been in a bootstrap-like situation in which the gaugino degrees of freedom would have contributed to the relation that we have employed to reduce them by. 

\item In the right image of Figure \ref{fig:bb2} no constraint occurs due to the unimodularity condition because $\H_{R = +} = 0$ and the kinetic terms of the gauge bosons are given by:
\ba
&	- \tr_{\rep{i}{j}} \mathbb{F}_{\mu\nu}'\mathbb{F}'^{\mu\nu}\nn\\ 
&= \tr_{\rep{i}{j}}\Big(g_iF_{\mu\nu}^i - g_jF_{\mu\nu}^{j\,o} + \frac{g_{B_i}}{N_i}B_{i\mu\nu} - \frac{g_{B_j}}{N_j} B_{j\mu\nu}\Big)^2\nn\\
&= N_jg_i^2\tr_{N_i} F_{\mu\nu}^iF^{\mu\nu}_i + N_ig_j^2\tr_{N_j} F_{\mu\nu}^j F^{\mu\nu}_j + N_iN_j\Big(\frac{g_{B_i}B_i}{N_i} - \frac{g_{B_j}B_j}{N_j}\Big)_{\mu\nu}\Big(\frac{g_{B_i}B_i}{N_i} - \frac{g_{B_j}B_j}{N_j}\Big)^{\mu\nu}.\label{eq:bb2-gaugekinterms_right}
\ea
\end{itemize}

Here for the second time we stumble upon problems with the fact that the spectral action gives us an on shell action only. The problem is twofold.  First, there is ---as in the case of \B{i} and \B{j}--- a mismatch in the degrees of freedom off shell between $\fer{} \equiv \fer{ij}$ and $\sfer \equiv \sfer_{ij}$. We compensate for this by introducing a bosonic auxiliary field $F_{ij} \in C^{\infty}(M, \rep{i}{j})$ and its conjugate. They appear in the action via
\begin{align}\label{eq:bb2-auxfields}
	S[F_{ij}, F^*_{ij}] &= - \int_M \tr_{N_j} F_{ij}^* F_{ij}\sqrt{g}\mathrm{d}^4x.
\end{align}
From the Euler-Lagrange equations, it follows that $F_{ij} = F_{ij}^* = 0$, i.e.~$F_{ij}$ and its conjugate only have degrees of freedom off shell. Secondly, the four-scalar self-interaction of $\sfer$ poses an obstacle for a supersymmetric action; regardless of its specific form, a supersymmetry transformation of such a term must involve three scalars and one fermion, a term that cannot be canceled by any other. The standard solution is to rewrite these terms using the auxiliary fields $G_i'$, $G_j'$ that the building blocks of the first type provide us, such that we recover \eqref{eq:bb2-action-term2} on shell. The next lemma tells us that we can do this.

\begin{lem}\label{lem:bb2-offshell}
If $\H_{F,R = +} \ne 0$ then the four-scalar terms \eqref{eq:exprM1} of an almost-commutative geometry that consists of a single building block \B{ij} of the second type can be written in terms of auxiliary fields $G_{i,j} \in C^{\infty}(M, su(N_{i,j}))$ and $H \in C^{\infty}(M, u(1))$, as follows:
\ba
	\mathcal{L}(G_{i,j}, H, \sfer, \asfer) &= -\frac{1}{2n_i}\tr G_i^2 -\frac{1}{2n_j}\tr G_j^2 - \frac{1}{2}H^2 - \tr G_i\P_i' \sfer\asfer - \tr G_j\asfer\P_j'\sfer -  H\tr\Q'\sfer\asfer,\label{eq:bb2-auxterms}
\ea
where in the terms featuring $G_{i,j}$ the trace is over the $N_{i,j}\times N_{i,j}$-matrices and with
\begin{align}
	\P_i' &= \sqrt{\frac{f(0)}{\pi^2n_i} N_i} C_{iij}^*C_{iij},&
	\P_j' &= \sqrt{\frac{f(0)}{\pi^2n_j} N_j} C_{ijj}^*C_{ijj},&
	\Q' &= \sqrt{\frac{f(0)}{\pi^2}} (C_{iij}^*C_{iij} + C_{ijj}^*C_{ijj})\label{eq:solP}
\end{align}
matrices on $M$-dimensional family space.
\end{lem}
\begin{proof}
	Required for any building block \B{ij} of the second type are the building blocks \B{i} and \B{j} of the first type, initially providing auxiliary fields $G_{i,j} \equiv G^{a}_{i,j} T^a_{i,j} \in C^{\infty}(M, su(N_{i,j}))$ and $H_{i,j} \in C^{\infty}(M, u(1))$. Here the $T^a_{i,j}$ denote the generators of $su(N_{i,j})$ in the fundamental (defining) representation and are normalized according to $\tr T^a_{i,j}T^b_{i,j} = n_{i,j}\delta_{ab}$, where $n_{i,j}$ is the \emph{constant of the representation}. After applying the unimodularity condition \eqref{eq:unimod_new} in the case that $\H_{R = +} \ne 0$ (the left image of Figure \ref{fig:bb2}) for the gauge field and its transformation, only one $u(1)$ auxiliary field $H$ remains. We thus consider the Lagrangian \eqref{eq:bb2-auxterms} with $\P_{i,j}', \Q'$ self-adjoint. (These coefficients are written inside the trace since they may have family indices. However, the combinations $\P_i'\sfer_{}\asfer{}$ and $\asfer_{}\P_j'\sfer_{}$ cannot have family-indices anymore, since $G_i$ and $G_j$ do not.) Applying the Euler-Lagrange equations to this Lagrangian yields
\bas
	G_i^a &= - \tr T^a_i\P_i'\sfer\asfer, &
	G_j^a &= - \tr T^a_j\asfer\P_j'\sfer, &
	H &= - \tr\Q'\sfer\asfer
\eas
and consequently \eqref{eq:bb2-auxterms} equals \emph{on shell} 
\bas
\mathcal{L}(G_{i,j}, H, \sfer, \asfer) &= \frac{1}{2}\tr (T_i^a\P_i'\sfer_{ij}\asfer_{ij})^2 + \frac{1}{2}\tr (T_j^a\asfer_{ij}\P_j'\sfer_{ij})^2 + \frac{1}{2}\tr (\Q'\sfer_{ij}\asfer_{ij})^2\nn\\
&= \frac{n_i}{2}\Big(|\P_i'\sfer\asfer|^2 - \frac{1}{N_i}|\P'^{1/2}_i\sfer|^4\Big) + \frac{n_j}{2}\Big(|\asfer\P_j'\sfer|^2 - \frac{1}{N_j}|\P_j'^{1/2}\sfer|^4\Big) + \frac{1}{2}|\Q'^{1/2}\sfer|^4.
\eas
Here we have employed the identity
\ba\label{eq:idn-sun-gens}
	(T^a_{i,j})_{mn} (T^a_{i,j})_{kl} = n_{i,j}\Big(\delta_{ml}\delta_{kn} - \frac{1}{N_{i,j}}\delta_{mn}\delta_{kl}\Big).
\ea
With the choices \eqref{eq:solP}
we indeed recover the four-scalar terms \eqref{eq:exprM1} of the spectral action. 
\end{proof}

Even though in the case that $\H_{F,R = +} = 0$ (the right image of Figure \ref{fig:bb2}) the unimodularity condition cannot be used to relate the $u(1)$ fields $H_i$ and $H_j$ to each other, a similar solution is possible:
\begin{cor}
If $\H_{R = +} = 0$ then the four-scalar terms \eqref{eq:exprM1} of a building block \B{ij} of the second type can be written off shell using the Lagrangian
\ba
	\mathcal{L}(G_{i,j}, H_{i,j}, \sfer, \asfer) &= -\frac{1}{2n_i}\tr G_i^2 -\frac{1}{2n_j}\tr G_j^2 - \frac{1}{2}H_i^2 - \frac{1}{2}H_j^2 -  \tr G_i\P_i'\sfer\asfer - \tr G_j\asfer \P_j'\sfer \nn\\&\qquad- H_i \tr\Q'_i\sfer\asfer - H_j\tr\Q'_j\sfer\asfer,\label{eq:bb2-auxterms2}
\ea
with
\begin{align*}
	\P'_i &= \sqrt{\frac{f(0)}{\pi^2n_i} N_i} C_{iij}^*C_{iij},&
	\P'_j &= \sqrt{\frac{f(0)}{\pi^2n_j} N_j} C_{ijj}^*C_{ijj},&
	\Q'_i &= \Q'_j = \sqrt{\frac{f(0)}{2\pi^2}}(C_{iij}^*C_{iij} + C_{ijj}^*C_{ijj}),
\end{align*}
not carrying a family-index.
\end{cor}

In both cases we have obtained a system that has equal bosonic and fermionic degrees of freedom, both on shell and off shell.

\subsubsection{The final action and supersymmetry}

 We first turn to the case that $\H_{R = + }\ne 0$. Reducing the degrees of freedom by identifying half of the $u(1)$ fields with the other half and rewriting \eqref{eq:bb2-action} to an off shell action 
we find the \emph{extra} contributions
\begin{align*}
	&\inpr{J_M\afer{R}}{D_A\fer{L}} + \inpr{J_M\afer{R}}{\gamma^5(\gau{iR}'C_{iij}\sfer + C_{ijj}\sfer\gau{jR}')}\nonumber \\
	&\qquad + \inpr{J_M\fer{L}}{\gamma^5(\asfer C_{iij}^*\gau{iL}' + \gau{jL}'\asfer C_{ijj}^*)}\nonumber\\
	 & \qquad + \int_M\Big[|\n_{ij}D_\mu \sfer|^2 - \tr_{N_i}\big( \P_i'\sfer\asfer G_i\big) - \tr_{N_j}\big(\asfer \P_j'\sfer G_j\big) - H\tr_{N_i} \Q'\sfer\asfer - \tr_{N_j^{\oplus M}}F_{ij}^*F_{ij}\Big]
\end{align*}
to the total action, with 
\bas
	\gau{i}' &= \gau{i} + \gau{i}^0 \id_{N_i},&
	\gau{j}' &= \gau{j} - N_{j/i} \gau{i}^0 \id_{N_j}
\eas
and $G_{i,j} \in C^{\infty}(M, su(N_{i,j}))$, $H \in C^{\infty}(M, u(1))$. For notational convenience we will suppress the subscripts in the traces when no confusion is likely to arise. In addition, adding a building block \B{ij} slightly changes the expressions for the pre-factors of the kinetic terms of $A_{i\mu}$ and $A_{j\mu}$ (cf.~Remark \ref{rmk:bb2-rmk} below). \\

As a final step we scale the sfermion $\sfer_{ij}$ according to 
\ba\label{eq:bb3-scalingfields}
	\sfer_{ij}&\to \n_{ij}^{-1} \sfer_{ij},&
	\asfer_{ij}&\to  \asfer_{ij}\n_{ij}^{-1},
\ea
and the gauginos according to \eqref{eq:bb1-scaling} to give us the correctly normalized kinetic terms for both:
\begin{align}
	&\inpr{J_M\afer{R}}{D_A\fer{L}} + \inpr{J_M\afer{R}}{\gamma^5[\gau{iR}'\Cw{i,j}\sfer + \Cw{j,i}\sfer\gau{jR}']}\nonumber \\
	&\qquad + \inpr{J_M\fer{L}}{\gamma^5[\asfer \Cw{i,j}^*\gau{iL}' + \gau{jL}'\asfer \Cw{j,i}^*]}\nonumber\\
	 & \qquad + \int_M\Big[|D_\mu \sfer|^2 - \tr\big( \P_i\sfer\asfer G_i\big) - \tr\big(\asfer \P_j\sfer G_j\big) - \tr HQ\sfer\asfer - \tr_{N_j^{\oplus M}}F_{ij}^*F_{ij}\Big]\label{eq:bb2-action-offshell}.
\end{align}
Here we have written
\ba	\label{eq:bb2-scalingG}
\Cw{i,j} &:= \frac{C_{iij}}{\sqrt{n_i}}\n_{ij}^{-1},& \Cw{j,i} &:= \frac{C_{ijj}}{\sqrt{n_j}}\n_{ij}^{-1}, &
	\P_{i,j} &:= \n_{ij}^{-1} \P_{i,j}' \n_{ij}^{-1} &
	\Q &:= \n_{ij}^{-1} \Q' \n_{ij}^{-1}
\ea	
for the scaled versions of the parameters. For this action we have:

\begin{theorem}\label{prop:bb2}
		The total action that is associated to $\B{i} \oplus \B{j} \oplus \B{ij}$, given by \eqref{eq:SYM} and \eqref{eq:bb2-action-offshell}, is supersymmetric under the transformations \eqref{eq:bb1-transforms},
\begin{subequations}\label{eq:susytransforms4}
\begin{align}
				\delta \sfer &= c_{ij}(J_M\epsilon_L, \gamma^5 \fer{L})_{\cS},  &\delta \asfer &=  c_{ij}^*(J_M\epsilon_R, \gamma^5 \afer{R})_{\cS}\label{eq:transforms4.1},\\
				\delta \fer{L} &= c_{ij}' \gamma^5 [\can_A, \sfer]\epsilon_R + d_{ij}'F_{ij}\epsilon_L,  & \delta \afer{R} &=  c_{ij}'^*\gamma^5 [\can_A, \asfer]\epsilon_L + d_{ij}'^* F_{ij}^*\epsilon_R\label{eq:transforms4.2}
\
\end{align}
\end{subequations}
and
\begin{subequations}\label{eq:susytransforms5}
\begin{align}
	\delta F_{ij} &= d_{ij}(J_M\eR, \can_A\fer{L})_{\cS} + d_{ij,i}(J_M\eR, \gamma^5\gau{iR}\sfer)_{\cS} - d_{ij,j}(J_M\eR, \gamma^5\sfer\gau{jR})_{\cS},\label{eq:transforms4.4a}\\
	 \delta F_{ij}^* &= d_{ij}^*(J_M\eL, \can_A\afer{R})_{\cS} + d_{ij,i}^*(J_M\eL, \gamma^5\asfer\gau{iL})_{\cS} - d_{ij,j}^*(J_M\eL, \gamma^5\gau{jL}\asfer)_{\cS}\label{eq:transforms4.4b},
\end{align}
\end{subequations}
with $c_{ij}, c_{ij}', d_{ij}, d_{ij}', d_{ij,i}$ and $d_{ij,j}$ complex numbers, if and only if
\ba
	\Cw{i,j} &= \sgnc_{i,j}\sqrt{\frac{2}{\K_i}}g_i\id_M,& \Cw{j,i} &=  \sgnc_{j,i}\sqrt{\frac{2}{\K_j}}g_j\id_M,& 
\P_{i}^2 &= \frac{g_{i}^2}{\K_{i}}\id_M, & \P_{j}^2 &= \frac{g_{j}^2}{\K_{j}}\id_M,\label{eq:bb2-resultCiij}
\ea
for the unknown parameters of the finite Dirac operator (where $\id_M$ is the identity on family-space, which equals unity of $\fer{ij}$ has no family index) and 
	\bas
		c_{ij}' &=  c_{ij}^* = \sgnc_{i,j} \sqrt{2\K_i}c_i = - \sgnc_{j,i} \sqrt{2\K_j}c_j,\nn\\ 
		 d_{ij} &= d_{ij}'^* = \sgnc_{i,j} \sqrt{\frac{\K_i}{2}} \frac{d_{ij,i}}{g_i}  = - \sgnc_{j,i} \sqrt{\frac{\K_j}{2}} \frac{d_{ij,j}}{g_j},&
		c_{G_i} &= \sgnc_i \sqrt{\K_i} c_i,
	\eas	
	 with $\sgnc_{i}, \sgnc_{i,j}, \sgnc_{j,i} \in \{\pm 1\}$ for the transformation constants.
	\end{theorem}
	\begin{proof}
			Since the action \eqref{eq:SYM} is already supersymmetric by virtue of Theorem \ref{thm:bb1}, we only have to prove that the same holds for the contribution \eqref{eq:bb2-action-offshell} to the action from \B{ij}. The detailed proof of this fact can be found in Appendix \ref{sec:bb2-proof}.
	\end{proof}

Then for $C_{iij}$ and $\P_{i,j}$ that satisfy these relations (setting $\K_{i,j} = 1$), the supersymmetric action (but omitting the $u(1)$-terms for conciseness now) reads:
\begin{align}
	&\inpr{J_M\afer{R}}{\can_A\fer{L}} + \sqrt{2}\inpr{J_M\afer{R}}{\gamma^5(\sgnc_{i,j} g_i\gau{iR}\sfer + \sgnc_{j,i} \sfer g_j\gau{jR})}\nonumber \\
	&\qquad + \sqrt{2}\inpr{J_M\fer{L}}{\gamma^5(\sgnc_{i,j} \asfer g_i\gau{iL} + \sgnc_{j,i} g_j\gau{jL}\asfer )}\nonumber\\
	 & \qquad + \int_M\Big[|D_\mu \sfer|^2 - g_i\tr_{N_i}\big(\sfer\asfer G_i\big) - g_j\tr_{N_j}\big(\asfer \sfer G_j\big) - \tr_{N_j^{\oplus M}} F_{ij}^*F_{ij}\Big]\label{eq:bb2-action-susy},
\end{align}
i.e.
we recover the pre-factors for the fermion-sfermion-gaugino and four-scalar interactions that are familiar for supersymmetry. The signs $\sgnc_{i,j}$ and $\sgnc_{j,i}$ above can be chosen freely.

\begin{rmk}\label{rmk:bb2-obstr}In the case that $\H_{R = +} = 0$, there is an interaction 
\ba
	\propto \int_M B_{i\mu\nu}B^{\mu\nu}_j\label{eq:u1cross}
\ea
present (see the last term of \eqref{eq:bb2-gaugekinterms_right}). Transforming the gauge fields appearing in that interaction shows that the supersymmetry of the total action requires an interaction 
\bas
	\propto \inpr{J_M\gau{i}^0}{\dirac \gau{j}^0}, 
\eas
a term that the fermionic action does not provide. Thus, a situation in which there are two different $u(1)$ fields that both act on the same representation \rep{i}{j} is an obstruction for supersymmetry. This is also the reason that a supersymmetric action with gauge groups $U(N_{i,j})$ is not possible in the presence of a representation \rep{i}{j}, since 
\bas
- \tr_{\rep{i}{j}}\mathbb{F}_{\mu\nu} \mathbb{F}^{\mu\nu}	&= \tr_{\rep{i}{j}}(g_iF^i_{\mu\nu} - g_jF^{j\,o}_{\mu\nu})(g_iF_i^{\mu\nu} - g_jF_{j}^{\mu\nu\,o}) \\
			&= N_jg_i^2\tr F^i_{\mu\nu}F_i^{\mu\nu} + N_ig_j^2\tr F^j_{\mu\nu}F_j^{\mu\nu} - 2g_ig_j\tr F^i_{\mu\nu}\tr F_j^{\mu\nu}, 
\eas
of which the last term spoils supersymmetry. Averting a theory in which two independent $u(1)$ gauge fields act on the same representation will be seen to put an important constraint on realistic supersymmetric models from noncommutative geometry.
\end{rmk}

Note that it is not per se the presence of an $R = -1$ off-diagonal fermion in the first place that is causing this; in a spectral triple that contains at least one $R = +1$ fermion the interaction \eqref{eq:u1cross} vanishes due to the unimodularity condition \eqref{eq:unimod_new}.

\begin{rmk}\label{rmk:bb2-rmk}
	In the previous section we have compactly written
	\begin{align*}
		\K_i = \frac{f(0)}{3\pi^2}2N_ig_i^2n_i
	\end{align*}
	only partly for notational convenience. There are two other reasons. The first is that since the kinetic terms for the gauge bosons are normalized to $-1/4$, $\mathcal{K}_i$ must in the end have the value of $1$. This puts a relation between $f(0)$ and $g_i$. This is the same as in the Standard Model \cite[\S 17.1]{CCM07}. Secondly, the expression for $\mathcal{K}_i$ depends on the contents of the spectral triple. As \eqref{eq:bb2-gaugekinterms} shows, when the Hilbert space is extended with \rep{i}{j} and its opposite (both having $R = 1$), then \eqref{eq:expressionK1} changes to
\begin{align}
	\K_i &= \frac{f(0)}{3\pi^2}g_i^2n_i(2N_i + MN_j), & \K_j &= \frac{f(0)}{3\pi^2}g_j^2n_j(MN_i + 2N_j), & \K_{B} &= \frac{4f(0)}{3\pi^2}\frac{N_j}{N_i}M g_B^2.\label{eq:expressionK2}
\end{align}
Here $M$ denotes the number of generations that the fermion--sfermion pair comes in. In fact, the relation between the coupling constant(s) $g_i$ and the function $f$ should be evaluated only for the full spectral triple. In this case however, setting all three terms equal to one, implies the GUT-like relation 
		\bas
			n_i(2N_i + MN_j) g_i^2 = n_j(2N_j + MN_i) g_j^2 = 4\frac{N_j}{N_i}M g_B^2.
		\eas 
\end{rmk}

What remains, is to check whether there exist solutions for $C_{iij}$ and $C_{ijj}$ that satisfy the supersymmetry constraints \eqref{eq:bb2-resultCiij}. 
\begin{prop}\label{lem:bb2-nosol}
	Consider an almost-commutative geometry whose finite algebra is of the form $M_{N_i}(\com) \oplus M_{N_j}(\com)$. The particle content and action associated to this almost-commutative geometry are both supersymmetric off shell if and only if it consists of two disjoint building blocks \B{i,j} of the first type, for which $N_{i}, N_{j} > 1$.
\end{prop}
\begin{proof}
We will prove this by showing that the action of a single building block \B{ij} of the second type is not supersymmetric, falling back to Theorem \ref{thm:bb1} for a positive result. For the action of a \B{ij} of the second type to be supersymmetric requires the existence of parameters $C_{iij}$ and $C_{ijj}$ that ---after scaling according to \eqref{eq:bb2-scalingG}--- satisfy \eqref{eq:bb2-resultCiij} both directly and indirectly via $\P_{i,j}$ of the form \eqref{eq:solP}. To check whether they directly satisfy \eqref{eq:bb2-resultCiij} we note that the pre-factor $\n_{ij}^2$ for the kinetic term of the sfermion $\sfer_{ij}$ appearing in \eqref{eq:bb2-scalingG} itself is an expression in terms of $C_{iij}$ and $C_{ijj}$. We multiply the first relation of \eqref{eq:bb2-resultCiij} with its conjugate and multiply with $\n_{ij}$ on both sides to get
\bas
	C_{iij}^*C_{iij} = \frac{2}{\K_i}n_ig_i^2 \n_{ij}^2.
\eas
Inserting the expression \eqref{eq:exprN} for $\n_{ij}^2$, we obtain
\bas
	C_{iij}^*C_{iij} &= g_i^2n_i\frac{f(0)}{\pi^2}\frac{1}{\K_i}\Big[N_iC_{iij}^*C_{iij} + N_jC_{ijj}^*C_{ijj}\Big].
\eas
From \eqref{eq:bb2-scalingG} and \eqref{eq:bb2-resultCiij} we infer that $C_{ijj}^*C_{ijj} = (n_jg_j^2/n_ig_i^2)C_{iij}^*C_{iij}$, i.e.~we require:
\bas
	\K_i &= \frac{f(0)}{\pi^2}\Big[g_i^2n_iN_i + n_jg_j^2N_j\Big].
\eas
If we use the expressions \eqref{eq:expressionK2} for the pre-factors of the gauge bosons' kinetic terms to express the combinations $f(0)n_{i,j}g_{i,j}^2/\pi^2$ in terms of $N_{i,j}$ and $M$, the requirement for consistency reads
\bas
		1 &= \bigg(\frac{3N_i}{2N_i + MN_j} + \frac{3N_j}{MN_i + 2N_j}\bigg).
\eas
The only solutions to this equation are given by $M = 4$ and $N_i = N_j$. However, inserting the solution \eqref{eq:bb2-resultCiij} for $C_{iij}^*C_{iij}$ into the expression \eqref{eq:solP} for $\P_{i}, \P_j$ (necessary to write the action off shell) gives
\bas
	\P_{i}^2 &= 4\frac{f(0)}{\pi^2}N_{i}g_{i}^4\frac{n_{i}}{\K_i^2},&
	\P_{j}^2 &= 4\frac{f(0)}{\pi^2}N_{j}g_{j}^4\frac{n_{j}}{\K_j^2},
\eas
with an $\id_M$ where appropriate. We again use Remark \ref{rmk:bb2-rmk} to replace $f(0)g_i^2/(\pi^2\K_i)$ by an expression featuring $N_{i,j}$, $M$ and $n_{i,j}$. This yields
\bas
	\P_{i}^2 &= \frac{12N_i}{2N_i + MN_j}\frac{g_i^2}{\K_i} = 2\frac{g_i^2}{\K_i},&
	\P_{j}^2 &= \frac{12N_j}{2N_j + MN_i}\frac{g_j^2}{\K_j} = 2\frac{g_j^2}{\K_j}
\eas
for the values $M = 4$, $N_i = N_j$ that gave the correct fermion-sfermion-gaugino interactions. We thus have a contradiction with the demand on $\P_{i,j}^2$ from \eqref{eq:bb2-resultCiij}, necessary for supersymmetry.
\end{proof}

We shortly pay attention to a case that is of similar nature but lies outside the scope of the above Proposition. 

\begin{rmk}
	For $\A_F = \com \oplus \com$, a building block \B{ij} of the second type does not have a supersymmetric action either. In this case there are only $u(1)$ fields present in the theory and $G_i$, $G_j$ and $H$ are seen to coincide. It is possible to rewrite the four-scalar interaction of the spectral action off shell, but this set up also suffers from a similar problem as in Proposition \ref{lem:bb2-nosol}.
\end{rmk}

We can extend the result of Proposition \ref{lem:bb2-nosol} to components of the finite algebra that are defined over other fields than $\com$. For this, we first need the following lemma.

\begin{lem}\label{lem:bb2-otherfields}
	The inner fluctuations \eqref{eq:inner_flucts} of $\dirac$ caused by a component of the finite algebra that is defined over $\mathbb{R}$ or $\mathbb{H}$, are traceless.
\end{lem}
\begin{proof}
	The inner fluctuations are of the form
	\bas
		i \gamma^\mu A_\mu^{\mathbb{F}},\quad  A_\mu^{\mathbb{F}} = \sum_i a_i \partial_\mu (b_i),\qquad \text{with}\quad a_i, b_i \in C^{\infty}(M, M_{N}(\mathbb{F})),\quad \mathbb{F} = \mathbb{R}, \mathbb{H}.
	\eas
	This implies that $A_\mu^\mathbb{F}$ is itself an $M_{N}(\mathbb{F})$-valued function. For the inner fluctuations to be self-adjoint, $A_\mu^\mathbb{F}$ must be skew-Hermitian. In the case that $\mathbb{F} = \mathbb{R}$ this implies that all components on the diagonal vanish and consequently so does the trace. In the case that $\mathbb{F} = \mathbb{H}$, all elements on the diagonal must themselves be skew-Hermitian. Since all quaternions are of the form
	\bas
	\begin{pmatrix}
		\alpha & \beta\\
		- \bar \beta & \bar \alpha
	\end{pmatrix} \quad \alpha, \beta \in \com,
	\eas
	this means that the diagonal of $A_\mu^\mathbb{H}$ consists of purely imaginary numbers that vanish pairwise. Its trace is thus also $0$.
\end{proof}

Then we have 
\begin{theorem}
	Consider an almost-commutative geometry whose finite algebra is of the form $M_{N_i}(\mathbb{F}_i) \oplus M_{N_j}(\mathbb{F}_j)$ with $\mathbb{F}_i, \mathbb{F}_j = \mathbb{R}, \com, \mathbb{H}$. If the particle content and action associated to this almost-commutative geometry are both supersymmetric off shell, then it consists of two disjoint building blocks \B{i,j} of the first type, for which $N_{i}, N_{j} > 1$.
\end{theorem}
\begin{proof}
Not only do we have different possibilities for the fields $\mathbb{F}_{i,j}$ over which the components are defined, but we can also have various combinations for the values of the $R$-parity. We cover all possible cases one by one.\\
 
If $R = +1$ on the representations in the finite Hilbert space that describe the gauginos, then the gauginos and gauge bosons have the same $R$-parity and the particle content is not supersymmetric. \\

If $R = -1$ for these representations, and $R = +1$ on the off-diagonal representations, suppose at least one of the $\mathbb{F}_i$, $\mathbb{F}_j$ is equal to $\mathbb{R}$ or $\mathbb{H}$. Then using Lemma \ref{lem:bb2-otherfields} we see that after application of the unimodularity condition \eqref{eq:unimod_new} there is no $u(1)$-valued gauge field left. Lemma \ref{lem:gobinogo} then also causes the absence of a $u(1)$-auxiliary field that is needed to write the four-scalar action off shell as in Lemma \ref{lem:bb2-offshell}. If both $\mathbb{F}_i$ and $\mathbb{F}_j$ are equal to $\com$ we revert to Proposition \ref{lem:bb2-nosol} to show that there is no supersymmetric solution for $M$ and $N_{i,j}$ that satisfies the demands for $\Cw{i,j}$, $\Cw{j,i}$ and $\P_{i,j}$ from supersymmetry.\\

In the third case $R = -1$ on the off-diagonal representations in $\H_F$. If both $\mathbb{F}_{i,j}$ are equal to $\mathbb{R}$ or $\mathbb{H}$ then there is no $u(1)$ gauge field and thus the spectral action cannot be written off shell. If either $\mathbb{F}_i$ or $\mathbb{F}_j$ equals $\mathbb{R}$ or $\mathbb{H}$, then there is one $u(1)$-field, but the calculation for the action carries through as in Proposition \ref{lem:bb2-nosol} and there is no supersymmetric solution for $M$ and $N_{i,j}$. Finally, if both $\mathbb{F}_{i,j}$ are equal to $\com$, there are two $u(1)$-fields and the cross term as in Remark \ref{rmk:bb2-obstr} spoils supersymmetry. \\

Thus, all almost-commutative geometries for which $\A_F = M_{N_i}(\mathbb{F}_i) \oplus M_{N_j}(\mathbb{F}_j)$ and that have off-diagonal representations fail to be supersymmetric off shell.
\end{proof}

The set up described in this section has the same particle content as the supersymmetric version of a single ($R = +1$) particle--antiparticle pair and corresponds in that respect to a single chiral superfield in the superfield formalism \cite[4.3]{DGR04}. In constrast, its action is not fully supersymmetric. 
We stress however, that the scope of Proposition \ref{lem:bb2-nosol} is that of a \emph{single} building block of the second type. As was mentioned before, the expressions for many of the coefficients typically vary with the contents of the finite spectral triple and they should only be assessed for the full model. \\

Another interesting difference with the superfield formalism is that a building block of the second type really requires two building blocks of the first type, describing gauginos and gauge bosons. In the superfield formalism a theory consisting of only a chiral multiplet, not having gauge interactions, is in many textbooks the first model to be considered. This underlines that noncommutative geometry inherently describes gauge theories.\\

There are ways to extend almost-commutative geometries by introducing new types of building blocks ---giving new possibilities for supersymmetry--- or by combining ones that we have already defined. In the next section we will cover an example of the latter situation, in which there arise interactions between two or more building blocks of the second type.

\subsubsection{Interaction between building blocks of the second type}\label{sec:2bb2}

In the previous section we have fully exploited the options that a finite algebra with two components over the complex numbers gave us. If we want to extend our theory, the finite algebra \eqref{eq:alg} needs to have a third summand --- say $M_{N_k}(\com)$. A building block of the first type (cf.~Section \ref{sec:bb1}) can easily be added, but then we already stumble upon severe problems:
\begin{prop}\label{prop:2bb2-obstr}
	The action \eqref{eq:spectral_action_acg_flat} of an almost-commutative geometry whose finite algebra consists of three summands $M_{N_{i,j,k}}(\com)$ over $\com$ and whose finite Hilbert space features building blocks \Bc{ij}{\pm} and \Bc{ik}{\pm} is not supersymmetric.
\end{prop}
\begin{proof}
	The inner fluctuations of the canonical Dirac operator on \rep{i}{j} and \rep{i}{k} read:
	\bas
		& \dirac + g_iA_i - g_jA_j^o + \frac{g_{B_i}}{N_i}B_i - \frac{g_{B_j}}{N_j}B_j, &
		& \dirac + g_iA_i - g_kA_k^o + \frac{g_{B_i}}{N_i}B_i - \frac{g_{B_k}}{N_k}B_k,
	\eas	
	where $A_{i,j,k} = \gamma_\mu A^{\mu}_{i, j, k}$, with $A^{\mu}_{i,j, k}(x) \in su(N_{i,j, k})$ and similarly $B^{\mu}_{i,j, k}(x) \in u(1)$. The unimodularity condition will, in the case that the representation of at least one of the two building blocks has $R = +1$, leave two of the three independent $u(1)$ fields ---say--- $B_i$ and $B_j$. The kinetic terms of the gauge bosons on both representations will then feature a cross term \eqref{eq:u1cross} of different $u(1)$ field strengths, an obstruction for supersymmetry.
\end{proof}

To resolve this, we allow ---inspired by the NCSM--- for one or more copies of the quaternions $\mathbb{H}$ in the finite algebra. If we define a building block of the first type over such a component (with the finite Hilbert space $M_2(\com)$ as a bimodule of the complexification $M_1(\mathbb{H})^\com = M_2(\com)$ of the algebra, instead of $\mathbb{H}$ itself, cf.~\cite[\S 4.1]{Bhowmick2011}, \cite{CC08}), the self-adjoint inner fluctuations of the canonical Dirac operator are already seen to be in $su(2)$ (e.g.~traceless) prior to applying the unimodularity condition. On a representation \rep{i}{j} (from a building block \Bc{ij}{\pm} of the second type), of which one of the indices comes from a component $\mathbb{H}$, only one $u(1)$ field will act.

\emph{From here on, using three or more components in the algebra, we will always assume at most two to be of the form $M_N(\com)$ and all others to be equal to $\mathbb{H}$}.

The action of an almost commutative geometry whose finite spectral triple features two building blocks of the second type sharing one of their indices (i.e.~that are in the same row or column in a Krajewski diagram) contains extra four-scalar contributions. The specific form of these terms depends on the value of the grading and of the indices appearing. When the first indices of two building blocks are the same, and they have the same grading (e.g.~\Bc{ji}{+} and \Bc{jk}{+}, cf.~Figure \ref{fig:2bb2s-different}) the resulting extra interactions are given by
\ba
	S_{ij,jk}[\sfer_{ij}, \sfer_{jk}] &= \frac{f(0)}{\pi^2} N_j \int_M |C_{ijj}\sfer_{ij}C_{jjk}\sfer_{jk}|^2 \sqrt{g}\mathrm{d}^4x.\label{eq:2bb2s-same}
\ea 
In the other case (cf.~Figure \ref{fig:2bb2s-same}) it is given by
\ba
S_{ij,jk}[\sfer_{ij}, \sfer_{jk}] &=	\frac{f(0)}{\pi^2} \int_M |C_{ijj}\sfer_{ij}|^2|C_{jjk}\sfer_{jk}|^2\sqrt{g}\mathrm{d}^4x.\label{eq:2bb2s-different}
\ea 
The paths corresponding to these contributions are depicted in Figure \ref{fig:2bb2s}.

\begin{figure}[ht]
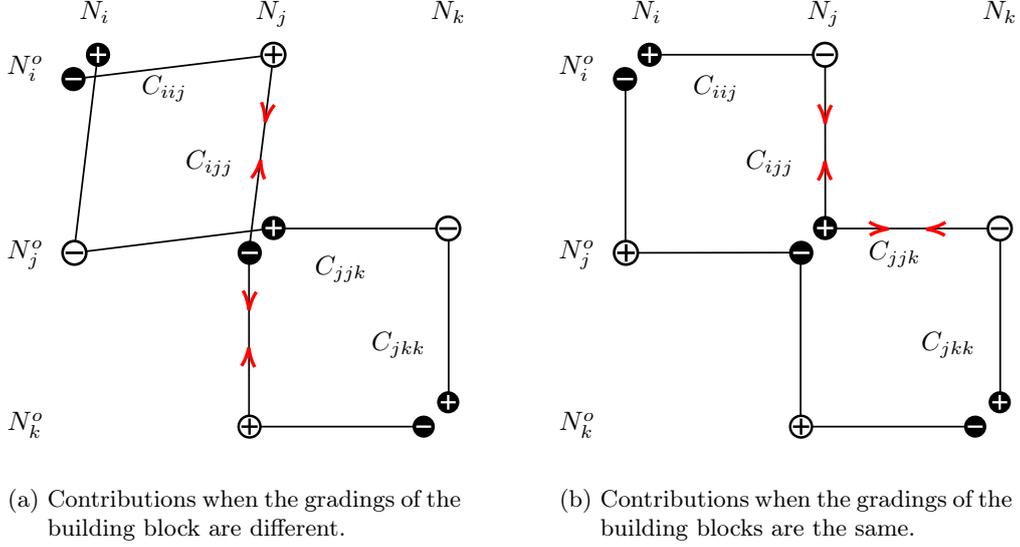

		\captionsetup{width=.9\textwidth}
	\centering
	\begin{subfigure}{.4\textwidth}
		\centering
		\def\svgwidth{\textwidth}
		\subimport{\the\svgpath}{2bb2s-different.pdf_tex}	
		\caption{Contributions when the gradings of the building block are different.}
		\label{fig:2bb2s-different}
	\end{subfigure}
	\hspace{30pt}
	\begin{subfigure}{.4\textwidth}
		\centering
		\def\svgwidth{\textwidth}
		\subimport{\the\svgpath}{2bb2s-same.pdf_tex}
		\caption{Contributions when the gradings of the building blocks are the same.}
		\label{fig:2bb2s-same}
	\end{subfigure}
\caption{In the case that there are two building blocks of the second type sharing one of their indices, there are extra interactions in the action.}
\label{fig:2bb2s}
\end{figure}

However, to write all four-scalar interactions from the spectral action off shell in terms of the auxiliary fields $G_{i,j,k}$, one requires interactions of the form of both \eqref{eq:2bb2s-same} and \eqref{eq:2bb2s-different} to be present. The reason for this is the following. Upon writing the four-scalar part of the action of the building blocks \B{ij} and \B{jk} in terms of the auxiliary fields as in Lemma \ref{lem:bb2-offshell}, we find for the terms with $G_j$ in particular:
\bas
	- \frac{1}{2n_j}\tr_{N_j} G_j^2 - \tr_{N_j} G_j\Big(\asfer_{ij}\P_{j,i}'\sfer_{ij}\Big) - \tr_{N_j} G_j\Big(\P_{j,k}'\sfer_{jk}\asfer_{jk}\Big).
\eas
On shell, the cross terms of this expression then give the additional four-scalar interaction
\ba\label{eq:2bb2-different-aux}
	&  n_j|\P_{j,i}'^{1/2}\sfer_{ij}\P_{j,k}'^{1/2}\sfer_{jk}|^2 - \frac{n_j}{N_j}|\P_{j,i}'^{1/2}\sfer_{ij}|^2|\P_{j,i}'^{1/2}\sfer_{jk}|^2.
\ea
When the scaled counterparts \eqref{eq:bb2-scalingG} of $\P_{j,i}'$ and $\P_{j,k}'$ satisfy the constraints \eqref{eq:bb2-resultCiij} for supersymmetry, this interaction reads
\bas
	&n_jg_j^2\Big( |\sfer_{ij}\sfer_{jk}|^2 - \frac{1}{N_j}|\sfer_{ij}|^2|\sfer_{jk}|^2\Big)
\eas
after scaling the fields. When having two or more building blocks of the second type that share one of their indices, we have either \eqref{eq:2bb2s-same} or \eqref{eq:2bb2s-different} in the spectral action, while we need \eqref{eq:2bb2-different-aux} for a supersymmetric action. To possibly restore supersymmetry we need additional interactions, such as those of the next section.

\subsection{Third building block: extra interactions}\label{sec:bb3}

\begin{figure}
	\begin{center}
		\def\svgwidth{.4\textwidth}
		\subimport{\the\svgpath}{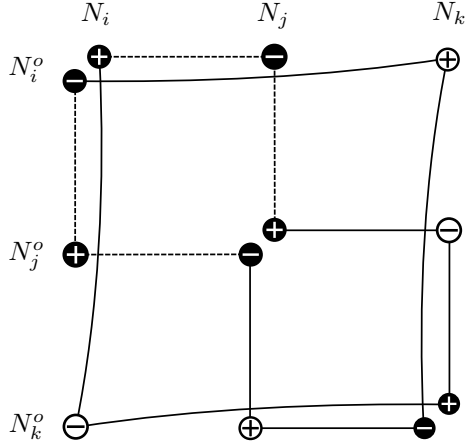}	
		\captionsetup{width=.9\textwidth}
	\caption{A situation in which all three building blocks of the second type are present whose two indices are either $i$, $j$ or $k$.} 
	\label{fig:3bb2}	
\end{center}
\end{figure}	

In a situation in which the finite algebra has three components and there are two adjacent building blocks of the second type, as depicted in Figure \ref{fig:2bb2s-same}, there is allowed a component 
		\begin{align}\label{eq:bb3-possible-comp}
		 	D_{ij}^{\phantom{ij}kj} : \rep{k}{j} \to \rep{i}{j}
		\end{align}
of the finite Dirac operator. We parametrize it with $\yuks{i}{k}$, that acts (non-trivially) on family space. Such a component satisfies the first order condition and its inner fluctuations 
	\bas
		 \sum_n a_n [\D{ij}{kj}, b_n] = \sum_n (a_i)_{n}\Big(\yuks{i}{k} (b_k)_{n} - (b_i)_{n}\yuks{i}{k}\Big)
	\eas
generate a scalar $\sfer_{ik} \in \rep{i}{k}$. Since there is no corresponding fermion $\fer{ik}$ present, a necessary condition for restoring supersymmetry is the existence of a building block \Bc{ik}{\pm} of the second type. The component \eqref{eq:bb3-possible-comp} then gives ---amongst others--- an extra fermionic contribution 
\begin{align*}
	\inpr{J_M\afer{ij}}{\gamma^5 \yuks{i}{k}\sfer_{ik}\afer{jk}}
\end{align*}
to the action. 
Using the transformations \eqref{eq:susytransforms4} and \eqref{eq:susytransforms5}, under which a building block of the second type is supersymmetric, 
we infer that this new term spoils supersymmetry. To overcome this, we need to add two extra components
\begin{align*}
\D{jk}{ik} : \rep{i}{k} & \to \rep{j}{k}, &\D{ij}{ik} : \rep{i}{k} &\to \rep{i}{j}
\end{align*}
to the finite Dirac operator, as well as their adjoints and the components that can be obtained by demanding that $[D_F,J_F] = 0$. We parametrize these two components with $\yuks{i}{j}$ and $\yuks{j}{k}$ respectively. They give extra contributions to the fermionic action that are of the form
\begin{align*}
	 \inpr{J_M\afer{jk}}{\gamma^5 \asfer_{ij}\yuks{i}{j}\fer{ik}} + \inpr{J_M\afer{ij}}{\gamma^5 \fer{ik}\asfer_{jk}\yuks{j}{k}}.
\end{align*}
Both components require the representation \rep{i}{k} to have an eigenvalue of $\gamma_F$ that is opposite to those of \rep{i}{j} and \rep{j}{k}. This is the situation as is depicted in Figure \ref{fig:3bb2}.

This brings us to the following definition:
\begin{defin}\label{def:bb3}
For an almost-commutative geometry in which \Bc{ij}{\pm}, \Bc{ik}{\mp} and \Bc{jk}{\pm} are present, a \emph{building block of the third type} \B{ijk} is the collection of all allowed components of the Dirac operator, mapping between the three representations \rep{i}{j}, \rep{i}{k} and \rep{k}{j} and their conjugates. Symbolically it is denoted by
\begin{align}\label{eq:bb3-comps}
	\BBB{ijk} = (0, \D{ij}{kj} + \D{jk}{ik} + \D{ij}{ik}) \in \H_F \oplus \End(\H_F).
\end{align}
\end{defin}
The Krajewski diagram corresponding to \B{ijk} is depicted in Figure \ref{fig:bb3}.

The parameters of \eqref{eq:bb3-comps} are chosen such that the sfermions $\sfer_{ij}$ and $\sfer_{jk}$ are generated by the inner fluctuations of \yuk{i}{j} and \yuk{j}{k} respectively, whereas $\sfer_{ik}$ is generated by $\yuks{i}{k}$. This is because $\sfer_{ik}$ crosses the particle/antiparticle-diagonal in the Krajewski diagram. Note that $i$, $j$, $k$ are labels, not matrix indices.

\begin{figure}[ht]
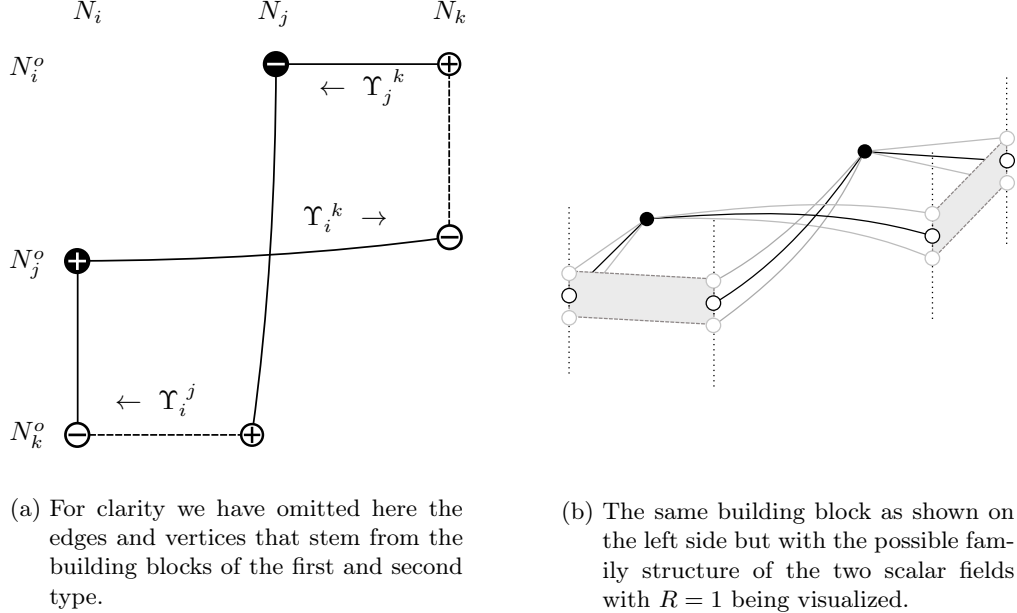

\centering
\begin{subfigure}{.4\textwidth}
\centering
		\def\svgwidth{\textwidth}
		\subimport{\the\svgpath}{bb3.pdf_tex}	
	\caption{For clarity we have omitted here the edges and vertices that stem from the building blocks of the first and second type.} 
	\label{fig:bb3-kraj}
\end{subfigure}
\hspace{30pt}
\begin{subfigure}{.4\textwidth}
		\centering
		\def\svgwidth{\textwidth}
		\subimport{\the\svgpath}{kraj_diagram_generations_bb3.pdf_tex}
		\caption{The same building block as shown on the left side but with the possible family structure of the two scalar fields with $R = 1$ being visualized.}
		\label{fig:bb3-gen}
\end{subfigure}
	\caption{A building block \B{ijk} of the third type in the language of Krajewski diagrams.}
	\label{fig:bb3}
\end{figure}	

There are several possible values of $R$ that the vertices and edges can have. Requiring a grading that yields $-1$ on each of the diagonal vertices, all possibilities for an explicit construction of $R \in \A_F \otimes \A_F^o$ are given by $R = - P\otimes P^o$, $P = (\pm 1, \pm 1, \pm 1) \in \A_F$ where each of the three signs can vary independently. This yields 8 possibilities, but each of them appears in fact twice. Of the effectively four remaining combinations, three have one off-diagonal vertex that has $R = -1$ and in the other combination all three off-diagonal vertices have $R = -1$. These four possibilities are depicted in Figure \ref{fig:bb3Pos}. We will typically work in the case of the first image of Figure \ref{fig:bb3Pos}, as is visualised in Figure \ref{fig:bb3-gen}, and will indicate where changes might occur when working in one of the other possibilities. If in this context the $R = 1$ representations in $\H_F$ come in $M$ copies (`generations'), all components of the finite Dirac operator are in general acting non-trivially on these $M$ copies, except $C_{iij}$ and $C_{ijj}$, since they parametrize components of the finite Dirac operator mapping between $R = -1$ representations. 

\begin{figure}[ht]
	\begin{center}
		\captionsetup{width=.9\textwidth}
		\includegraphics[width=.8\textwidth]{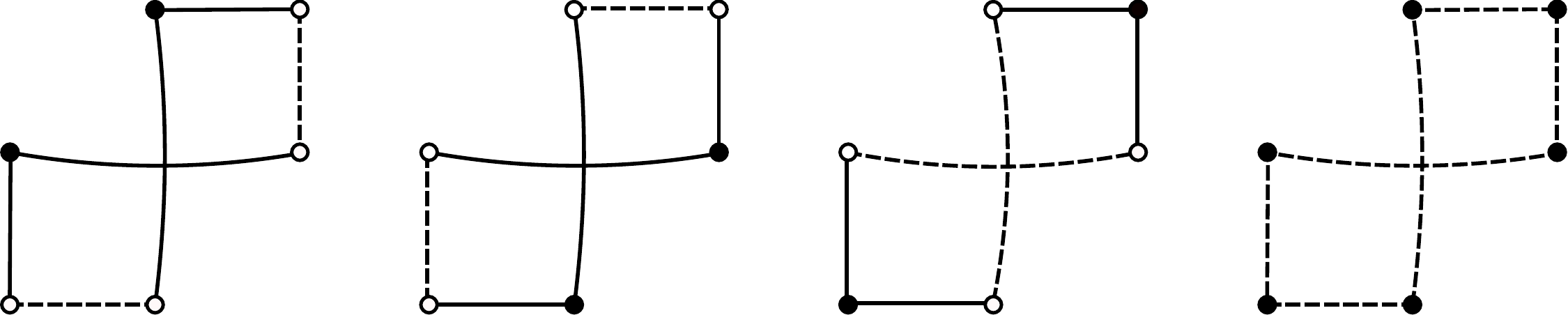}
	\caption{All possible combinations of values for the $R$-parity operator in a building block of the third type. Three of those possibilities have one representation on which $R = -1$, in the other possibility all three of them have $R = -1$. This last option essentially entails having no family structure.} 
	\label{fig:bb3Pos}	
\end{center}
\end{figure}

Note that in the action the expressions \eqref{eq:exprN} for the pre-factors $\n_{ij}^2$, $\n_{ik}^2$ and $\n_{jk}^2$ of the sfermion kinetic terms all get an extra contribution from the new edges of the Krajewski diagram of Figure~\ref{fig:bb3}. The first of these becomes
\ba
	\n_{ij}^2 & \to\frac{f(0)}{2\pi^2}(N_iC_{iij}^*C_{iij} + N_jC_{ijj}^*C_{ijj} + N_k\yuks{i}{j}\yuk{i}{j})\label{eq:kin_term_bb3}.
\ea
The other two can be obtained replacing $N_i$, $C_{iij}$, $C_{ijj}$ and $\yuk{i}{j}$ by their respective analogues.

The presence of a building block of the third type allows us to take a specific parametrization of the $C_{iij}$ in terms of $\yuk{i}{j}$. To this end, we introduce the shorthand notations
\ba
	q_i &:= \frac{f(0)}{\pi^2}g_i^2, & r_i &:= q_in_i, & \w{ij} &:= 1 - r_iN_i - r_jN_j,\label{eq:kintermnorm}
\ea
where we can infer from the normalization of the kinetic terms of the gauge bosons (i.e.~setting $\K_i = 1$) that $q_i$ must be rational. Then, similarly as in Proposition \ref{lem:bb2-nosol}, we write out $C_{iij}^*C_{iij}$, with $C_{iij}$ satisfying \eqref{eq:bb2-resultCiij} from supersymmetry, and insert the pre-factor \eqref{eq:kin_term_bb3} of the kinetic term. This reads 
\bas
C_{iij}^*C_{iij} &= r_i\big(N_i C_{iij}^*C_{iij} + N_jC_{ijj}^*C_{ijj} + N_k\yuks{i}{j}\yuk{i}{j}\big).
\eas
Using $r_i C_{ijj}^*C_{ijj} = r_jC_{iij}^*C_{iij}$, which can be directly obtained from the result \eqref{eq:bb2-resultCiij}, we obtain
\ba
	C_{iij}^*C_{iij} &= \frac{r_i}{\w{ij}}N_k\yuks{i}{j}\yuk{i}{j}\label{eq:bb3-expressionCs}
\ea
for the parametrization of $C_{iij}$ that satisfies \eqref{eq:bb2-resultCiij}. For future convenience we will take 
\ba\label{eq:bb3-expressionCs-sqrt}
	C_{iij} = \sgnc_{i,j} \sqrt{\frac{r_i}{\w{ij}}}\big(N_k \yuks{i}{j}\yuk{i}{j}\big)^{1/2},
\ea
with $\sgnc_{i,j} \in \{\pm\}$ the sign introduced in Theorem \ref{prop:bb2}. The other parameter, $C_{ijj}$, can be obtained by $r_i \to r_j$, $\sgnc_{i,j} \to \sgnc_{j,i}$. This yields for the pre-factor \eqref{eq:kin_term_bb3} of the kinetic term of $\sfer_{ij}$:
\ba
\n_{ij}^2&= \frac{f(0)}{2\pi^2}\bigg(N_i\frac{r_i}{\w{ij}} + N_j\frac{r_j}{\w{ij}} + 1\bigg) N_k\yuks{i}{j}\yuk{i}{j} = \frac{f(0)}{2\pi^2}\frac{1}{\w{ij}}N_k\yuks{i}{j}\yuk{i}{j}. \label{eq:bb3-exprN}
\ea
prior to the scaling \eqref{eq:bb3-scalingfields}. When $\sfer_{ij}$ has $R = 1$ and therefore does not carry a family structure (as in Figure \eqref{fig:bb3-gen}) then the trace over the representations where $\sfer_{ij}\asfer_{ij}$ and $\asfer_{ij}\sfer_{ij}$ are in, decouples from that over $M_{M}(\com)$. Consequently, the third term in \eqref{eq:kin_term_bb3} and the right hand sides of the solutions \eqref{eq:bb3-expressionCs-sqrt} and \eqref{eq:bb3-exprN} receive additional traces over family indices, i.e.~$N_k \yuks{i}{j}\yuk{i}{j} \to N_k \tr_M \yuks{i}{j}\yuk{i}{j}$. The strategy to write $C_{iij}$ in terms of parameters of building blocks of the third type works equally well when the kinetic term of $\sfer_{ij}$ gets contributions from multiple building blocks of the third type. In that case $N_k\yuks{i}{j}\yuk{i}{j}$ must be replaced by a sum of all such terms: $\sum_l N_l\yuks{i,l}{j}\yuk{i,l}{j}$ (see e.g.~Section \ref{sec:2bb3}), where the label $l$ is used to distinguish the building blocks \B{ijl} that all give a contribution to the kinetic term of $\sfer_{ij}$. 

There are several contributions to the action as a result of adding a building block of the third type. 
The action is given by
\begin{align}
	S_{ijk}[\zeta, \szeta] = S_{f, ijk}[\zeta, \szeta] + S_{b,ijk}[\szeta], \label{eq:bb3action}
\end{align}
with its fermionic part $S_{f, ijk}[\zeta, \szeta]$ reading
\begin{align}
	S_{f, ijk}[\zeta, \szeta] &= 
\inpr{J_M \afer{ij}}{\gamma^5\fer{ik}\asfer_{jk}\yuks{j}{k}} + 
\inpr{J_M \afer{ij}}{\gamma^5\yuks{i}{k}\sfer_{ik}\afer{jk}} + 
\inpr{J_M\afer{jk}}{\gamma^5\asfer_{ij}\yuks{i}{j}\fer{ik}} \nn\\ 
&\qquad + \inpr{J_M\afer{ik}}{\gamma^5\yuk{i}{j}\sfer_{ij}\fer{jk}} + 
\inpr{J_M\afer{ik}}{\gamma^5\fer{ij}\yuk{j}{k}\sfer_{jk}} + 
\inpr{J_M\fer{jk}}{\gamma^5\asfer_{ik}\yuk{i}{k}\fer{ij}}.\label{eq:bb3-action-ferm}
\end{align}
 The bosonic part of the action is given by:
\ba
	S_{b,ijk}[\szeta]	&=	\frac{f(0)}{2\pi^2} \Big[N_i|\yuk{j}{k}\sfer_{jk}\asfer_{jk}\yuks{j}{k}|^2 + N_j|\yuks{i}{k}\sfer_{ik}\asfer_{ik}\yuk{i}{k}|^2 +  N_k\tr_M(\yuks{i}{j}\yuk{i}{j})^2|\sfer_{ij}\asfer_{ij}|^2\Big] \nn\\
			&\qquad + S_{b,ij,jk}[\szeta] + S_{b,ik,jk}[\szeta] + S_{b,ij,ik}[\szeta], \label{eq:bb3-boson-action}
\ea
with
\ba
		S_{b,ij,jk}[\szeta]	&=	\frac{f(0)}{\pi^2} \Big[N_i |C_{iij}\sfer_{ij}\yuk{j}{k}\sfer_{jk}|^2 + N_k |\yuk{i}{j}\sfer_{ij}C_{jkk}\sfer_{jk}|^2 + |\sfer_{ij}|^2|\yuks{i}{j}\yuk{j}{k}\sfer_{jk}|^2\nn\\
		&\qquad	+ \Big(\tr \asfer_{jk}\yuks{j}{k}(\asfer_{ij}C_{iij}^*)^o(C_{iij}\sfer_{ij})^o \yuk{j}{k}\sfer_{jk} + \tr\asfer_{jk}C_{jjk}^*(\asfer_{ij}\yuks{i}{j})^o(\yuk{i}{j}\sfer_{ij})^oC_{jjk}\sfer_{jk}\nn\\
	&\qquad\qquad + \tr\asfer_{jk}\yuks{j}{k}(\asfer_{ij}C_{ijj}^*)^o(\yuk{i}{j}\sfer_{ij})^oC_{jjk}\sfer_{jk} + h.c.\Big)\Big],\label{eq:bb3-boson-action-part}
\ea
where the traces above are over $(\rep{k}{i}{})^{\oplus M}$. The fact that in this context $\sfer_{ij}$ has $R = 1$ makes it possible to separate the trace over the family-index in the last term of the first line of \eqref{eq:bb3-boson-action}. A more detailed derivation of the four-scalar action that corresponds to a building block of the third type, including the expressions for $S_{b,ik,jk}[\szeta]$ and $S_{b,ij,ik}[\szeta]$, is given in Appendix \ref{sec:bb3-calc-action}.\\ 

The expression \eqref{eq:bb3-boson-action} contains interactions that in form we either have seen earlier (cf.~\eqref{eq:exprM1}, \eqref{eq:2bb2s-same}) or that we needed but were lacking in a set up consisting only of building blocks of the second type (cf.~\eqref{eq:2bb2s-different}, see also the discussion in Section \ref{sec:2bb2}). In addition, it features terms that we need in order to have a supersymmetric action. \\

We can deduce from the transformations \eqref{eq:susytransforms4} that, for the expression \eqref{eq:bb3-action-ferm} (i.e.~the fermionic action that we have) to be part of a supersymmetric action, the bosonic action must involve terms with the auxiliary fields $F_{ij}$, $F_{ik}$ and $F_{jk}$ (that are available to us from the respective building blocks of the second type), coupled to two scalar fields. We will therefore formulate the most general action featuring these auxiliary fields and constrain its coefficients by demanding it to be supersymmetric in combination with \eqref{eq:bb3-action-ferm}. Subsequently, we will check if and when the spectral action \eqref{eq:bb3-boson-action} (after subtracting the terms that are needed for \eqref{eq:2bb2s-different}) is of the correct form to be written off shell in such a general form. This will be done for the general case in Section \ref{sec:4s-aux}.\\

The most general Lagrangian featuring the auxiliary fields $F_{ij}$, $F_{ik}$, $F_{jk}$ that can yield four-scalar terms is
\begin{align}
	S_{b, ijk, \mathrm{off}}[F_{ij}, F_{ik}, F_{jk}, \szeta] &= \int_M \mathcal{L}_{b, ijk, \mathrm{off}}(F_{ij}, F_{ik}, F_{jk}, \szeta)\sqrt{g}\mathrm{d}^4x, \label{eq:bb3-auxfields}
\end{align}
with
\bas
	\mathcal{L}_{b, ijk, \mathrm{off}}(F_{ij}, F_{ik}, F_{jk}, \szeta) &= - \tr F_{ij}^*F_{ij} + \big(\tr F_{ij}^*\beta_{ij,k}\sfer_{ik}\asfer_{jk} + h.c.\big) \nn\\
	&\qquad\qquad - \tr F_{ik}^*F_{ik} + \big(\tr F_{ik}^*\beta_{ik,j}^*\sfer_{ij}\sfer_{jk} + h.c.\big)\nn\\
			&\qquad\qquad\qquad - \tr F_{jk}^*F_{jk} + \big(\tr F_{jk}^*\beta_{jk,i}\asfer_{ij}\sfer_{ik} + h.c.\big)\nn.
\eas
Here $\beta_{ij,k}$, $\beta_{ik,j}$ and $\beta_{jk,i}$ are matrices acting on the generations and consequently the traces are performed over $\srep{j}^{\oplus M}$ (the first two terms) and $\srep{k}^{\oplus M}$ (the last four terms) respectively. Using the Euler-Lagrange equations the on shell counterpart of \eqref{eq:bb3-auxfields} is seen to be
	\begin{align*}
		S_{b, ijk, \mathrm{on}}[\szeta] = \int_M \sqrt{g}\mathrm{d}^4x\Big(|\beta_{ij,k}\sfer_{ik}\asfer_{jk}|^2 + |\beta_{ik,j}^*\sfer_{ij}\sfer_{jk}|^2 + |\beta_{jk,i}\asfer_{ij}\sfer_{ik}|^2\Big)
	\end{align*} 
	cf.~the second and third terms of \eqref{eq:bb3-boson-action}. We have the following result:


\begin{theorem}\label{thm:bb3}
	The action consisting of the sum of \eqref{eq:bb3-action-ferm} and \eqref{eq:bb3-auxfields} is supersymmetric under the transformations \eqref{eq:susytransforms4} and \eqref{eq:susytransforms5} if and only if the parameters of the finite Dirac operator are related via
\begin{align}
\yuk{j}{k}C_{jkk}^{-1} &= - (C_{ikk}^*)^{-1}\yuk{i}{k}, &
(C_{iik}^*)^{-1}\yuk{i}{k} &= - \yuk{i}{j}C_{iij}^{-1}, &
\yuk{i}{j}C_{ijj}^{-1} &= - \yuk{j}{k}C_{jjk}^{-1}.\label{eq:improvedUpsilons}
\end{align}
and
\ba\label{eq:bb3-susy-demand2}
\bps_{ij,k}\bp_{ij,k} &= \yukp{j}{k}\yukps{j}{k}= \yukp{i}{k}\yukps{i}{k}, &
\bps_{ik,j}\bp_{ik,j} &= \yukp{i}{j}\yukps{i}{j}= \yukp{j}{k}\yukps{j}{k}, \nn\\
\bps_{jk,i}\bp_{jk,i} &= \yukp{i}{k}\yukps{i}{k}= \yukp{i}{j}\yukps{i}{j}, &
\ea
where 
\bas
 \bp_{ij,k} &:= \n_{jk}^{-1}\beta_{ij,k}\n_{ik}^{-1},& \bp_{ik,j} &:= \n_{jk}^{-1} \beta_{ik,j} \n_{ij}^{-1},& \bp_{jk,i} &:= \n_{ij}^{-1}\beta_{jk,i}\n_{ik}^{-1}
\eas
and
\ba\label{eq:def-yukp}
	\yukp{i}{j} &:= \yuk{i}{j}\n_{ij}^{-1},& 	
	\yukp{i}{k} &:= \n_{ik}^{-1}\yuk{i}{k},& 	
	\yukp{j}{k} &:= \yuk{j}{k}\n_{jk}^{-1}, 	
\ea
denote the scaled versions of the $\beta_{ij,k}$'s and the $\yuk{i}{j}$'s respectively.
\end{theorem}
\begin{proof}
See Appendix \ref{sec:bb3-proof}.
\end{proof}

For future use we rewrite \eqref{eq:improvedUpsilons} using the parametrization \eqref{eq:bb3-expressionCs-sqrt} for the $C_{iij}$, giving 
\ba
	\sgnc_{i,j} \sqrt{\w{ij}}\,\yukw{i}{j} &= - \sgnc_{i,k}\sqrt{\w{ik}}\,\yukw{i}{k}, &
	\sgnc_{j,i} \sqrt{\w{ij}}\,\yukw{i}{j} &= - \sgnc_{j,k}\sqrt{\w{jk}}\,\yukw{j}{k}, \nn\\
	\sgnc_{k,i} \sqrt{\w{ik}}\,\yukw{i}{k} &= - \sgnc_{k,j}\sqrt{\w{jk}}\,\yukw{j}{k}, &&
\label{eq:improvedUpsilons1}
\ea
where we have written
\ba\label{eq:def-yukw}
	\yukw{i}{j} &:= \yuk{i}{j}(N_k\tr\yuks{i}{j}\yuk{i}{j})^{-1/2}, & \yukw{i}{k} &:= (N_j\yuk{i}{k}\yuks{i}{k})^{-1/2}\yuk{i}{k},& \yukw{j}{k} &:= \yuk{j}{k}(N_i\yuks{j}{k}\yuk{j}{k})^{-1/2}.
\ea
There is a trace over the generations in the first term because the corresponding sfermion $\sfer_{ij}$ has $R = 1$ and consequently no family-index. Using these demands on the parameters, the (spectral) action from a building block of the third type becomes much more succinct. First of all it allows us to reduce all three parameters of the finite Dirac operator of Definition \ref{def:bb3} to only one, e.g.~$\yuk{}{} \equiv \yuk{i}{j}$. Second, upon using \eqref{eq:improvedUpsilons} the second and third lines of \eqref{eq:bb3-boson-action-part} are seen to cancel.\footnote{More generally, this also happens for the other combinations: the four-scalar interactions of \eqref{eq:bb3-action-2} are seen to cancel those of \eqref{eq:bb3-action-5}}
If the demands \eqref{eq:improvedUpsilons} and \eqref{eq:bb3-susy-demand2} are met, the on shell action \eqref{eq:bb3action} that arises from a building block \B{ijk} of the third type reads
\ba
S_{ijk}[\zeta, \szeta, \mathbb{A}]
 &= g_m\sqrt{\frac{2\w{1}}{q_m}}\Big[\inpr{J_M\afer{2}}{\gamma^5\yukw{}{}\sfer_{1}\fer{3}} + \kappa_{j}
\inpr{J_M\afer{2}}{\gamma^5\fer{1}\yukw{}{}\sfer_{3}} + \kappa_{i}
\inpr{J_M\fer{3}}{\gamma^5\asfer_{2}\yukw{}{}\fer{1}} + h.c.\Big]\nn\\
	&\qquad
		+ g_m^2\frac{4\w{1}}{q_m}\Big[(1 - \w{2}) |\yukw{}{}\sfer_{1}\sfer_{3}|^2 
		+ (1 - \w{1}) |\yukw{}{}\sfer_{3}\asfer_{2}|^2 
		+ (1 - \w{3}) |\yukw{}{}\asfer_{2}\sfer_{1}|^2\Big].\label{eq:bb3-action-final}
\ea
Here we used the shorthand notations $ij \to 1$, $ik \to 2$, $jk \to 3$ and $\kappa_{j} = \sgnc_{j,i}\sgnc_{j,k}$, $\kappa_{i} = \sgnc_{i,j}\sgnc_{i,k}$ to avoid notational clutter as much as possible and where we have written everything in terms of $\yukw{}{} \equiv \yukw{i}{j}$ (as defined above), the parameter that corresponds to the sfermion having $R = 1$ (and consequently also multiplicity $1$). The index $m$ in $g_m$ and $q_m$ can take any of the values that appear in the model, e.g.~$i$, $j$ or $k$. As with a building block of the second type there is a sign ambiguity that stems from those of the $C_{iij}$. In addition, the terms that are not listed here but are in \eqref{eq:bb3action} give contributions to terms that already appeared in the action from building blocks of the second type. See Section \ref{sec:4s-aux} for details on this.\\

For notational convenience we have used two different notations for scaled variables: $\yukw{i}{j}$ from \eqref{eq:def-yukw} and $\yukp{i}{j}$ from \eqref{eq:def-yukp}. Using the expression \eqref{eq:bb3-exprN} for $\n_{ij}$ in terms of $\yuk{i}{j}$ these are related via
\ba\label{eq:yukwyukp}
	\yukp{i}{k} \equiv \n_{ik}^{-1}\yuk{i}{k} = \sqrt{\frac{2\pi^2}{f(0)}\w{ik}}\big(N_j \yuk{i}{k}\yuks{i}{k}\big)^{-1/2}\yuk{i}{k} \equiv g_l\sqrt{\frac{2\w{ik}}{q_l}}\yukw{i}{k},
\ea
assuming that $\sfer_{ik}$ has $R = -1$. The other two scaled variables give analogous expressions but the order of $\yuk{}{}$ and $\yuks{}{}$ is reversed and the sfermion with $R = 1$ gets an additional trace over family indices.

\begin{rmk}\label{rmk:bb3-relativesigns}
Note that we can use this result to say something about the signs of the $C_{iij}$ appearing in a building block of the third type. We first combine all three equations of \eqref{eq:improvedUpsilons} into one,
\bas
	\yuk{j}{k} = (-1)^3 (C_{iik}C_{ikk}^{-1})^{*} \yuk{j}{k}(C_{jjk}^{-1}C_{jkk})(C_{ijj}C_{iij}^{-1}),
\eas
when it is $C_{iij}$ and $C_{ijj}$ that do not have a family structure. All these parameters are only determined up to a sign. We will write 
\bas
	C_{iij}C_{ijj}^{-1} = s_{ij} \sqrt{\frac{n_i\K_j}{n_j\K_i}} \frac{g_i}{g_j},\qquad\text{with } s_{ij} := \sgnc_{i,j}\sgnc_{j,i} = \pm 1,
\eas
cf.~\eqref{eq:bb2-resultCiij}, etc.~which gives  
$
	\yuk{j}{k} = - s_{ij}s_{jk}s_{ki} \yuk{j}{k}
$
 for the relation above. So for consistency either one, or all three combinations of $C_{iij}$ and $C_{ijj}$ associated to a building block \B{ij} that is part of a \B{ijk} must be of opposite sign.
\end{rmk}

\begin{rmk}\label{rmk:bb3-R=1}
	If instead of $\sfer_{ij}$ it is $\sfer_{ik}$ or $\sfer_{jk}$ that has $R = 1$ (see Figure \ref{fig:bb3Pos}) the demand on the parameters $\yuk{i}{j}$, $\yuk{i}{k}$ and $\yuk{j}{k}$ is a slightly modified version of \eqref{eq:improvedUpsilons}: 
\ba
(\yuk{j}{k}C_{jkk}^{-1})^t &= - (C_{ikk}^*)^{-1}\yuk{i}{k}, &
(C_{iik}^*)^{-1}\yuk{i}{k} &= - \yuk{i}{j}C_{iij}^{-1}, &
\yuk{i}{j}C_{ijj}^{-1} &= - (\yuk{j}{k}C_{jjk}^{-1})^t,\label{eq:improvedUpsilons2}
\ea
where $A^t$ denotes the transpose of the matrix $A$. This result can be verified by considering Lemma \ref{lem:bb3-lem1} for these cases.
\end{rmk}

By introducing a building block of the third type we generated the interactions that we lacked in a situation with multiple building blocks of the second type. The wish for supersymmetry thus forces us to extend any model given by Figure \ref{fig:3bb2} with a building block of the third type.\\

If we again seek the analogy with the superfield formalism, then a building block of the third type is a Euclidean analogy of an action on a Minkowskian background that comes from a superpotential term
\ba\label{eq:bb3-superpot}
	\int \Big(\mathcal{W}(\{\Phi_m\})\Big|_F + h.c.\Big)\mathrm{d}^4 x,\quad \text{with}\quad \mathcal{W}(\{\Phi_m\}) = f_{mnp} \Phi_m\Phi_n \Phi_p,
\ea
where $\Phi_{m,n,p}$ are chiral superfields, $f_{mnp}$ is symmetric in its indices \cite[\S 5.1]{DGR04} and with $|_{F}$ we mean multiplying by $\bar\theta\bar\theta$ and integrating over superspace $\int \mathrm{d}^2\theta\mathrm{d}^2\bar\theta$. To specify this statement, we write $\Phi_{ij} = \phi_{ij} + \sqrt{2}\theta \psi_{ij} + \theta\theta F_{ij}$ for a chiral superfield. Similarly, we introduce $\Phi_{jk}$ and $\Phi_{ki}$. We then have that
\bas
	\int_M \Big[\Phi_{ij} \Phi_{jk}\Phi_{ki}  \Big]_F + h.c. 
&= \int_M  -  \psi_{ij}\phi_{jk}\psi_{ki} -  \psi_{ij}\psi_{jk}\phi_{ki}-  \phi_{ij}\psi_{jk}\psi_{ki}\nn\\
	&\qquad\qquad + F_{ij}\phi_{jk}\phi_{ki} + \phi_{ij}\phi_{jk}F_{ki} + \phi_{ij}F_{jk}\phi_{ki} + h.c. 
\eas
This gives on shell the following contribution\footnote{On a Minkowskian background the product of a superfield and its conjugate appears in the action as $F_{ij}^*F_{ij}$, i.e.~with pre-factor $+1$ \cite[\S 4.3]{DGR04}, in contrast to \eqref{eq:bb2-auxfields}.}:
\bas
& - \int_M \Big(\psi_{ij}\phi_{jk}\psi_{ki} +  \psi_{ij}\psi_{jk}\phi_{ki}+  \phi_{ij}\psi_{jk}\psi_{ki} + \frac{1}{2}|\phi_{jk}\phi_{ki}|^2 + \frac{1}{2}|\phi_{ij}\phi_{jk}|^2 + \frac{1}{2}|\phi_{ki}\phi_{ij}|^2 + h.c.\Big), 
\eas
to be compared with \eqref{eq:bb3-action-final}. In a set up similar to that of Figure \ref{fig:3bb2}, but with the chirality of one or two of the building blocks \B{ij}, \B{jk} and \B{ik} being flipped, not all three components of $D_F$ such as in Definition \ref{def:bb3} can still be defined, see Figure \ref{fig:bb3_opposite_grading}. Interestingly, one can check that in such a case the resulting action corresponds to a superpotential that is not holomorphic, but e.g.~of the form $\Phi_{ij} \Phi_{ik}\Phi_{jk}^\dagger$ instead. To see this, we calculate the action \eqref{eq:bb3-superpot} in this case, giving 
\bas
	\int_M \Big[\Phi_{ij} \Phi_{jk}^\dagger  \Phi_{ki}\Big]_F + h.c. 
&=  \int_M -  \psi_{ij}\phi_{jk}^*\psi_{ki} + F_{ij}\phi_{jk}^*\phi_{ki} + \phi_{ij}\phi_{jk}^*F_{ki} + h.c., 
\eas
which on shell equals
\bas
- \int_M \psi_{ij} \phi_{jk}^* \psi_{ki} + \frac{1}{2}|\phi_{jk}^*\phi_{ki}|^2 + \frac{1}{2}|\phi_{ij}\phi_{jk}^*|^2 + h.c.
\eas
This is indeed analogous to the interactions that the spectral triple depicted in Figure \ref{fig:bb3_opposite_grading} (still) gives rise to.

\begin{figure}[ht]
\centering
	\def\svgwidth{.4\textwidth}
	\subimport{\the\svgpath}{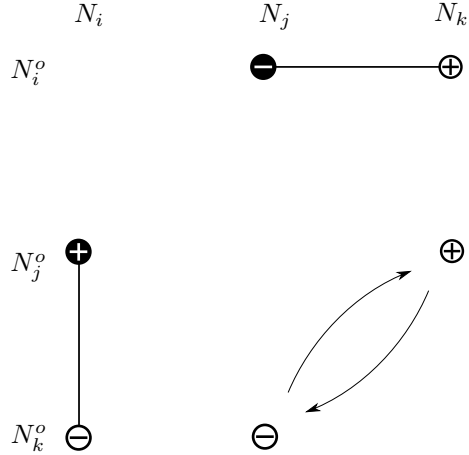}	
	\captionsetup{width=.6\textwidth}
	\caption{A set up similar to that of Figure \ref{fig:bb3}, but with the values of the grading reversed for $\rep{j}{k}$ and its opposite. Consequently, only one of the three components that characterize a building block of the first type can now be defined.}
	\label{fig:bb3_opposite_grading}
\end{figure}


\subsubsection{Interaction between building blocks of the third type}\label{sec:2bb3}
Suppose we have two building blocks \B{ijk} and \B{ijl} of the third type that share two of their indices, as is depicted in Figure \ref{fig:2bb3s}. This situation gives rise to the following extra terms in the action:
\ba
	&\frac{f(0)}{\pi^2}\Big[N_j|\asfer_{jk}C_{jjk}^*C_{jjl}\sfer_{jl}|^2 + N_i|\asfer_{jk}\yuks{j}{k}\yuk{j}{l}\sfer_{jl}|^2 + |\yuk{j}{k}\sfer_{jk}|^2|\yuks{i}{l}\sfer_{il}|^2\Big] + (i \leftrightarrow j)\nn\\
	&\quad + \frac{f(0)}{\pi^2}\Big(N_i  \tr C_{iik}\sfer_{ik}\asfer_{jk}\yuks{j}{k}\yuk{j}{l}\sfer_{jl}\asfer_{il}C_{iil}^* 
	 +  N_j\tr \yuks{i}{k}\sfer_{ik}\asfer_{jk}C_{jjk}^*C_{jjl}\sfer_{jl}\asfer_{il}\yuk{i}{l} + h.c.\Big),\label{eq:2bb3-action}
\ea
where with `$(i\leftrightarrow j)$' we mean the expression preceding it, but everywhere with $i$ and $j$ interchanged. The first line of \eqref{eq:2bb3-action} corresponds to paths within the two building blocks \B{ijk} and \B{ijl} (such as the ones depicted in Figure \ref{fig:2bb3s-inner}) and the second line corresponds to paths of which two of the edges come from the building blocks of the second type that were needed in order to define the building blocks of the third type (Figure \ref{fig:2bb3s-outer}).

If we scale the fields appearing in this expression according to \eqref{eq:bb3-scalingfields} and use the identity \eqref{eq:improvedUpsilons} for the parameters of a building block of the third type, we can write \eqref{eq:2bb3-action} more compactly as
\ba
&4n_jr_jN_jg_j^2|\asfer_{jk}\sfer_{jl}|^2 + 4\frac{g_m^2}{q_m}\w{ij}^2N_i|\asfer_{jk}\yukws{k}{}\yukw{l}{}\sfer_{jl}|^2 + 4\frac{g_m^2}{q_m}\w{ij}^2|\yukw{k}{}\sfer_{jk}|^2|\yukws{l}{}\sfer_{il}|^2 + (i \leftrightarrow j, \yuk{}{} \leftrightarrow \yuks{}{})\nn\\
	&\qquad + \kappa_{k}\kappa_{l}\, 4\frac{g_m^2}{q_m}(1 - \w{ij})\w{ij}\tr \yukw{l}{}\yukws{k}{}\sfer_{ik}\asfer_{jk}\sfer_{jl}\asfer_{il} + h.c.,\label{eq:2bb3-action-scaled}
\ea
where $\kappa_{k} = \sgnc_{k,i}\sgnc_{k,j}$, $\kappa_{l} = \sgnc_{l,i}\sgnc_{l,j} \in \{\pm 1\}$, $\yukw{k}{} \equiv \yukw{i,k}{j}$ of \B{ijk} and $\yukw{l}{} \equiv \yukw{i, l}{j}$ of \B{ijl}, as defined in \eqref{eq:def-yukw} but with contributions from two building blocks of the third type: 
\begin{subequations}
\ba
	\yukw{i,k}{j} &= \yuk{i,k}{j} (N_k \tr \yuks{i,k}{j}\yuk{i,k}{j} + N_l\tr\yuks{i,l}{j}\yuk{i,l}{j})^{-1/2}, \\
	\yukw{i,l}{j} &= \yuk{i,l}{j} (N_k \tr \yuks{i,k}{j}\yuk{i,k}{j} + N_l\tr\yuks{i,l}{j}\yuk{i,l}{j})^{-1/2}.\label{eq:def-yukw-mult}
\ea
\end{subequations}
This expression can be generalized to any number of building blocks of the third type. In addition, we have assumed that $s_{ik}s_{il} = s_{jk}s_{jl}$ for the products of the relative signs between the parameters $C_{iik}$ and $C_{ikk}$ etc.~(cf.~Remark \ref{rmk:bb3-relativesigns}).

\begin{figure}
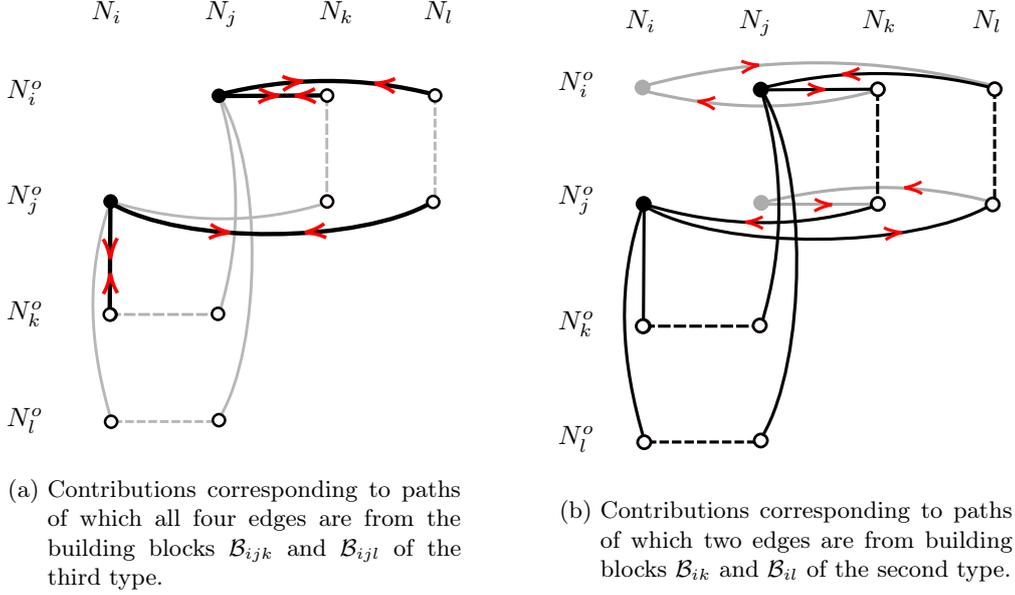

	\centering
	\begin{subfigure}{.4\textwidth}
		\centering
		\def\svgwidth{\textwidth}
		\subimport{\the\svgpath}{2bb3s-inner.pdf_tex}
		\caption{Contributions corresponding to paths of which all four edges are from the building blocks \B{ijk} and \B{ijl} of the third type.}
		\label{fig:2bb3s-inner}
	\end{subfigure}
	\hspace{30pt}
	\begin{subfigure}{.4\textwidth}
		\centering
		\def\svgwidth{\textwidth}
		\subimport{\the\svgpath}{2bb3s-outer.pdf_tex}	
		\caption{Contributions corresponding to paths of which two edges are from building blocks \B{ik} and \B{il} of the second type.}
		\label{fig:2bb3s-outer}
	\end{subfigure}
\caption{In the case that there are two building blocks of the third type sharing two of their indices, there are extra four-scalar contributions to the action. They are given by \eqref{eq:2bb3-action}.}
\label{fig:2bb3s}
\end{figure}

These new interactions must be accounted for by the auxiliary fields. The first and second terms are of the form \eqref{eq:2bb2s-same} and should therefore be covered by the auxiliary fields $G_{i,j}$. The third term is of the form \eqref{eq:2bb2s-different} and should consequently be described by the combination of $G_{i,j}$ and the $u(1)$-field $H$. The second line of \eqref{eq:2bb3-action} should be rewritten in terms of the auxiliary field $F_{ij}$. This can indeed be achieved via the off shell Lagrangian
%
%
\bas
-\tr F_{ij}^*F_{ij}	+ \big(\tr F_{ij}^*( \beta_{ij,k} \sfer_{ik}\asfer_{jk} + \beta_{ij,l}\sfer_{il}\asfer_{jl}) + h.c.\big),
\eas
which on shell gives the following cross terms:
\ba
 \tr \beta_{ij,l}^*\beta_{ij,k}\sfer_{ik}\asfer_{jk}\sfer_{jl}\asfer_{il} + h.c. \label{eq:2bb3-aux}
\ea
In form, this indeed corresponds to the second line of \eqref{eq:2bb3-action-scaled}. In Section \ref{sec:4s-aux} a more detailed version of this argument is presented.

Furthermore, it can be that there are four different building blocks of the third type that all share one particular index ---say \B{ikl}, \B{ikm}, \B{jkl} and \B{jkm}, sharing index $k$--- then there arises one extra interaction, that is of the form
\bas
	& N_k\frac{f(0)}{\pi^2}\big[\tr\yuks{i}{m}\sfer_{im}\asfer_{jm}\yuk{j}{m}\yuks{j}{l}\sfer_{jl}\asfer_{il}\yuk{i}{l} + h.c.\big].\nn
\eas
Scaling the fields and rewriting the parameters using \eqref{eq:improvedUpsilons1} gives
\ba
	&4\frac{g_n^2}{q_n}\w{ik}\w{jk}
 \tr \tilde \Upsilon_{l}\tilde \Upsilon_{m}^*\sfer_{im}\asfer_{jm}(\tilde\Upsilon_{l}'\tilde\Upsilon_{m}'^*)^*\sfer_{jl}\asfer_{il} + h.c., \label{eq:4bb3s}
\ea
where $g_n$ can equal any of the coupling constants that appear in the theory and we have written 
\bas 
	\tilde \Upsilon_{m} &\equiv \yuk{i,m}{k}(N_m\yuks{i,m}{k}\yuk{i,m}{k} + N_l\yuks{i,l}{k}\yuk{i,l}{k})^{-1/2},\nn\\
 \tilde \Upsilon_{m}' &\equiv \yuk{j,m}{k}(N_m\yuks{j,m}{k}\yuk{j,m}{k} + N_l\yuks{j,l}{k}\yuk{j,l}{k})^{-1/2},
\eas
and the same for $m \leftrightarrow l$. The path to which such an interaction corresponds, is given in Figure \ref{fig:4bb3s}. One can check that this interaction can only be described off shell by invoking either one or both of the auxiliary fields $F_{ij}$ and $F_{lm}$. This means that in order to have a chance at supersymmetry, the finite spectral triple that corresponds to the Krajewski diagram of Figure \ref{fig:4bb3s} requires in addition at least \B{ij} or \B{lm}.

\begin{figure}
	\begin{center}
		\def\svgwidth{.5\textwidth}
		\subimport{\the\svgpath}{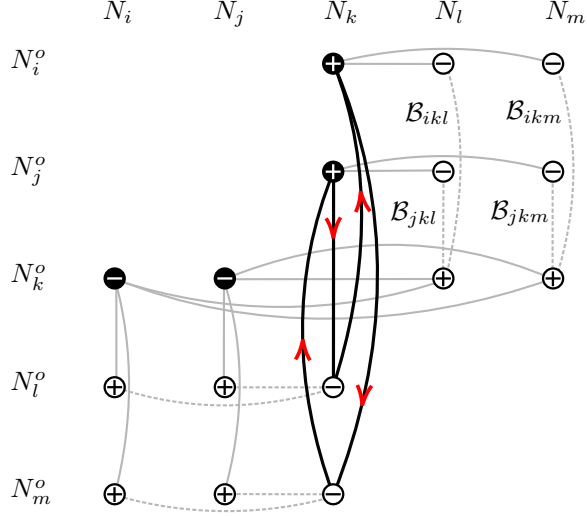}	
		\captionsetup{width=.8\textwidth}
	\caption{When four building blocks of the third kind share one common index (in this case $k$) and each pair of building blocks shares one of its two remaining indices ($i$, $j$, $l$ or $m$) with one other building block, there is an additional path that contributes to the trace of $D_F^4$ (including its inner fluctuations). The interaction is given by \eqref{eq:4bb3s}.}
	\label{fig:4bb3s}	
\end{center}
\end{figure}	

\subsection{Higher degree building blocks?}

The first three building blocks that gave supersymmetric actions are characterized by one, two and three indices respectively. One might wonder whether there are building blocks of higher order, carrying four or more indices. \\

Each of the elements of a finite spectral triple is characterized by one (components of the algebra, adjoint representations in the Hilbert space), two (non-adjoint representations in the Hilbert space) or three (components of the finite Dirac operator that satisfy the order-one condition) indices. For each of these elements corresponding building blocks have been identified. Any object that carries four or more different indices (e.g.~two or more off-diagonal representations, multiple components of a finite Dirac operator) must therefore be part of more than one building block of the first, second or third type. These blocks are, so to say, the irreducible ones. \\

This does not imply that there are no other building blocks left to be identified. However, as we will see in the next section, they are characterized by less than four indices.

\subsection{Mass terms}

There is a possibility that we have not covered yet. The finite Hilbert space can contain two or more copies of one particular representation. This can happen in two slightly different ways. The first is when there is a building block $\B{11'}$ of the second type, on which the same component $\com$ of the algebra acts both on the left and on the right in the same way. For the second way it is required that there are two copies of a particular building block $\B{ij}$ of the second type. If the gradings of the representations are of opposite sign (in the first situation this is automatically the case for finite KO-dimension $6$, in the second case by construction) there is allowed a component of the Dirac operator whose inner fluctuations will not generate a field, rather the resulting term will act as a mass term. In the first case such a term is called a Majorana mass term. We will cover both of them separately.

\subsubsection{Fourth building block: Majorana mass terms}\label{sec:bb4}

The finite Hilbert space can, for example due to some breaking procedure \cite{CC08, CCM07}, contain representations
\begin{align*}
	&\repl{1}{1'} \oplus \repl{1'}{1} \simeq \mathbb{C} \oplus \mathbb{C},
\end{align*} 
which are each other's antiparticles, e.g.~these representations are not in the adjoint (`diagonal') representation, but the same component $\com$ of the algebra\footnote{For a component $\mathbb{R}$ in the finite algebra this would work as well, but such a component would not give rise to gauge interactions and is therefore unfavourable.} acts on them. Then there is allowed a component $\D{1'1}{11'}$ of the Dirac operator connecting the two. It satisfies the first order condition \eqref{eq:order_one} and its inner fluctuations automatically vanish. Consequently, this component does not generate a scalar, unlike the typical component of a finite Dirac operator. Writing $(\xi, \xi') \in (\com\oplus \com)^{\oplus M}$ (where $M$ denotes the multiplicity of the representation) for the finite part of the fermions, the demand of $D_F$ to commute with $J_F$ reads
\bas
		(\D{11'}{1'1} \bar\xi, \D{1'1}{11'}\bar \xi') = \Big(\overline{\D{1'1}{11'}\xi}, \overline{\D{11'}{1'1}\xi'}\Big).
\eas
Using that $(\D{ij}{ik})^* = \D{ik}{ij}$ this teaches us that the component must be a symmetric matrix. It can be considered as a Majorana mass for the particle $\fer{11'}$ whose finite part is in the representation $\repl{1}{1'}$ (cf.~the Majorana mass for the right handed neutrino in the Standard Model \cite{CCM07}). Then we have 
\begin{defin}\label{def:bb4}
	For an almost-commutative geometry that contains a building block \B{11'} of the second type, a \emph{building block of the fourth type} \BBBB{11'} consists of a component
	\begin{align*}
		\D{1'1}{11'} : \repl{1}{1'} \to \repl{1'}{1} 
	\end{align*}
	of the finite Dirac operator. Symbolically it is denoted by 
\begin{align*}
	\BBBB{11'} = (0, \D{1'1}{11'}) \in \H_F \oplus \End(\H_F),
\end{align*}
where for the symmetric matrix that parametrizes this component we write $\maj$.
\end{defin}
In the language of Krajewski diagrams such a Majorana mass is symbolized by a dotted line, cf.~Figure \ref{fig:MajMass}.

\begin{figure}
	\begin{center}
		\def\svgwidth{.35\textwidth}
		\subimport{\the\svgpath}{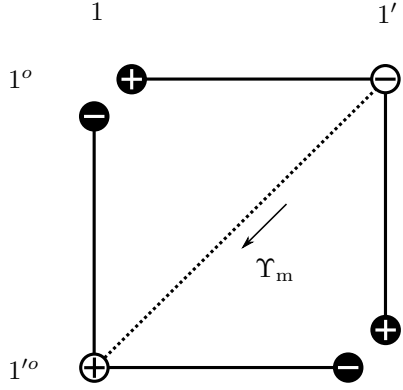}	
		\captionsetup{width=.5\textwidth}
		\caption{A component of the finite Dirac operator that acts as a Majorana mass is represented by a dotted line in a Krajewski diagram.}
		\label{fig:MajMass}	
\end{center}
\end{figure}	

A \BBBB{11'} adds the following to the action \eqref{eq:spectral_action_acg_flat}:
\begin{align}
&	\frac{1}{2}\inpr{J_M\fer{11'L}}{\gamma^5 \maj^* \fer{11'L}} + \frac{1}{2}\inpr{J_M\afer{11'R}}{\gamma^5 \maj \afer{11'R}}\nn\\
&\qquad  + \frac{f(0)}{\pi^2} \Big[|\maj \asfer_{11'}C_{111'}^*|^2 + |\maj \asfer_{11'}C_{1'1'1}^*|^2 + \sum_j \Big( |\maj\yuk{1'}{j}\sfer_{1'j}|^2 + |\maj^*\yuks{1}{j}\sfer_{1j}|^2\Big) \Big]\nn\\[-7pt]
& \qquad + \frac{f(0)}{\pi^2}\sum_j \Big(\tr (\asfer_{11'}C_{111'}^*)^o\maj\yuk{1'}{j}\sfer_{1'j}\asfer_{1j}C_{11j}^* \nn\\[-7pt]
&\qquad\qquad\qquad\qquad +  \tr \maj(\asfer_{11'}C_{1'1'1}^*)^oC_{1'1'j}\sfer_{1'j}\asfer_{1j}\yuk{1}{j}\nn\\
&\qquad\qquad\qquad\qquad	+  \tr \maj\yuk{1'}{j}\sfer_{1'j}(\asfer_{11'}\yuks{1}{1'})^o\asfer_{1j}\yuk{1}{j} + h.c.\Big),\label{eq:bb4-action}
\end{align}
where the traces are over $(\repl{1}{1'})^{\oplus M}$. In this expression, the first contribution comes from the inner product. The paths in the Krajewski diagram corresponding to the other contributions are depicted in Figure \ref{fig:bb4-paths}. In this set up it is $\sfer_{1'j}$ that does not have a family index. Consequently we can separate the traces over the family-index and that over \srep{j} in the penultimate term of the second line of \eqref{eq:bb4-action}. We would like to rewrite the above action in terms of $\yukw{}{} \equiv \yuk{1'}{j}$ by using the identity \eqref{eq:improvedUpsilons2}. For this we first need to rewrite the $C_{iij}$ to the $C_{ijj}$ by employing Remark \ref{rmk:bb3-relativesigns}. Writing out the family indices of the third and fourth line of \eqref{eq:bb4-action} gives
\ba
&  \tr ((\asfer_{11'}C_{111'}^*)^o\maj)_a\sfer_{1'j}\asfer_{1jc}(C_{11j}^*(\yuk{1'}{j})^t)_{ca} + \tr (\maj(\asfer_{11'}C_{1'1'1}^*)^o)_aC_{1'1'j}\sfer_{1'j}\asfer_{1jc}(\yuk{1}{j})_{ca}\nn\\
&=  \sqrt{\frac{n_1\K_j}{n_j\K_1}}\frac{g_1}{g_j}(s_{1'j} - s_{1j}s_{11'})\Big[\tr (\asfer_{11'}C_{11'1'}^*)^o\maj \yuk{1'}{j}\sfer_{1'j}\asfer_{1j}C_{1jj}^*\Big]\label{eq:bb4-idnyuks},
\ea
where $a, b, c$ are family indices, $s_{ij}$ is the product of the signs of $C_{iij}$ and $C_{ijj}$ (cf.~the notation in Remark \ref{rmk:bb3-relativesigns}) and where we have used that $\maj$ is a symmetric matrix. 

\begin{figure}
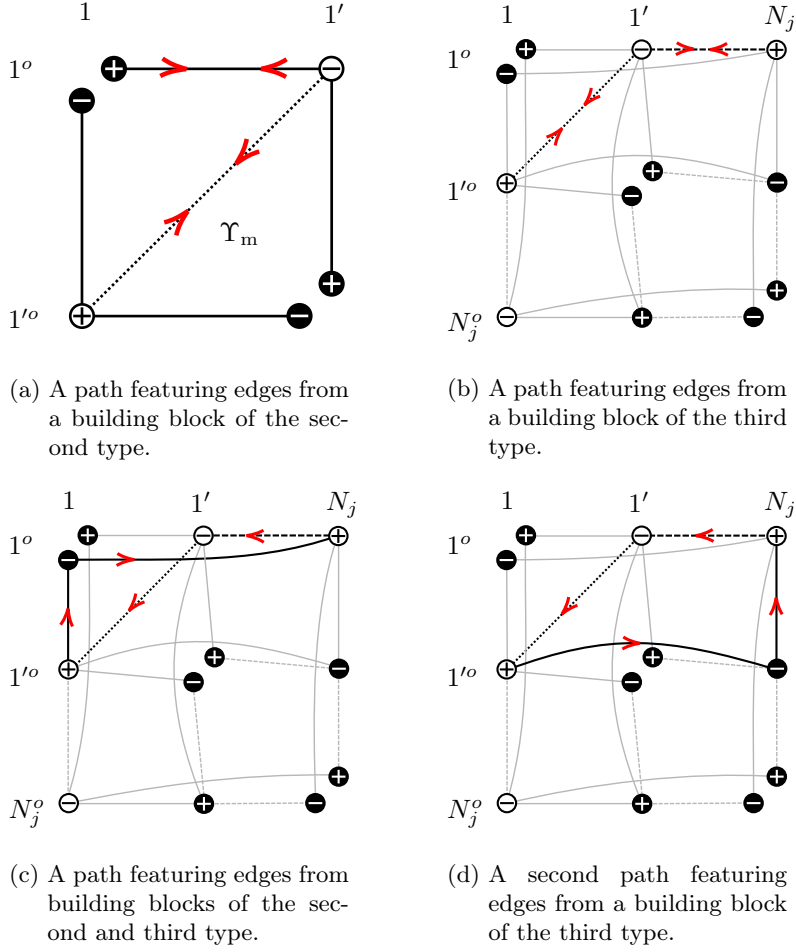

	\centering
	\begin{subfigure}{.3\textwidth}
		\centering
		\def\svgwidth{\textwidth}
		\subimport{\the\svgpath}{bb4-path1.pdf_tex}
		\caption{A path featuring edges from a building block of the second type.}
		\label{fig:bb4-path1}
	\end{subfigure}
	\hspace{30pt}
\begin{subfigure}{.3\textwidth}
		\centering
		\def\svgwidth{\textwidth}
		\subimport{\the\svgpath}{bb4-path3.pdf_tex}
		\caption{A path featuring edges from a building block of the third type.}
		\label{fig:bb4-path3}
	\end{subfigure}
	\\[10pt]
	\begin{subfigure}{.3\textwidth}
		\centering
		\def\svgwidth{\textwidth}
		\subimport{\the\svgpath}{bb4-path2.pdf_tex}
		\caption{A path featuring edges from building blocks of the second and third type.}
		\label{fig:bb4-path2}
	\end{subfigure}
	\hspace{30pt}
\begin{subfigure}{.3\textwidth}
		\centering
		\def\svgwidth{\textwidth}
		\subimport{\the\svgpath}{bb4-path4.pdf_tex}
		\caption{A second path featuring edges from a building block of the third type.}
		\label{fig:bb4-path4}
	\end{subfigure}

\captionsetup{width=.8\textwidth}
\caption{In the case that there is a building block of the fourth type, there are extra interactions in the action.}
\label{fig:bb4-paths}
\end{figure}

Then to make things a bit more apparent, we scale the fields in \eqref{eq:bb4-action} (with the third and fourth line replaced by \eqref{eq:bb4-idnyuks}) according to \eqref{eq:bb3-scalingfields} and put in the expressions for the $C_{ijj}$ from \eqref{eq:bb3-expressionCs-sqrt}, which gives
\ba
&	\frac{1}{2}\inpr{J_M\fer{11'L}}{\gamma^5 \maj^* \fer{11'L}} + \frac{1}{2}\inpr{J_M\afer{11'R}}{\gamma^5 \maj \afer{11'R}}\nn\\
&\qquad +  4r_1|\maj\asfer_{11'}|^2 + 2\sum_j \w{1j}\Big(|\maj\yukw{j}{}|_M^2 |\sfer_{1'j}|^2 + |\maj^*\yukws{j}{}\sfer_{1j}|^2\Big) \nn\\
& \qquad\qquad + \kappa_{1'}\kappa_j\sum_j 2g_m\sqrt{\frac{2\w{1j}}{q_m}} \Big( \tr \asfer_{11'}(r_1 + \w{1j}\yukw{j}{}\yukws{j}{})^t\maj\yukw{j}{}\sfer_{1'j}\asfer_{1j} + h.c.\Big), 
\label{eq:bb4-action-scaled}
\ea
where we have written $|a|^2_M = \tr_M a^*a$ for the trace over the family-index, $\yukw{j}{} \equiv \yukw{1'}{j}$, and where $\kappa_{1'} = \sgnc_{1',j}\sgnc_{1',1}$, $\kappa_{j} = \sgnc_{j,1'}\sgnc_{j,1}\in \{\pm 1\}$. We replaced $\asfer_{11'}^o$ by $\asfer_{11'}$ since these coincide when $\asfer_{11'}$ is a gauge singlet. Consequently, the traces are now over $\mathbf{1}^{\oplus M}$. In addition we used the relation \eqref{eq:improvedUpsilons2} between $\yuk{1}{j}$, $\yuk{1'}{j}$ and $\yuk{1}{1'}$, the symmetry of $\maj$ and that $g_1 \equiv g_{1'}$ (which follows from the set up) and consequently $r_1 = r_{1'}$ and $\w{1'j} = \w{1j}$. In contrast to the previous case, not all scalar interactions that appear here can be accounted for by auxiliary fields:

\begin{lem}\label{prop:bb4}

	For a finite spectral triple that contains, in addition to building blocks of the first, second and third type, one building block of the fourth type, the only terms in the associated spectral action that can be written off shell using the available auxiliary fields are those featuring $\sfer_{11'}$ or its conjugate. 	
\end{lem}
\begin{proof}
The bosonic terms in \eqref{eq:bb4-action} must be the on shell expressions of an off shell Lagrangian that features the auxiliary fields available to us. Respecting gauge invariance, the latter must be  
\ba\label{eq:bb4-auxfields}
- \tr F_{11'}^*F_{11'} - \Big(\tr F_{11'}^*\big(\gamma_{11'} \asfer_{11'} + \sum_j \beta_{11',j}\sfer_{1j}\asfer_{1'j}\big) + h.c.\Big).
\ea
On shell this then gives the following contributions featuring $\sfer_{11'}$ and its conjugate: 
\bas
|\gamma_{11'} \asfer_{11'}|^2 
 + \sum_j 
 \big(\tr \gamma_{11'}\asfer_{11'}\sfer_{1'j}\asfer_{1j}\beta_{11',j}^* + h.c.\big),
\eas
which corresponds at least in form to all bosonic terms of \eqref{eq:bb4-action-scaled}, except the second term of the second line.
\end{proof}

We can use an argument similar to the one we used for building blocks of the third type:

\begin{lem}\label{lem:bb4-aux-susy}
	The action consisting of the fermionic terms of \eqref{eq:bb4-action-scaled} and the terms of \eqref{eq:bb4-auxfields} that do not feature $\beta_{11',j}$ or its conjugate is supersymmetric under the transformations \eqref{eq:susytransforms5} iff
	\ba\label{eq:bb4-aux-susy-demand}
			\gamma^*_{11'}\gamma_{11'} = \maj^*\maj
	\ea
	and the gauginos represented by the black vertices in Figure \ref{fig:bb4-path1} that have the same chirality are associated with each other.
\end{lem}
\begin{proof}
	See Section \ref{sec:bb4-proof}.
\end{proof}

Combining the above two Lemmas, then gives the following result.

\begin{prop}\label{cor:bb4}
The action \eqref{eq:bb4-action-scaled} of a single building block of the fourth type breaks supersymmetry only softly via
\bas
 2\sum_j \w{1j}\Big(|\maj\yukw{j}{}|_M^2|\sfer_{1'j}|^2 + |\maj^*\yukws{j}{}\sfer_{1j}|^2\Big) \nn
\eas
iff
\ba\label{eq:bb4-susy-demands}
	r_1 &= \frac{1}{4} &&\text{and} & \w{1j}\yukw{j}{}\yukws{j}{} &= \Big(- \frac{1}{4} \pm \frac{\kappa_{1'}\kappa_j}{2}\Big)\id_M,
\ea
where the latter should hold for all $j$ appearing in the sum in \eqref{eq:bb4-action}. Here $\kappa_{1'} = \sgnc_{1',j}\sgnc_{1',1}, \kappa_j = \sgnc_{j,1'}\sgnc_{j,1}\in \{\pm 1\}$.
\end{prop} 
\begin{proof}
To prove this, we must match the coefficients of the contribution \eqref{eq:bb4-action-scaled} to the spectral action from a building block \B{11'} to those of the auxiliary fields \eqref{eq:bb4-auxfields}. This requires
\ba\label{eq:bb4-susy-demands-interm}
	\gamma_{11'} &= 2\sqrt{r_1}e^{i\phi_\gamma}\maj, & \kappa_{1'}\kappa_j 2g_m \sqrt{\frac{2\w{1j}}{q_m}}(r_1\id_M  + \w{1j} \yukw{j}{}\yukws{j}{})^t\maj\yukw{j}{} = \gamma_{11'}(\beta_{11',j}^*)^t
\ea
for all $j$, where $e^{i\phi_\gamma}$ denotes the phase ambiguity left in $\maj$ from \eqref{eq:bb4-aux-susy-demand} and where we have used the symmetry of $\maj$. From supersymmetry $\gamma_{11'}$ is in addition constrained by \eqref{eq:bb4-aux-susy-demand}, which requires the first relation of \eqref{eq:bb4-susy-demands} to hold. For the building block \B{11'j} to have a supersymmetric action we demand
\bas	
		\beta_{11',j}^* &= g_m \sqrt{\frac{2\w{1j}}{q_m}} e^{-i\phi_{\beta_j}}(\yukw{j}{})^t,
\eas
which can be obtained by combining the demand \eqref{eq:bb3-susy-demand2} with the relation \eqref{eq:yukwyukp}, but keeping Remark \ref{rmk:bb3-R=1} in mind since it is $\sfer_{1'j}$ that does not have a family index. As is with $\maj$, the demand \eqref{eq:bb3-susy-demand2} determines $\beta_{11',j}$ only up to a phase $\phi_{\beta_j}$. Comparing this with the second demand of \eqref{eq:bb4-susy-demands-interm}, inserting \eqref{eq:bb4-aux-susy-demand} and using the symmetry of $\maj$, we must have 
\bas
	\phi_\gamma &= \phi_{\beta_j} \mod \pi, & 2(r_1\id_M + \w{1j}\yukw{j}{}\yukws{j}{}) = \pm \kappa_{1'}\kappa_j 2\sqrt{r_1}\id_M.
\eas
 Inserting the first relation of \eqref{eq:bb4-susy-demands}, its second relation follows. The second term of the second line of \eqref{eq:bb4-action-scaled} cannot be accounted for by the auxiliary fields at hand, which establishes the result.
\end{proof}

It is not per se impossible to write all of \eqref{eq:bb4-action-scaled} off shell in terms of auxiliary fields, but to avoid the obstruction from Lemma \ref{prop:bb4} at least requires the presence of mass terms for the representation $\sfer_{1j}$ and $\sfer_{1'j}$ such as the ones that are discussed in the next section.

%
%
%
%
%

\subsubsection{Fifth building block: `mass' terms}\label{sec:bb5}

If there are two building blocks of the second type with the same indices ---say $i$ and $j$--- but with different values for the grading, we are in the situation as depicted in Figure \ref{fig:bb5}. On the basis 
\ba\label{eq:bb5-basis}
	\Big[(\rep{i}{j})_L \oplus (\rep{j}{i})_R \oplus (\rep{i}{j})_R \oplus (\rep{j}{i})_L\Big]^{\oplus M}, 
\ea
the most general finite Dirac operator that satisfies the demand of self-adjointness, the first order condition \eqref{eq:order_one} and that commutes with $J_F$ is of the form
\ba\label{eq:bb5-DF}
	D_F = \begin{pmatrix}
			0  & 0 & \mu_i + \mu_j^o  & 0 \\
		0 & 0 & 0 & (\mu_i^o)^* + \mu_j^*\\
		\mu_i^* + (\mu_j^*)^o & 0 & 0 & 0   \\
		0 &  \mu_i^o + \mu_j  & 0 & 0 
	\end{pmatrix}
\ea
with $\mu_i \in M_{N_iM}(\com)$ and $\mu_j \in M_{N_jM}(\com)$. The inner fluctuations for general such matrices $\mu_{i,j}$ will generate scalar fields in the representations $M_{N_{i,j}}(\com)$. If we want these components to result in mass terms in the action, we should restrict them both to only act non-trivially on possible generations, i.e.~for a single generation the components are equal to a complex number. We will write $\mu := \mu_i + \mu_j^o \in M_{M}(\com)$ for the restricted component. 

This gives rise to the following definition.
\begin{defin}\label{def:bb5}
For a finite spectral triple that contains building blocks \Bc{ij}{\pm} and \Bc{ij}{\mp} of the second type (both with multiplicity $M$), a \emph{building block of the fifth type} is a component of $D_F$ that runs between the representations of the two building blocks and acts only non-trivially on the $M$ copies. Symbolically:
\begin{align*}
	\B{\mathrm{mass}, ij} = (0, \D{ijL}{ijR}) \in \H_F \oplus \End(\H_F).
\end{align*}
We denote this component with $\mu \in M_M(\com)$. 
\end{defin}

\begin{figure}
	\begin{center}
		\def\svgwidth{.45\textwidth}
		\subimport{\the\svgpath}{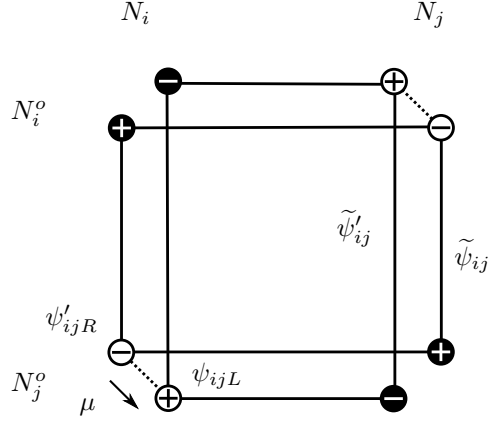}	
		\captionsetup{width=.7\textwidth}
		\caption{The case with two building blocks of the second type that have the same indices but an opposite grading; a component of the finite Dirac operator mapping between the two copies will generate a mass-term, indicated by the dotted line with the `$\mu$'.}
		\label{fig:bb5}	
\end{center}
\end{figure}	

If for convenience we restrict to the upper signs for the chiralities of the building blocks and write 
\bas
 (\fer{ijL},  \afer{ijR}, \fer{ijR}', \afer{ijL}')
\eas
for the elements of $L^2(M, S \otimes \H_F)$ on the basis \eqref{eq:bb5-basis} (where the first two fields are associated to \Bc{ij}{+} and the last two to \Bc{ij}{-}), then the contribution of \eqref{eq:bb5-DF} to the fermionic action reads
\ba
	S_{f,\textrm{mass}}[\zeta] &= \frac{1}{2}\inpr{J(\fer{ijL}, \afer{ijR}, \fer{ijR}', \afer{ijL}')}{\gamma^5 D_F (\fer{ijL}, \fer{ijR}', \afer{ijR}, \afer{ijL}')}\nn\\
&\qquad= \inpr{J_M \afer{ijR}}{\gamma^5 \mu\fer{ijR}'}
+ \inpr{J_M\afer{ijL}'}{\gamma^5 \mu^*\fer{ijL}}.\label{eq:bb5-action-ferm}
\ea
Let $\sfer$ and $\sfer'$ be the sfermions that are associated to \Bc{ij}{+} and \Bc{ij}{-} respectively, then the extra contributions to the spectral action as a result of adding this building block are given by
\begin{align}
	S_{b,\textrm{mass}}[\szeta] &= \frac{f(0)}{\pi^2}(N_i|\mu^*C_{iij}\sfer_{ij}|^2 + N_j|\mu^* C_{ijj}\sfer_{ij}|^2 + N_i|\mu C_{iij}'\sfer_{ij}'|^2 + N_j|\mu C_{ijj}'\sfer_{ij}'|^2)\nn\\
&\quad + \frac{f(0)}{\pi^2}\sum_{k}\Big[ N_i \tr \mu^*\asfer_{ij}'C_{iij}'^*C_{iik}\sfer_{ik}\asfer_{jk}\yuks{j}{k} + N_j \tr \asfer_{ij}'C_{ijj}'^*\,\mu^*\yuks{i}{k}\sfer_{ik}\asfer_{jk}C_{jjk}^* + h.c.\nn\\
&\qquad + \Big(N_j\tr_M(\mu\mu^*\yuks{i}{k}\yuk{i}{k})|\sfer_{ik}|^2 + N_i|\mu\yuk{j}{k}\sfer_{jk}|^2\Big)\Big],\label{eq:bb5-action}
\end{align}
where the second and third lines arise in a situation where for some $k$, $\B{ijk}$ is present. The paths corresponding to these expressions are depicted in Figure \ref{fig:bb5-paths}. Here, the $C_{iij}$ with a prime correspond to the components of the Dirac operator of \Bc{ij}{-}. We assume that they also satisfy \eqref{eq:bb2-resultCiij}. In this context $\sfer_{ik}$ does not have a family-index and consequently we could separate the traces in the first term of the third line of \eqref{eq:bb5-action}.\\

In a similar way as with the building block of the fourth type we can rewrite the second line of \eqref{eq:bb5-action} using Remarks \ref{rmk:bb3-relativesigns} and \ref{rmk:bb3-R=1}, giving
\ba
&\frac{f(0)}{\pi^2} \Big[N_i \tr (\asfer_{ij}'C_{iij}'^*)_aC_{iik}\sfer_{ik}\asfer_{jkb}(\yuks{j}{k}(\mu^*)^t)_{ba} \nn\\
&\qquad + N_j \tr (\asfer_{ij}'C_{ijj}'^*)_a\,(\mu^*\yuks{i}{k})_{ac}\sfer_{ik}\asfer_{jkb}(C_{jjk}^*)_{bc} + h.c.\Big]\nn\\
&= s_{jk}\bigg(\frac{N_ir_i + N_jr_j}{\sqrt{n_jn_k}g_jg_k}\bigg) \tr \asfer_{ij}'C_{ijj}'^*\mu^*\yuks{i}{k}\sfer_{ik}\asfer_{jk}C_{jkk}^*  + h.c.\label{eq:bb5-relativesigns}
\ea
Replacing the second line of \eqref{eq:bb5-action} with \eqref{eq:bb5-relativesigns} and then scaling the fields and rewriting $\yuk{i}{j}$ and $\yuk{j}{k}$ in terms of $\yuk{i}{k} \equiv\yuk{}{}$ using the identities \eqref{eq:improvedUpsilons2}, reduces the bosonic contribution \eqref{eq:bb5-action} to
\begin{align}
	& 2(1 - \w{ij})\big(|\mu^* \sfer_{ij}|^2 + |\mu \sfer_{ij}'|^2\big) + 2\sum_{k}\bigg[\kappa_{j} g_l (1 - \w{ij})\sqrt{\frac{2\w{ik}}{q_l}} \tr \asfer_{ij}'\mu^*\yukws{}{}\sfer_{ik}\asfer_{jk} + h.c.\nn\\
&\qquad\qquad + \w{ik}\Big(N_j|\yuk{}{}\mu|_M^2|\sfer_{ik}|^2 + N_i|\mu\yukw{}{}\sfer_{jk}|^2\Big)\bigg],\label{eq:bb5-action-scaled}
\end{align}
where we have again employed the notation $|a|_M^2 = \tr_M a^*a$ for the trace over the family-index and used that $s_{jk}\sgnc_{j,i}\sgnc_{k,j} = \sgnc_{j,i}\sgnc_{j,k} \equiv \kappa_j \in \{\pm\}$. The index $l$ can take any of the values that appear in the model. \\

\begin{figure}
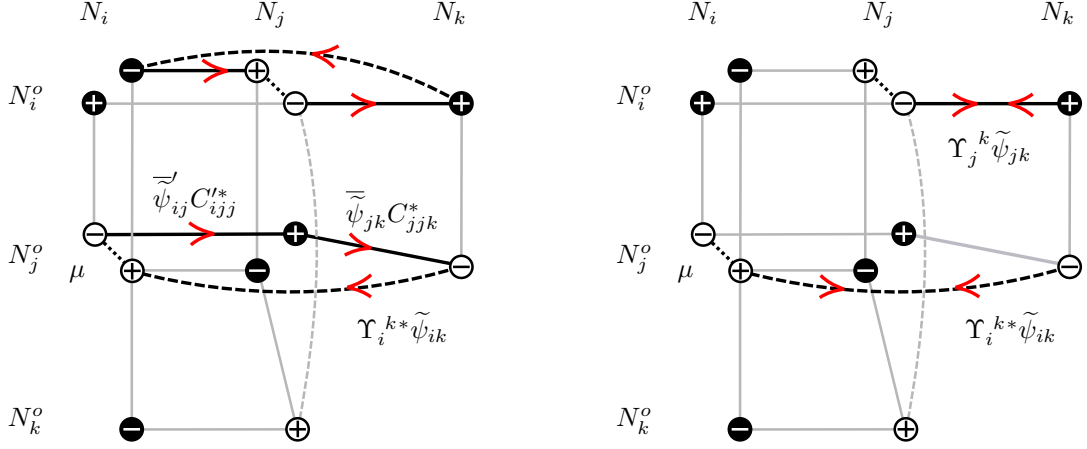

	\centering
	\begin{subfigure}{.45\textwidth}
		\centering
		\def\svgwidth{\textwidth}
		\subimport{\the\svgpath}{bb5-path1.pdf_tex}
		\caption{A path with $\mu$, featuring edges from a building block of the second and third type.}
		\label{fig:bb5-path1}
	\end{subfigure}
	\hspace{30pt}
	\begin{subfigure}{.45\textwidth}
		\centering
		\def\svgwidth{\textwidth}
		\subimport{\the\svgpath}{bb5-path2.pdf_tex}
		\caption{A path with $\mu$, featuring only edges from building blocks of the third and fifth type.}
		\label{fig:bb5-path2}
	\end{subfigure}
		\caption{In the case of a building block of the fifth type, there are various extra contributions to the action, depending on the content of the finite spectral triple.}
		\label{fig:bb5-paths}	
\end{figure}	

Here we have a similar result as in the previous section:

\begin{lem}\label{lem:bb5}

	For a finite spectral triple that contains, in addition to building blocks of the first, second and third type, one building block of the fifth type, the only terms in the associated spectral action that can be written off shell are those featuring $\sfer_{ij}$, $\sfer_{ij}'$ or their conjugates. 	
\end{lem}
\begin{proof}
	In order to rewrite the first terms of \eqref{eq:bb5-action-scaled} in terms of auxiliary fields, we must introduce an interaction featuring one auxiliary field $F$ and one sfermion. Since $\sfer_{ij}$ and $\sfer_{ij}'$ are in the same representation of the algebra, we can choose whether to couple $\sfer_{ij}$ to $F_{ij}$ (corresponding to \Bc{ij}{+}) or to $F_{ij}'$ (corresponding to \Bc{ij}{-}). The same holds for $\sfer_{ij}'$. Transforming the fermions in \eqref{eq:bb5-action-ferm} according to \eqref{eq:susytransforms4} suggests that, in order to have a chance at supersymmetry, we must couple $F_{ij}'$ to $\sfer_{ij}$ and $F_{ij}$ to $\sfer_{ij}'$. We thus write 
	\ba\label{eq:bb5-auxfields}
		- \tr F_{ij}^*F_{ij} - \tr F_{ij}'^*F_{ij}' - \big( \tr F_{ij}^*\delta_{ij}'\sfer_{ij}' + \tr F_{ij}'^*\delta_{ij} \sfer_{ij} + h.c.\big)
	\ea
with $\delta_{ij}, \delta_{ij}' \in M_M(\com)$. This yields on shell 
$
	|\delta_{ij}\sfer_{ij}|^2 + |\delta_{ij}'\sfer_{ij}'|^2
$, which is indeed of the same form as the first two terms in \eqref{eq:bb5-action-scaled}. In the case that there is a building block \B{ijk} of the third type present, the extra contributions to the action must come from the cross terms of
\bas
		- \tr F_{ij}^*F_{ij} - \tr F_{ij}'^*F_{ij}' -  \Big[ \tr F_{ij}^*\big(\delta_{ij}'\sfer_{ij}' + \beta_{ij,k}\sfer_{ik}\asfer_{jk}\big) + \tr F_{ij}'^*\delta_{ij}\sfer_{ij} + h.c.\Big]
\eas
where the interaction with $\beta_{ij,k}$ corresponds to the second term of \eqref{eq:bb3-auxfields}. On shell this gives us the additional interaction
\ba
 \tr \asfer_{ij}'\delta_{ij}'^*\beta_{ij,k}\sfer_{ik}\asfer_{jk} + h.c.\label{eq:bb5-cross-term}
\ea
In form, this indeed coincides with the second line of \eqref{eq:bb5-action-scaled}. The last two terms of \eqref{eq:bb5-action-scaled} do not appear here and consequently they cannot be addressed using the auxiliary fields that are available to us when having only building blocks of the first, second and third type.
\end{proof}

Similar as with the previous building blocks we can check what the demands for off shell supersymmetry are.

\begin{lem}\label{prop:bb5}
	The action consisting of the fermionic action \eqref{eq:bb5-action-ferm} and the off shell action \eqref{eq:bb5-auxfields} is supersymmetric under the transformations \eqref{eq:susytransforms5} if and only if 
	\ba\label{eq:bb5-constraints}
		\delta\delta^* &= \mu^*\mu,& \delta'\delta'^* &= \mu\mu^*.
	\ea
\end{lem}
\begin{proof}
	See Section \ref{sec:bb5-proof}.
\end{proof}

Combining the above lemmas gives the following result for a building block of the fifth type.

\begin{prop}\label{cor:bb5}
	For a finite spectral triple that contains, in addition to building blocks of the first, second and third type, one building block of the fifth type, the action of a single building block of the fifth type breaks supersymmetry only softly via 
\bas
&\w{ik}\Big(N_j|\yukw{}{}\mu|_M^2|\sfer_{ik}|^2 + N_i|\mu\yukw{}{}\sfer_{jk}|^2\Big)
\eas
iff
\bas
	\w{ij} &= \frac{1}{2}
\eas
and the product of the possible phases of $\delta'^*$ and $\beta_{ij,k}$ (cf.~\eqref{eq:bb5-constraints} and \eqref{eq:bb3-susy-demand2} respectively) is equal to $\sgnc_{j,i}\sgnc_{j,k}$.
\end{prop}
\begin{proof}
This follows from comparing the spectral action \eqref{eq:bb5-action-scaled} with the off shell action \eqref{eq:bb5-auxfields} and using the demands \eqref{eq:bb5-constraints} and \eqref{eq:bb3-susy-demand2}.
\end{proof}

The form of the soft breaking term suggests that, in order to let it be part of a truly supersymmetric action, we have the following necessary requirement. Each two building blocks of the second type that are connected to each other via an edge of a building block of the third type, both need to have a building block of the fifth type defined on them. In the case above this would have been $\sfer_{ik}$ and $\sfer_{jk}$. 

\section{Conditions for a supersymmetric spectral action}\label{sec:4s-aux}

Our aim is to determine whether the total action that corresponds to an almost-commutative geometry consisting of various of the five identified building blocks, is supersymmetric. More than once we used the following strategy for that. First, we identified the off shell counterparts for the contributions of $\tr_F \Phi^4$ to the (on shell) spectral action, using the available auxiliary fields and coefficients whose values were undetermined still. Second, we derived constraints for these coefficients based on the demand of having supersymmetry for the fermionic action and this off shell action. Finally, we should check if the off shell interactions correspond on shell to the spectral action again, when their coefficients satisfy the constraints that supersymmetry puts on them. If this is the case then the action from noncommutative geometry is an on shell counterpart of an off shell action that is supersymmetric.

In the previous sections we have experienced multiple times that the pre-factors of all bosonic interactions can get additional contributions when extending the almost-commutative geometry. As was stated before, we should therefore assess whether or not the demands from supersymmetry on the coefficients are satisfied for the final model only. In this section we will present an overview of all four-scalar interactions that have appeared previously, from which building blocks their pre-factors get what contributions and which demands hold for them. We identify several such demands, thus constructing a checklist for supersymmetry.

\begin{enumerate}
\item To have supersymmetry for a building block \B{ij} of the second type, the components of the finite Dirac operator should satisfy \eqref{eq:bb2-resultCiij}, after scaling them. For a single building block of the second type this demand can only be satisfied for $N_i = N_j$ and $M = 4$ (Proposition \ref{lem:bb2-nosol}). When \B{ij} is part of a building block of the third type the demand is automatically satisfied via the solution \eqref{eq:bb3-expressionCs-sqrt}.

\item A necessary requirement to have supersymmetry for any building block \B{ijk} of the third type (Section \ref{sec:bb3}), is that the scaled parameters of the finite Dirac operator that make up such a building block satisfy 
\ba\label{eq:demand-0}
	\w{jk}\yukws{j}{k}\yukw{j}{k} &= \w{ik}\yukws{i}{k}\yukw{i}{k} =  
	\w{ij}\yukws{i}{j}\yukw{i}{j} =: \Om{ijk}^*\Om{ijk}.
\ea
This relation can be obtained from \eqref{eq:improvedUpsilons1}, multiplying each term with its conjugate. For notational convenience we have introduced the variable $\Om{ijk}^*\Om{ijk}$. 

\item Terms $\propto |\sfer_{ij}\asfer_{ij}|^2$ appear for the first time with a building block of the second type (\eqref{eq:exprM1} in Section \ref{sec:bb2}) but also get contributions from a building block \B{ijk} of the third type (first term of \eqref{eq:bb3-boson-action}). The total expression reads 
	\ba
		& \frac{f(0)}{2\pi^2} \Big[ N_i|C_{iij}^*C_{iij}\sfer_{ij} \asfer_{ij}|^2 + N_j|C_{iij}^*C_{ijj}\sfer_{ij} \asfer_{ij}|^2 \nn\\
		&\qquad + \sum_k N_k |\yuks{i,k}{j}\yuk{i,k}{j}\sfer_{ij}\asfer_{ij}|^2\Big]\nn\\
		%
		&\to 2\frac{g_i^2}{q_i}\bigg|\bigg( N_ir_i^2 + \alpha_{ij}  \sum_k N_k(\Om{ijk}^*\Om{ijk})^2\bigg)^{1/2}\sfer_{ij}\asfer_{ij}\bigg|^2\nn\\
		&\qquad + 2\frac{g_j^2}{q_j}\bigg|\bigg(N_jr_j^2 + (1 - \alpha_{ij}) \sum_k N_k (\Om{ijk}^*\Om{ijk})^2\bigg)^{1/2}\sfer_{ij}\asfer_{ij}\bigg|^2,\nn
	\ea
	upon scaling the fields. Here we have introduced a parameter $\alpha_{ij} \in \mathbb{R}$ that tells how any new contributions are divided over the initial two. Such terms can only be described off shell using the auxiliary fields $G_{i}$ and $G_{j}$ (cf.~Lemma \ref{lem:bb2-offshell}) via
\bas
	 - \frac{1}{2n_i}\tr G_i\big(G_i + 2n_i\P_i \sfer_{ij}\asfer_{ij}\big)  
	 - \frac{1}{2n_j}\tr G_j\big(G_j + 2n_j \asfer_{ij}\P_j\sfer_{ij}\big),
\eas
which on shell equals
\bas
	\frac{n_i}{2}|\P_i \sfer_{ij}\asfer_{ij}|^2 + \frac{n_j}{2}|\asfer_{ij}\P_j \sfer_{ij}|^2, 
\eas
cf.~\eqref{eq:bb2-auxterms}. Comparing this with the above expression sets the coefficients $\P_{i}$ and $\P_{j}$:
\bas
		\frac{n_i}{2}\P_i^2 &= 2\frac{g_i^2}{q_i}\bigg(N_ir_i^2 + \alpha_{ij}  \sum_k N_k (\Om{ijk}^*\Om{ijk})^2\bigg), \nn\\
		\frac{n_j}{2}\P_j^2 &= 2\frac{g_j^2}{q_i}\bigg(N_jr_j^2 + (1 - \alpha_{ij}) \sum_k N_k (\Om{ijk}^*\Om{ijk})^2\bigg),\nn
\eas
where there is an additional trace over the last terms if $\sfer_{ij}$ has no family index. If the action is supersymmetric then \eqref{eq:bb2-resultCiij} can be used with $\K_i = \K_j = 1$ and the above relations read
	\ba\label{eq:demand-1}
		\frac{r_i}{4} &= N_ir_i^2 + \alpha_{ij}  \sum_k N_k \tr[(\Om{ijk}^*\Om{ijk})^2],\nn\\
  	\frac{r_j}{4} &= N_jr_j^2 + (1 - \alpha_{ij}) \sum_k N_k \tr[(\Om{ijk}^*\Om{ijk})^2], 
	\ea
when $\sfer_{ij}$ has no family index and 
	\ba\label{eq:demand-1.1}
		\frac{r_i}{4}\id_M &= N_ir_i^2\id_M + \alpha_{ij}  \sum_k N_k (\Om{ijk}^*\Om{ijk})^2, \nn\\
  	\frac{r_j}{4}\id_M &= N_jr_j^2\id_M + (1 - \alpha_{ij}) \sum_k N_k (\Om{ijk}^*\Om{ijk})^2, 
	\ea
when it does. Here we have used that $r_i = q_in_i$. 

\item An interaction $\propto |\sfer_{ij}\sfer_{jk}|^2$ can receive contributions in two different ways; one comes from a building block \B{ijk} of the third type \eqref{eq:bb3-action-final}, the other comes from two adjacent building blocks \B{ijl} and \B{jkl} (first and second term of \eqref{eq:2bb3-action-scaled}, but occurs only for particular values of the grading): 
\bas
	&	g_m^2\frac{4\w{ij}}{q_m}(1 - \w{ik}) |\yukw{i,k}{j}\sfer_{ij}\sfer_{jk}|^2 + 4\bigg( 
n_jr_jN_jg_j^2|\sfer_{ij}\sfer_{jk}|^2 + \frac{g_m^2}{q_m}\w{ij}\w{jk}N_l|\yukw{i,l}{j}\sfer_{ij}\yukw{j,l}{k}\sfer_{jk}|^2\bigg).\nn
\eas
From this, however, we need to subtract the value $n_j g_j^2|\sfer_{ij}\sfer_{jk}|^2$ that is expected from the cross term
\bas
 -  \tr G_j\big(\P_{j,i}\asfer_{ij}\sfer_{ij} + \P_{j,k}\sfer_{jk}\asfer_{jk}\big),
\eas
that should already be there when the almost-commutative geometry contains \Bc{ij}{\pm} and \Bc{jk}{\mp} but nevertheless does not appear in the spectral action (see Section \ref{sec:2bb2} and the discussion above Theorem \ref{prop:bb3}). 
The remaining terms must be accounted for by 
\ba\label{eq:bb3-auxfields2}
	- \tr F_{ik}^*F_{ik} + \big(\tr F_{ik}^*\bps_{ik,j}\sfer_{ij}\sfer_{jk} + h.c.\big) 
\ea
which equals
\bas \tr \asfer_{jk}\asfer_{ij}\bp_{ik,j}\bps_{ik,j}\sfer_{ij}\sfer_{jk}
\eas
on shell. Since $\beta_{ik,j}\beta_{ik,j}^*$ is positive definite we can also write the above as
\bas
	 |(\bp_{ik,j}\bps_{ik,j})^{1/2}\sfer_{ij}\sfer_{jk}|^2.
\eas
Comparing the above relations, the off shell action \eqref{eq:bb3-auxfields2} corresponds on shell to the spectral action, iff
\bas
	\bp_{ik,j}\bps_{ik,j} &= g_m^2\frac{4\w{ij}}{q_m}(1 - \w{ik}) \yukws{i,k}{j}\yukw{i,k}{j} - n_j g_j^2\id_M\nn\\
			&\qquad + 4\bigg(n_jr_jN_j g_j^2\id_M + \frac{g_m^2}{q_m}\w{ij}\w{jk}N_l(\yukw{i,l}{j}\yukw{j,l}{k})^*(\yukw{i,l}{j}\yukw{j,l}{k})\bigg),
\eas
where we have assumed that it is $\sfer_{ij}$ not having a family structure.
Furthermore, from the demand of supersymmetry $\bp_{ik,j}$ must satisfy 
\bas
	\bp_{ik,j}\bps_{ik,j} &= g_m^2\frac{2\w{ij}}{q_m}\yukws{i}{j}\yukw{i}{j} \equiv 2\frac{g_m^2}{q_m}\Om{ijk}^*\Om{ijk}
\eas
i.e.~\eqref{eq:bb3-susy-demand2},\footnote{In fact, in \eqref{eq:bb3-susy-demand2} the variables are in reversed order compared to here but looking at \eqref{eq:bb3-constr2} ---from which the former is derived--- one sees immediately that this also holds.} but with $\yukp{}{}$ replaced by $\yukw{}{}$ using \eqref{eq:yukwyukp}. Combining the above two relations, we require that 
\bas
	\frac{2g_m^2}{q_m}\Om{ijk}^*\Om{ijk}& = 4\frac{g_m^2}{q_m}(1 - \w{ik})\Om{ijk}^*\Om{ijk} - n_j g_j^2\id_M \nn\\ &\qquad + 4\bigg( 
n_jr_jN_jg_j^2\id_M + \frac{g_m^2}{q_m}\w{ij}\w{jk}N_l(\yukw{i,l}{j}\yukw{j,l}{k})^*(\yukw{i,l}{j}\yukw{j,l}{k})\bigg),
\eas
using the notation introduced in \eqref{eq:demand-0}. Setting $m = j$ in particular, this reduces to 
\ba\label{eq:demand-2}
	 2(1 - 2\w{ik})\Om{ijk}^*\Om{ijk} - r_j \id_M  + 4\bigg( 
N_jr_j^2\id_M + \w{ij}\w{jk}N_l(\yukw{i,l}{j}\yukw{j,l}{k})^*(\yukw{i,l}{j}\yukw{j,l}{k})\bigg) = 0.
\ea


%

\item The interaction $\propto \tr\sfer_{ik}\asfer_{jk}\sfer_{jl}\asfer_{il}$ only appears in the case of two adjacent building blocks \B{ijk} and \B{ijl} of the third type (cf.~the Lagrangian \eqref{eq:2bb3-action-scaled}). Equating this term to \eqref{eq:2bb3-aux} that appears from the auxiliary field $F_{ij}$, gives
\bas
 \kappa_{k}\kappa_{l} 4\frac{g_m^2}{q_m}(1 - \w{ij})\w{ij}\tr \yukw{l}{}\yukws{k}{}\sfer_{ik}\asfer_{jk}\sfer_{jl}\asfer_{il} + h.c. = \tr \bps_{ij,l}\bp_{ij,k}\sfer_{ik}\asfer_{jk}\sfer_{jl}\asfer_{il} + h.c.,
\eas
with $\kappa_{k} = \sgnc_{k,i}\sgnc_{k,j}, \kappa_{l} = \sgnc_{l,i}\sgnc_{l,j}$. From the demand of supersymmetry $\bps_{ij,l}$ and $\bp_{ij,k}$ should satisfy \eqref{eq:bb3-susy-demand2}. Their phases, if any, must be opposite modulo $\pi$ for the action to be real. We write $\phi_{kl}$ for the remaining sign ambiguity. Inserting these demands above and using \eqref{eq:yukwyukp} requires that $\kappa_{k}\kappa_{l} 4\w{ij}(1- \w{ij}) = 2\phi_{kl}\w{ij}$ for this interaction to be covered by the auxiliary field $F_{ij}$. This has two solutions, the only acceptable of which is 
\ba\label{eq:demand-3}
	\phi_{kl} &= \kappa_{k}\kappa_{l},& \w{ij} &= \frac{1}{2} \quad\Longrightarrow\quad
			  r_iN_i + r_jN_j = \frac{1}{2},
\ea
where we have used \eqref{eq:kintermnorm}.
\item From the spectral action interactions $\propto |\sfer_{ij}|^4$ only appear in the context of a building block of the second type as 
\bas
	\frac{f(0)}{\pi^2}|C_{iij}\sfer_{ij}|^2 |C_{ijj}\sfer_{ij}|^2 \to 4\frac{g_l^2}{q_l}r_ir_j |\sfer_{ij}|^4,
\eas
see \eqref{eq:bb2-action}. Via the auxiliary fields on the other hand they appear in two ways; from the $G_{i,j}$ and via the $u(1)$-field $H$ (see Lemma \ref{lem:bb2-offshell} for both). The latter give on shell the contributions
\bas
	\bigg(\frac{\Q_{ij}^2}{2} - n_i\frac{\P_i^2}{2N_i} - n_j\frac{\P_j^2}{2N_j}\bigg)|\sfer_{ij}|^4,
\eas		
where the minus-signs stem from the identity \eqref{eq:idn-sun-gens} between the generators $T^a_{i,j}$ of $su(N_{i,j})$. Demanding supersymmetry, $\P_i^2$ must equal $g_i^2$ and similarly $\P_j^2 = g_j^2$. In order for the interactions from the spectral action to equal the above equation, $\Q_{ij}^2$ is then set to be
	\ba\label{eq:demand-4}
		\Q_{ij}^2 = \frac{g_l^2}{q_l}\bigg(8r_ir_j + \frac{r_i}{N_i} + \frac{r_j}{N_j}\bigg).
	\ea
	In the case that $\sfer_{ij}$ has family indices, the expressions for $\P_{i,j}^2$ and $\Q_{ij}^2$ must be multiplied with the $M \times M$ identity matrix $\id_M$.

\item Interactions $\propto |\sfer_{ij}|^2|\sfer_{jk}|^2$ (having one common index $j$) appear via the spectral action in two different ways. First of all from two adjacent building blocks \B{ij} and \B{jk} of the second type (cf.~\eqref{eq:2bb2s-different}), and secondly from a building block of the third type (second line of \eqref{eq:bb3-boson-action}). This gives
\bas
& \frac{f(0)}{\pi^2}\Big( |C_{ijj}\sfer_{ij}|^2|C_{jjk}\sfer_{jk}|^2
		+ |\sfer_{ij}|^2|\yuks{i}{j}\yuk{j}{k}\sfer_{jk}|^2\Big) \nn\\
&\qquad \to 4\frac{g_l^2}{q_l}\Big( r_j^2|\sfer_{ij}|^2|\sfer_{jk}|^2
		 + \w{jk}\w{ij}|\sfer_{ij}|^2|\yukws{i}{j}\yukw{j}{k}\sfer_{jk}|^2\Big),
\eas
where we have assumed $\sfer_{ij}$ not to have a family-index. We can write this as
\bas
	& 4\frac{g_l^2}{q_l}\big|\sfer_{ij}\big|^2 \big|\big(r_j^2\id_M + \w{ij}\w{jk}(\yukws{i}{j}\yukw{j}{k})^*\yukws{i}{j}\yukw{j}{k}\big)^{1/2}\sfer_{jk}\big|^2.
\eas
From the auxiliary fields these terms can appear via $G_j$ (with coefficients $\P_{j,i}$ and $\P_{j,k}$, i.e.~as in \eqref{eq:2bb2-different-aux}) and via the $u(1)$-field $H$ with coefficients $\Q_{ij}$ and $\Q_{jk}$:
\bas
	\bigg[\Q_{ij}\Q_{jk} - n_j\frac{\P_{j,i}\P_{j,k}}{N_j}\bigg]|\sfer_{ij}|^2 |\sfer_{jk}|^2.
\eas
Equating the terms from the spectral action and those from the auxiliary fields, and 
inserting the values for the coefficients $\P_{j,i}$, $\P_{j,k}$ (from \eqref{eq:bb2-resultCiij}), $\Q_{ij}$ and $\Q_{jk}$ (from \eqref{eq:demand-4}) that we obtain from supersymmetry, we require
\ba\label{eq:demand-5}
&\bigg(2r_ir_j + \frac{r_i}{4N_i} + \frac{r_j}{4N_j}\bigg)
\bigg(2r_jr_k + \frac{r_j}{4N_j} + \frac{r_k}{4N_k}\bigg)\id_M\nn\\
&\qquad = \Big[\Big(r_j^2 + \frac{r_j}{4N_j}\Big)\id_M + \w{ij}\w{jk}(\yukws{i}{j}\yukw{j}{k})^*\yukws{i}{j}\yukw{j}{k}\Big]^2.
	\ea

\item There are interactions $\propto |\sfer_{ik}|^2|\sfer_{jl}|^2$ and $\propto |\sfer_{jk}|^2|\sfer_{il}|^2$ that arise from two adjacent building blocks \B{ijk} and \B{ijl} of the third type. The first of these is given by 
\bas
4\frac{g_m^2}{q_m}|(\w{ik}\yukw{i,j}{k}\yukws{i,j}{k})^{1/2}\sfer_{ik}|^2|(\w{jl}\yukws{j,i}{l}\yukw{j,i}{l})^{1/2}\sfer_{jl}|^2,\nn
\eas
see \eqref{eq:2bb3-action-scaled}. Since the interactions are characterized by four different indices, the auxiliary fields $G_i$ cannot account for these and consequently they should be described by the $u(1)$-field $H$:
\bas
	|\Q_{ik}^{1/2}\sfer_{ik}|^2|\Q_{jl}^{1/2}\sfer_{jl}|^2.
\eas
 In order for the spectral action to be written off shell we thus require that 
	\bas
		 \Q_{ik}\Q_{jl} &= 4\frac{g_m^2}{q_m}\Om{ijk}\Om{ijk}^*\Om{ijl}^*\Om{ijl}.\nn
	\eas
	With $\Q_{ik}$ and $\Q_{jl}$ being determined by \eqref{eq:demand-4} from the demand of supersymmetry, we can infer from this that for the squares of these expressions we must have
	\ba
		\Big(2r_ir_k + \frac{r_i}{4N_i} + \frac{r_k}{4N_k}\Big)\id_M &= \Om{ijk}\Om{ijk}^*,\nn\\ 
		\Big(2r_jr_l + \frac{r_j}{4N_j} + \frac{r_l}{4N_l}\Big)\id_M &= \Om{ijl}^*\Om{ijl}\label{eq:demand-6}.
	\ea

\item As was already covered in Section \ref{sec:bb4}, a building block \BBBB{} of the fourth type only breaks supersymmetry softly iff
\ba\label{eq:demand-7}
	r_1 &= \frac{1}{4}&&\text{and}& \w{1j}\yukw{j}{}\yukws{j}{} &= \Big(- \frac{1}{4} \pm \frac{\kappa_{1'}\kappa_{j}}{2}\Big)\id_M
\ea
(see Proposition \ref{cor:bb4}), where the latter should hold for each building block \B{11'j} of the third type. Here $\kappa_{1'}, \kappa_{j} \in \{\pm 1\}$.  

\item Covered in Section \ref{sec:bb5}, a building block \B{\textrm{mass}, ij} of the fifth type also breaks supersymmertry only softly iff
\ba\label{eq:demand-8}
	 \w{ij} &= \frac{1}{2},
\ea
see Proposition \ref{cor:bb5}. 

\end{enumerate}

To be able to say whether an almost-commutative geometry that is built out of building blocks of the first to the fifth type has a supersymmetric action then entails checking whether all the relevant relations above are satisfied. 

\subsection{Applied to a single building block of the third type}

We apply a number of the demands above to the case of a single building block of the third type (and the building blocks of the second and first type that are needed to define it) to see whether this possibly exhibits supersymmetry. We will assume that $\fer{ij}$ has $R = - 1$ (and consequently no family index), but of course we could equally well have taken one of the other two (see e.g.~Remark \ref{rmk:bb3-R=1}). 
The generalization of Remark \ref{rmk:bb2-rmk} for the expressions of the $r_i$ that results from normalizing the gauge bosons' kinetic terms is
\bas
	r_i &= \frac{3}{2N_i + N_j + MN_k}, &
	r_j &= \frac{3}{N_i + 2N_j + MN_k}, &
	r_k &= \frac{3}{M(N_i + N_j) + 2N_k}.
\eas

For the first of the demands of the previous section, \eqref{eq:demand-0}, one of the three terms that are equated to each other reads
\bas
	\w{ik}\yukw{i}{k}\yukws{i}{k} &\equiv \w{ik}(N_j\yuk{i}{k}\yuks{i}{k})^{-1/2}\yuk{i}{k}\yuks{i}{k}(N_j\yuk{i}{k}\yuks{i}{k})^{-1/2} = \frac{\w{ik}}{N_j}\id_M = \w{ik}\yukws{i}{k}\yukw{i}{k},
\eas
where we have used the definition \eqref{eq:def-yukw} of $\yukw{i}{k}$. Similarly,
\bas
	\w{jk}\yukws{j}{k}\yukw{j}{k} &= \frac{\w{jk}}{N_i}\id_M && 
\text{and} &
	\w{ij}\yukws{i}{j}\yukw{i}{j} &= \frac{\w{ij}}{N_k}\yuks{i}{j}\yuk{i}{j} (\tr \yuks{i}{j}\yuk{i}{j})^{-1}
\eas
for the other two. Equating these, we obtain: 
\ba\label{eq:idn-bb3-sols}
 \frac{\w{ik}}{N_j}\id_M &= \frac{\w{jk}}{N_i}\id_M = \frac{\w{ij}}{N_k}\yuks{i}{j}\yuk{i}{j} (\tr \yuks{i}{j}\yuk{i}{j})^{-1},
\ea
i.e.~$\yuk{i}{j}$ is constrained to be proportional to a unitary matrix. Taking the trace gives the demand
\ba\label{eq:demand-0-bb3}
 M\frac{\w{ik}}{N_j} &= M\frac{\w{jk}}{N_i} = \frac{\w{ij}}{N_k}.
\ea
Given the expressions for $r_{i,j,k}$ above, we can test whether this demand admits solutions. Indeed, we find
\ba\label{eq:solutions-bb3}
	N_i &= N_j = N_k \equiv N,& M &= 1 \lor 2.
\ea
In the first case we find that 
	\bas
		r_iN_i &= r_jN_j = r_k N_k = \frac{3}{4},& \w{ij} &= \w{ik} = \w{jk} = - \frac{1}{2}, 
	\eas 
	whereas in the second case we have 
	\bas
		r_iN_i &= r_jN_j = \frac{3}{5},& r_kN_k &= \frac{1}{2},& \w{ij} &= - \frac{1}{5},& \w{ik} &= \w{jk} = - \frac{1}{10}.
	\eas

Next, we have the demand \eqref{eq:demand-1} to ensure that terms of the form $|\sfer_{ij}\asfer_{ij}|^2$ can be written off shell in a supersymmetric manner. In this context it reads 
\bas
		\frac{r_i}{4} &= N_ir_i^2 + \alpha_{ij} N_k \w{ij}^2\tr[(\yukws{i}{j}\yukw{i}{j})^2],\nn\\
  	\frac{r_j}{4} &= N_jr_j^2 + \alpha_{ji} N_k \w{ij}^2\tr[(\yukws{i}{j}\yukw{i}{j})^2],\nn
\eas
for $\sfer_{ij}$ (where the trace in the last term comes from the fact that $\sfer_{ij}$ does not have family indices) and 
\bas
		\frac{r_k}{4}\id_M &= N_kr_k^2\id_M + \alpha_{kj} N_i \w{jk}^2(\yukws{j}{k}\yukw{j}{k})^2, \nn\\
		\frac{r_j}{4}\id_M &= N_jr_j^2\id_M + \alpha_{jk} N_i \w{jk}^2(\yukws{j}{k}\yukw{j}{k})^2,\nn\\
		\frac{r_k}{4}\id_M &= N_kr_k^2\id_M + \alpha_{ki} N_j \w{ik}^2(\yukw{i}{k}\yukws{i}{k})^2,\nn\\
		\frac{r_i}{4}\id_M &= N_ir_i^2\id_M + \alpha_{ik} N_j \w{ik}^2(\yukw{i}{k}\yukws{i}{k})^2,
\eas
for $\sfer_{jk}$ and $\sfer_{ik}$ respectively. Here we have written $\alpha_{ji} = 1 - \alpha_{ij}$, etc. We can remove all variables $\yukw{i}{j}$, $\yukw{i}{k}$ and $\yukw{j}{k}$ by using the squares of the expressions in \eqref{eq:idn-bb3-sols}. This gives
\bas
		\frac{N_ir_i}{4} &= (N_ir_i)^2 + \alpha_{ij} N_k \frac{\w{jk}^2}{N_i}M, &
  	\frac{N_jr_j}{4} &= (N_jr_j)^2 + \alpha_{ji} N_k \frac{\w{ik}^2}{N_j}M,\nn\\
		\frac{N_kr_k}{4} &= (N_kr_k)^2 + \alpha_{kj} N_k \frac{\w{jk}^2}{N_i}, &
		\frac{N_jr_j}{4} &= (N_jr_j)^2 + \alpha_{jk} N_i \frac{\w{ik}^2}{N_j},\nn\\
		\frac{N_kr_k}{4} &= (N_kr_k)^2 + \alpha_{ki} N_k \frac{\w{ik}^2}{N_j},&
		\frac{N_ir_i}{4} &= (N_ir_i)^2 + \alpha_{ik} N_j \frac{\w{jk}^2}{N_i},
\eas
where the $M$ in the first line above comes from taking the trace over $\id_M$. Comparing the expressions featuring the same combinations $r_{i}N_{i}$, $r_{j}N_{j}$, $r_{k}N_{k}$ and using \eqref{eq:demand-0-bb3} we must have that
\bas
	\alpha_{ij}N_kM &= \alpha_{ik}N_j,&
	(1- \alpha_{jk})N_i &= (1 - \alpha_{ik})N_j,&
	(1 - \alpha_{ij})N_kM &= \alpha_{jk}N_i.
\eas
Since both solutions \eqref{eq:solutions-bb3} to the relation \eqref{eq:demand-0-bb3} have $N_i = N_j = N_k$, this solves 
\bas
	\alpha_{ij} &= \frac{1}{2},& \alpha_{ik} &= \frac{1}{2}M, & \alpha_{jk} &= \frac{1}{2}M
\eas
and the demands above reduce to
\bas
		N_ir_i &=4(N_ir_i)^2 + 2 \w{jk}^2M, &
  	N_jr_j &=4(N_jr_j)^2 + 2 \w{ik}^2M, &
		N_kr_k &=4(N_kr_k)^2 + \w{ik}^2(4 - 2M).
\eas
We can check that for neither of the two cases of \eqref{eq:solutions-bb3} these are satisfied. As a cross check of this result we will employ one more demand.

In the context of a single building block of the third type the demand \eqref{eq:demand-2} that is necessary to write terms of the form $|\sfer_{ij}\sfer_{jk}|^2$ off shell in a supersymmetric manner, reduces to
\bas
	 2(1 - 2\w{ik})\w{ik} &= r_jN_j ,&
	 2(1 - 2\w{jk})\w{jk} &= r_iN_i ,&
	 2(1 - 2\w{ij})\w{ij}\yuks{i}{j}\yuk{i}{j}  &= r_kN_k\id_M \tr \yuks{i}{j}\yuk{i}{j}.
\eas
We can use \eqref{eq:demand-0-bb3} to rewrite the last equation in terms of $\w{ik}$ or $\w{jk}$. In any way, the LHS are seen to be negative for all values of $\w{ij}$, $\w{ik}$ and $\w{jk}$ allowed by the solutions \eqref{eq:solutions-bb3}, whereas $r_{i}N_{i}$, $r_{j}N_{j}$ and $r_{k}N_{k}$ are necessarily positive. We thus get a contradiction. 

A single building block of the third type (together with the building blocks needed to define it) is thus not supersymmetric.



\section{Summary and conclusions}

The main subject of this paper are almost-commutative geometries of the form 
	\bas
		(C^{\infty}(M, \A_F), L^2(M, S\otimes \H_F), \dirac \otimes 1 + \gamma_5 \otimes D_F; \gamma_5\otimes \gamma_F, J_M\otimes J_F)
	\eas 
 of KO-dimension $2$ on a flat, $4$-dimensional background $M$. We have dressed these with a grading $R : \H \to \H$ called \emph{$R$-parity}. We have shown that such almost-commutative geometries provide an arena suited for describing fields theories that have a supersymmetric particle content. This was done by identifying five different \emph{building blocks}; constituents of a finite spectral triple that yield an almost-commutative geometry whose particle content has an equal number of (off shell) fermionic and bosonic degrees of freedom. In addition they contain the right interactions to make them eligible for supersymmetric theories. These five building blocks are listed in Table \ref{tab:bbs}.\\

\begin{table}[ht]
	\setlength{\extrarowheight}{3pt}
\begin{tabularx}{\textwidth}{XlllX}
\toprule	
	& \textbf{Building block} & \textbf{Required} &  \textbf{Counterpart in superfield formalism} & \\ 
\midrule
	& \B{i}	(\S\ref{sec:bb1})									&	 ---													& Vector multiplet & \\
	& \BB{ij}	(\S\ref{sec:bb2})								&	 \B{i}, \B{j}									& Chiral multiplet & \\
	& \BBB{ijk}	(\S\ref{sec:bb3})							&	 \BBB{ij}, \BBB{ik}, \BBB{jk}	& Superpotential with three chiral superfields & \\
	& \BBBB{11'}	(\S\ref{sec:bb4})						&	 \B{11'}											& Majorana mass for $\fer{11'}$, $\sfer_{11'}$ & \\
	& \B{\mathrm{mass}, ij} (\S\ref{sec:bb5})	& \Bc{ij}{+}, \Bc{ij}{-}				& A mass(-like) term for $\fer{ij}, \sfer_{ij}$& \\
\bottomrule	
\end{tabularx}
\caption{The building blocks of a supersymmetric spectral triple. In the last column we have listed their counterparts in the superfield formalism.}
\label{tab:bbs}
\end{table}

Although we have not been using the notion of superspace and superfields, the building blocks themselves can thus be seen as an alternative. However, a significant difference between the two approaches is that if a certain superfield enters the action, then automatically all its component fields do too. For the components of these building blocks this need not be true; without \emph{demanding} supersymmetry we are free to e.g.~define a finite Hilbert space consisting of only the representation \rep{i}{j} (and its conjugate), without its superpartner arising from a component of the finite Dirac operator. However, the philosophy to include each component of $D_F$ that is not explicitly forbidden by the demands on a spectral triple turned out to be a fruitful one in obtaining models that have a supersymmetric particle content, as long as we start by adding gauginos to the finite Hilbert space.\\

It is far from automatic, though, that when the field content is supersymmetric also the action is. First of all, there is a number of obstructions to a supersymmetric action:
\begin{enumerate}
	\item A single building block \B{i} of the first type (i.e.~without a building block \B{ij} of the second type, for some $j$) for which $N_i = 1$, has vanishing bosonic interactions (Remark \ref{rmk:bb1-obstr}).
	\item A single building block \B{ij} of the second type that has $R = -1$, has two different $u(1)$ gauge fields that interact whereas the corresponding gauginos do not (Remark \ref{rmk:bb2-obstr}).
	\item If the finite algebra contains more than two components $M_{N_i}(\com)$, $M_{N_j}(\com)$ and $M_{N_k}(\com)$ over $\com$ and there is a set of two or more building blocks \B{ij}, \B{ik} that share three different indices, then there are two different $u(1)$ gauge fields that interact, whereas the corresponding gauginos do not (Proposition \ref{prop:2bb2-obstr}).
\end{enumerate}
Second, for a set up that avoids these three obstructions, the question is whether the four-scalar interactions that are generated by the spectral action are rewritable as an off shell action in terms of the auxiliary fields that are available to us. On top of this, the pre-factors of the interactions with the auxiliary fields are dictated by supersymmetry. Both the form of the action functional used in noncommutative geometry and supersymmetry thus put demands on the pre-factors of interactions which together heavily constrain the number of possible solutions. Typical for almost-commutative geometries is that there are new contributions to various expressions when extending a model. The question whether for the `full theory' the coefficients are such that these terms do have an off shell counterpart, is then phrased in terms of the demands listed in Section \ref{sec:4s-aux}.\\ 

%
%


Despite all these technical calculations and detailed issues, we have a definite handle on which almost-commutative geometries exhibit a supersymmetric action and which do not. To obtain an exhaustive list of examples that do satisfy all demands requires an automated strategy, in which step by step models are extended with building blocks and it is checked whether they satisfy the aforementioned demands. Whatever the outcome of such a strategy will be, the examples of supersymmetric almost-commutative geometries will be sparse. This is markedly different from the more generic superfield formalism, but at the same time the models that do satisfy all demands will enjoy a very special status.\\



This paper covered a particular class of spectral triples that is of direct interest for model building in particle physics; almost-commutative geometries whose background was flat and four dimensional and whose finite spectral triple is of KO-dimension $6$. We have restricted ourselves to theories with one supersymmetry charge (i.e.~$N = 1$ supersymmetry). Similar analyses can of course be done for manifolds with other dimensions than $4$, manifolds that are not flat, theories with $N =2$ and $N = 4$ supersymmetry and finite spectral triples of different KO-dimensions.  \\

Although we have encountered a couple of interactions that break supersymmetry softly, we have no thorough analysis of this phenomenon yet. This will be the subject of the second part in this series of papers.

\section*{Acknowledgements}
The authors would like to thank John Barrett for giving useful comments. One of the authors would like to thank the Dutch Foundation for Fundamental Research on Matter (FOM) for funding this work.
\appendix

\section{The action from a building block of the third type}\label{sec:bb3-calc-action}

In this section we derive in detail the action that comes from a building block \B{ijk} of the third type (cf.~Section \ref{sec:bb3}), such as that of Figure \ref{fig:bb3}. If we constrain ourselves for now to the off-diagonal part of the finite Hilbert space, then on the basis 
\bas
	\H_{F,\mathrm{off}} &= (\rep{i}{j})_L\, \oplus\, (\rep{i}{k})_R\, \oplus\, (\rep{j}{k})_L\nn\\
		&\qquad  \oplus\, (\rep{j}{i})_R\, \oplus\, (\rep{k}{i})_L\, \oplus\, (\rep{k}{j})_R
\eas
the most general allowed finite Dirac operator is of the form
\ba
	D_F= & \begin{pmatrix}
			0 &   \yuks{j}{k\,o}& 0 & 0 & 0 &   \yuks{i}{k}\\
			  \yuk{j}{k\,o}	& 0 &   \yuk{i}{j} & 0 & 0 & 0 \\
			0 &   \yuks{i}{j} & 0  &   \yuks{i}{k\,o}& 0 & 0 \\
			0 & 0 &   \yuk{i}{k\,o}& 0 &   \yuk{j}{k} & 0 \\	
			0 & 0 & 0 &   \yuks{j}{k} & 0 &   \yuks{i}{j\,o}\\	
			  \yuk{i}{k} & 0 & 0 & 0 &   \yuk{i}{j\,o}& 0 
	\end{pmatrix}\label{eq:bb3-DF}
\ea
We write for a generic element $\zeta$ of $\frac{1}{2} (1 + \gamma) L^2(S \otimes \H_{F, \mathrm{off}})$ 
\bas
	\zeta = (\fer{ijL}, \fer{ikR}, \fer{jkL}, \afer{ijR}, \afer{ikL}, \afer{jkR})
\eas
where $\afer{ijR} \in L^2( S_- \otimes \rep{j}{i})$, etc. Applying the matrix \eqref{eq:bb3-DF} to this element yields
\bas
	\gamma^5 D_F \zeta &= \gamma^5\Big(\fer{ikR}\asfer_{jk}\yuks{j}{k} + \yuks{i}{k}\sfer_{ik}\afer{jkR}, \fer{ijL}\yuk{j}{k}\sfer_{jk} + \yuk{i}{j}\sfer_{ij}\fer{jkL}, \\
						&\qquad \asfer_{ij}\yuks{i}{j}\fer{ikR} + \afer{ijR}\yuks{i}{k}\sfer_{ik}, \fer{jkL}\asfer_{ik}\yuk{i}{k} + \yuk{j}{k}\sfer_{jk}\afer{ikL}, \\
				&\qquad \asfer_{jk}\yuks{j}{k}\afer{ijR} + \afer{jkR}\asfer_{ij}\yuks{i}{j}, \afer{ikL}\yuk{i}{j}\sfer_{ij} + \asfer_{ik}\yuk{i}{k}\fer{ijL}\Big).
\eas
Notice that for the pairs $(i,j)$ and $(j,k)$ we always encounter $\sfer_{ij}$ in combination with $\yuk{i}{j}$, whereas for $(i,k)$ it is the combination $\sfer_{ik}$ and $\yuks{i}{k}$. This has to do with the fact that the sfermion $\sfer_{ik}$ crosses the particle/antiparticle-diagonal in the Krajewski diagram. Since 
\bas
	J\zeta &= J(\fer{ijL}, \fer{ikR}, \fer{jkL}, \afer{ijR}, \afer{ikL}, \afer{jkR}) \\
				&= (J_M \afer{ijR}, J_M\afer{ikL}, J_M\afer{jkR}, J_M\fer{ijL}, J_M\fer{ikR}, J_M\fer{jkL}),
\eas
the extra contributions to the inner product are written as
\bas
	&\frac{1}{2}\inpr{J\zeta}{\gamma^5 D_F\zeta} \nn\\ &= \frac{1}{2}\inpr{J_M \afer{ijR}}{\gamma^5(\fer{ikR}\asfer_{jk}\yuks{j}{k} + \sfer_{ik}\yuks{i}{k}\afer{jkR})} + 
	\frac{1}{2}\inpr{J_M\afer{ikL}}{\gamma^5(\fer{ijL}\yuk{j}{k}\sfer_{jk} + \yuk{i}{j}\sfer_{ij}\fer{jkL})} \\ 
&\qquad + \frac{1}{2}\inpr{J_M\afer{jkR}}{\gamma^5(\asfer_{ij}\yuks{i}{j}\fer{ikR} + \afer{ijR}\yuks{i}{k}\sfer_{ik})} + 
\frac{1}{2}\inpr{J_M\fer{ijL}}{\gamma^5(\fer{jkL}\asfer_{ik}\yuk{i}{k} + \yuk{j}{k}\sfer_{jk}\afer{ikL})} \\ 
&\qquad  + \frac{1}{2}\inpr{J_M\fer{ikR}}{\gamma^5(\asfer_{jk}\yuks{j}{k}\afer{ijR} + \afer{jkR}\asfer_{ij}\yuks{i}{j})} + 
\frac{1}{2}\inpr{J_M\fer{jkL}}{\gamma^5(\afer{ikL}\yuk{i}{j}\sfer_{ij} + \asfer_{ik}\yuk{i}{k}\fer{ijL})}.
\eas
Using the symmetry properties \eqref{eq:identitySymJ} of the inner product, this equals
\bas
&\inpr{J_M \afer{ijR}}{\gamma^5\fer{ikR}\asfer_{jk}\yuks{j}{k}} + 
\inpr{J_M \afer{ijR}}{\gamma^5\yuks{i}{k}\sfer_{ik}\afer{jkR}} + 
\inpr{J_M\afer{ikL}}{\gamma^5\fer{ijL}\yuk{j}{k}\sfer_{jk}} \\ 
&\qquad + \inpr{J_M\afer{ikL}}{\gamma^5\yuk{i}{j}\sfer_{ij}\fer{jkL}} + 
\inpr{J_M\afer{jkR}}{\gamma^5\asfer_{ij}\yuks{i}{j}\fer{ikR}} + 
\inpr{J_M\fer{jkL}}{\gamma^5\asfer_{ik}\yuk{i}{k}\fer{ijL}}.
\eas
We drop the subscripts $L$ and $R$, keeping in mind the chirality of each field, and for brevity we replace $ij \to 1$, $ik \to 2$, $jk \to 3$:
\ba
S_{123,F}[\zeta, \szeta] &= \inpr{J_M \afer{1}}{\gamma^5\fer{2}\asfer_{3}\yuks{3}{}} + 
\inpr{J_M \afer{1}}{\gamma^5\yuks{2}{}\sfer_{2}\afer{3}} + 
\inpr{J_M\afer{2}}{\gamma^5\fer{1}\yuk{3}{}\sfer_{3}} \nn\\ 
&\qquad + \inpr{J_M\afer{2}}{\gamma^5\yuk{1}{}\sfer_{1}\fer{3}} + 
\inpr{J_M\afer{3}}{\gamma^5\asfer_{1}\yuks{1}{}\fer{2}} + 
\inpr{J_M\fer{3}}{\gamma^5\asfer_{2}\yuk{2}{}\fer{1}}\label{eq:bb3-action-ferm-detail}.
\ea

The spectral action gives rise to some new interactions compared to those coming from building blocks of the second type. They arise from the trace of the fourth power of the finite Dirac operator and are given by the following list.

\begin{figure}
\begin{center}
	\def\svgwidth{.8\textwidth}
	\subimport{\the\svgpath}{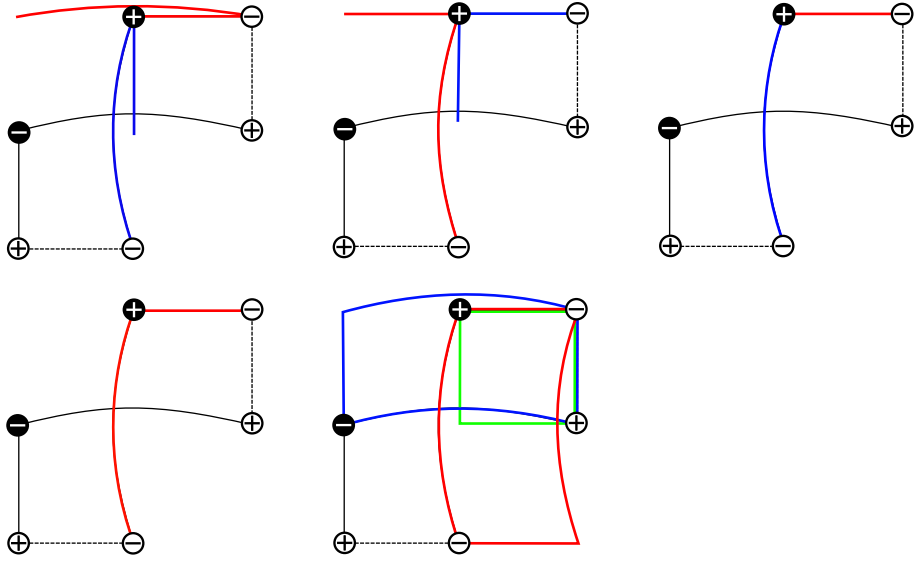}	
	\caption{The various contributions to $\tr D_F^4$ in the language of Krajewski diagrams corresponding to a building block \B{ijk} of the third type.}
	\label{fig:bb3-contributions}
\end{center}
\end{figure}

\begin{itemize}
	\item From paths of the type such as the one in the upper left corner of Figure \ref{fig:bb3-contributions} the contribution is 
	\ba
			& 8\Big[N_i|C_{iij}\sfer_{ij}\yuk{j}{k}\sfer_{jk}|^2  
			 + N_k|\yuk{i}{j}\sfer_{ij}C_{jkk}\sfer_{jk}|^2  
			 + N_j|\asfer_{ij}C_{ijj}^*\yuks{i}{k}\sfer_{ik}|^2\nn\\
			 &\qquad +  N_k|\asfer_{ij}\yuks{i}{j}C_{ikk}\sfer_{ik}|^2
		 + N_i|\yuk{j}{k}\sfer_{jk}\asfer_{ik}C_{iik}^*|^2	
		+ N_j|C_{jjk}\sfer_{jk}\asfer_{ik}\yuk{i}{k}|^2	\Big].\label{eq:bb3-action-1}
	\ea
	Here the multiplicity $8 = 2(1 + 1+ 2)$ comes from the fact that there are three vertices involved in each path, on each of which the path can start. In the case of the `middle' vertices the path can be traversed in two distinct orders. Furthermore a factor two comes from that each path occurs twice; also mirrored along the diagonal of the diagram.
	\item From paths such as the upper middle one in Figure \ref{fig:bb3-contributions} the contribution is:
	\ba
			& 8\Big[ \tr (C_{iij}\sfer_{ij})^o \yuk{j}{k}\sfer_{jk}\asfer_{jk}\yuks{j}{k}(\asfer_{ij}C_{iij}^*)^o
					+ \tr(\yuk{i}{j}\sfer_{ij})^oC_{jjk}\sfer_{jk}\asfer_{jk}C_{jjk}^*(\asfer_{ij}\yuks{i}{j})^o\nn\\
			&\qquad		+ \tr(\asfer_{ij}C_{iij}^*)^o\yuks{i}{k}\sfer_{ik}\asfer_{ik}\yuk{i}{k}(C_{iij}\sfer_{ij})^o
			 +  \tr(\asfer_{ij}\yuks{i}{j})^oC_{iik}\sfer_{ik}\asfer_{ik}C_{iik}^*(\yuk{i}{j}\sfer_{ij})^o\nn\\
			&\qquad +  \tr(\asfer_{ik}C_{ikk}^*)^o\yuk{j}{k}\sfer_{jk}\asfer_{jk}\yuks{j}{k}(C_{ikk}\sfer_{ik})^o
			  +  \tr(\asfer_{ik}\yuk{i}{k})^oC_{jjk}\sfer_{jk}\asfer_{jk}C_{jkk}^*(\yuks{i}{k}\sfer_{ik})^o	\Big],\label{eq:bb3-action-2}
	\ea
		where the arguments for determining the multiplicity are the same as for the previous contribution.
	\item From paths such as the upper right one in Figure \ref{fig:bb3-contributions}, going back and forth along the same edge twice, the contribution is:
	\ba
			4\Big[N_i|\yuk{j}{k}\sfer_{jk}\asfer_{jk}\yuks{j}{k}|^2 + N_j|\yuks{i}{k}\sfer_{ik}\asfer_{ik}\yuk{i}{k}|^2 + N_k|\yuk{i}{j}\sfer_{ij}\asfer_{ij}\yuks{i}{j}|^2\Big] \label{eq:bb3-action-3}
	\ea
	The multiplicity arises from $2$ vertices on which the path can start and each such path occurs again reflected.

	\item From paths such as the lower left one in Figure \ref{fig:bb3-contributions} the contribution is:
	\ba
		& 8 \Big[|\sfer_{ij}|^2|\yuk{i}{j}\yuks{i}{k}\sfer_{ik}|^2 + |\sfer_{ij}|^2|\yuks{i}{j}\yuk{j}{k}\sfer_{jk}|^2 + |\yuks{i}{k}\sfer_{ik}|^2|\yuk{j}{k}\sfer_{jk}|^2 \Big]. \label{eq:bb3-action-4}	
	\ea

	\item From paths such as the lower right one in Figure \ref{fig:bb3-contributions} the contribution is:
	\ba
			& 8 \Big[ \tr (\asfer_{ik}C_{iik}^*(\yuk{i}{j}\sfer_{ij})^o(\asfer_{ij} C_{iij}^*)^o \yuks{i}{k}\sfer_{ik})  + \tr(\asfer_{jk}\yuks{j}{k}(\asfer_{ij}C_{ijj}^*)^o(\yuk{i}{j}\sfer_{ij})^oC_{jjk}\sfer_{jk}) \nn\\
		&\qquad + \tr((\asfer_{ik}\yuk{i}{k})^oC_{jkk}\sfer_{jk}\asfer_{jk}\yuks{j}{k}(C_{ikk}\sfer_{ik})^o)  + h.c. \Big],\label{eq:bb3-action-5}
	\ea
	corresponding with the blue, green and red paths respectively. The multiplicity arises from the fact that any such path has four vertices on which it can start and also occurs reflected around the diagonal. Besides, each path can also be traversed in the opposite direction, hence the `h.c.'.
\end{itemize}
Adding \eqref{eq:bb3-action-1}, \eqref{eq:bb3-action-2}, \eqref{eq:bb3-action-3}, \eqref{eq:bb3-action-4} and \eqref{eq:bb3-action-5} the total \emph{extra} contribution to $\tr D_F^4$ from adding a building block \B{ijk} of the third type, is given by \eqref{eq:bb3-boson-action}.

\section{Proofs}\label{sec:AppSusySTAction}

In this section we give the actual proofs and calculations of the Lemmas and Theorems presented in the text. First we introduce some notation. With $(.,.)_{\mathcal{S}} : \Gamma^{\infty}(\mathcal{S}) \times \Gamma^{\infty}(\mathcal{S}) \to C^{\infty}(M)$ we mean the $C^{\infty}(M)$-valued Hermitian structure on $\Gamma^{\infty}(\mathcal{S})$. 
The Hermitian form on $\Gamma^{\infty}(\mathcal{S})$ is to be distinguished from the $C^{\infty}(M)$-valued form on $\H \equiv L^2(M, S\otimes \H_F)$:
\begin{align*}
		(., .)_{\H} : \Gamma(\mathcal{S} \otimes \H_F) \times \Gamma(\mathcal{S} \otimes \H_F) \to C^{\infty}(M)
\end{align*}
given by
\begin{align*}
		(\fer{1}, \fer{2})_{\H}  := (\zeta_1,\zeta_2)_{\mathcal{S}}\langle m_1, m_2 \rangle_{F},\qquad \psi_{1,2} = \zeta_{1,2} \otimes m_{1,2},
\end{align*}
where $\inpr{\,.\,}{\,.\,}_F$ denotes the inner product on the finite Hilbert space $\H_F$. The inner product on the full Hilbert space $\H$ is then obtained by integrating over the manifold $M$: 
\begin{align*}
	\langle \psi_1, \psi_2\rangle_{\H} := \int_M  (\psi_1, \psi_2)_{\H}\,\sqrt{g}\mathrm{d}^4x.
\end{align*}
If no confusion is likely to arise between $(.,.)_{\mathcal{S}}$ and $(.,.)_{\H}$, we omit the subscript.

In the proofs there appear a number of a priori unknown constants. To avoid confusion: capital letters always refer to parameters of the Dirac operator, lowercase letters always refer to proportionality constants for the superfield transformations. For the latter the number of indices determines what field they belong to: constants with one index belong to a gauge boson--gaugino pair, constants with two indices belong to a fermion--sfermion pair. 

\subsection{First building block}\label{sec:SYM}

This section forms the proof of Theorem \ref{prop:bb1}. In this case the action is given by \eqref{eq:SYM}. Its constituents are the ---flat--- metric metric $g$, the gauge field $A^j \in \End(\Gamma(\cS) \otimes su(N_j))$ and spinor $\gau{j} \in L^2(M, S \otimes su(N_j))$, both in the adjoint representation and the spinor after reducing its degrees of freedom (see Section {\ref{sec:equalizing}}).

Now for $\epsilon \equiv (\eL, \eR) \in L^2(M, S)$, decomposed into Weyl spinors that vanish covariantly (i.e.~$\nabla^S\epsilon = 0$), we define 
\begin{subequations}\label{eq:susytransforms2}
\begin{align}
		\delta A_j &= c_{j}\gamma^\mu\big[(J_M\eR, \gamma_\mu\gau{jL})_\cS +  (J_M\eL, \gamma_\mu\gau{jR})_\cS\big] \equiv \gamma^\mu (\delta A_{\mu j +} + \delta A_{\mu j -})\label{eq:transforms2.1},\\
		\delta \gau{jL,R} &= (c_{j}' F^j + c_{G_j}'G_j)\eLR,\qquad F^j  \equiv \gamma^\mu\gamma^\nu F_{\mu\nu}^{j}\label{eq:transforms2.2},\\ 
		\delta G_j &= c_{G_j}\big[(J_M\eL, \can_A\gau{jR})_{\cS} + (J_M\eR, \can_A\gau{jL})_{\cS}\big]\label{eq:transforms2.3},
\end{align}
\end{subequations}
where the coefficients $c_j, c_j', c_{G_j}, c_{G_j}'$ are yet to be determined. In the rest of this section we will drop the index $j$ for notational convenience and discard the factor $n_j$ from the normalization of the gauge group generators, since it appears in the same way for each term. 


\begin{itemize}
\item The fermionic part of the Lagrangian, upon transforming the fields, equals:
\begin{align}
	\langle J_M\gau{L}, \can_A\gau{R}\rangle &\to \int_M (J_M[c'F + c_{G}'G]\eL, \can_A\gau{R})_{\H} + (J_M\gau{L}, \can_A[c'F + c_{G}'G]\eR)_{\H}\nn\\
		&\qquad + gc(J_M\gau{L}, \gamma^\mu\ad[(J_M\eL, \gamma_\mu\gau{R})_{\cS} + (J_M\eR, \gamma_\mu\gau{L})_{\cS}]\gau{R})_{\H}\label{eq:bb1-transf1}.
\end{align}
Here we mean with $\ad(X)$ the adjoint: $\ad(X)Y := [X, Y]$.

\item The kinetic terms for the gauge bosons transform to:
\begin{align}
 \frac{1}{4}\K \int_M \tr_N F^{\mu\nu}F_{\mu\nu} &\to  c\frac{\K}{2} \int_M \tr_{N} F^{\mu\nu}\Big(\partial_{[\mu} \big[(J_M\eR, \gamma_{\nu]}\gau{L})_\cS +  (J_M\eL, \gamma_{\nu]}\gau{R})_\cS\big] \nn\\
	&\qquad -ig[(J_M\eR, \gamma_{\mu}\gau{L})_\cS +  (J_M\eL, \gamma_{\mu}\gau{R})_\cS, A_\nu]\nn\\
	&\qquad  -ig[A_\mu, (J_M\eR, \gamma_{\nu}\gau{L})_\cS +  (J_M\eL, \gamma_{\nu}\gau{R})_\cS]\Big)\sqrt{g}\mathrm{d}^4x\label{eq:bb1-transf2}.
\end{align}
where $A_{[\mu}B_{\nu]} \equiv A_\mu B_\nu - A_\nu B_\mu$.
\item And finally the term for the auxiliary fields transforms to
\begin{align}
	- \frac{1}{2}\int_M \tr_{N} G^2 &\to - c_{G}\int_M \tr_{N} G\big[ (J_M\eR, \can_A\gau{L})_{\cS} + (J_M\eL, \can_A\gau{R})_{\cS}\big].\label{eq:bb1-transf3}
\end{align}
\end{itemize}

If we collect the terms of \eqref{eq:bb1-transf1}, \eqref{eq:bb1-transf2} and \eqref{eq:bb1-transf3} containing the same field content, we get three groups of terms that separately need to vanish in order to have a supersymmetric theory. These groups are:
\begin{itemize}
\item one consisting of only one term with four fermionic fields (coming from the second line of \eqref{eq:bb1-transf1}):
\begin{align}
	gc(J_M\gau{L}, \gamma^\mu\ad(J_M\eL, \gamma_\mu\gau{R})_{\cS}\gau{R})_{\H}.\label{eq:bb1-group1}
\end{align}
There is a second such term with $\eL \to \eR$ and $\gau{R} \to \gau{L}$ that is obtained via $(J_M\eL, \gamma_\mu\gau{R})_{\cS} \to (J_M\eR, \gamma_\mu\gau{L})_{\cS}$.

\item one consisting of a gaugino and two or three gauge fields:
\begin{align}
	& \int_M \bigg[c' (J_M\gau{L}, \can_AF\eR)_{\H} + c\frac{\K}{2} \tr_{N} F^{\mu\nu}\Big(\partial_{[\mu}(J_M\eR, \gamma_{\nu]}\gau{L})_\cS  -ig\big[(J_M\eR, \gamma_\mu\gau{L})_\cS, A_\nu\big] \nn\\
	&\qquad-ig\big[A_\mu, (J_M\eR, \gamma_{\nu}\gau{L})_\cS\big]\Big)\bigg]\label{eq:bb1-group2}
\end{align}
featuring the third term of \eqref{eq:bb1-transf1} and the terms of \eqref{eq:bb1-transf2} featuring $\gau{L}$. There is another such group with $\eR \to \eL$ and $\gau{L} \to \gau{R}$ consisting of the first term of \eqref{eq:bb1-transf1} and the other terms of \eqref{eq:bb1-transf2}.

\item one consisting of the auxiliary field $G$, a gauge field and a gaugino:
\begin{align}
	 \int_M \Big[ c_{G}'(J_M\gau{L},\can_AG\eR)_{\H}- c_{G} \tr_{N} G(J_M\eR, \can_A\gau{L})_{\cS} \Big]\label{eq:bb1-group3}
\end{align}
featuring the second part of the third term of \eqref{eq:bb1-transf1} and the first term of \eqref{eq:bb1-transf3}. There is another such group with $\eR\to\eL$ and $\gau{L} \to \gau{R}$.
\end{itemize}
We will tackle each of these groups separately in the following Lemmas.


\begin{lem}\label{lem:bb1-group1}
The term \eqref{eq:bb1-group1} equals zero.
\end{lem}
\begin{proof}
Evaluating \eqref{eq:bb1-group1} point-wise, applying the finite inner product and using the normalization for the generators of the gauge group, yields up to a constant factor
\begin{align}
	f^{abc} (J_M\gau{L}^a, \gamma^\mu \gau{R}^b)_\cS(J_M\eL, \gamma_\mu \gau{R}^c)_\cS\label{eq:fierz_start}.
\end{align}
Here the $f^{abc}$ are the structure constants of the Lie algebra $SU(N)$. We employ a Fierz transformation (See Appendix \ref{sec:fierz}), using $C_{10} = - C_{14} = 4$, $C_{11} =  C_{13} = -2$, $C_{12} = 0$, to rewrite \eqref{eq:fierz_start} as
\begin{align*}
	& f^{abc} (J_M\gau{L}^a, \gamma^\mu \gau{R}^b)_\cS(J_M\eL, \gamma_\mu \gau{R}^c)_\cS = - \frac{1}{4}f^{abc}\Big[4(J_M\eL, \gau{R}^b)_\cS(J_M\gau{L}^a, \gau{R}^c)_\cS\\
&\qquad\qquad - 2 (J_M\eL, \gamma_\mu \gau{R}^b)_\cS					(J_M\gau{L}^a, \gamma^\mu \gau{R}^c)_\cS  - 2 (J_M\eL, \gamma_\mu \gamma^5 \gau{R}^b)_\cS(J_M\gau{L}^a, \gamma^\mu \gamma^5 \gau{R}^c)_\cS\\
&\qquad\qquad - 4 (J_M\eL, \gamma^5 \gau{R}^b)_\cS						(J_M\gau{L}^a, \gamma^5 \gau{R}^c)_\cS\Big].
\end{align*}
The first and last terms on the right hand side of this expression are seen to cancel each other, whereas the second and third term add. We retain 
\begin{align*}
	& f^{abc} (J_M\gau{L}^a, \gamma^\mu \gau{R}^b)_\cS(J_M\eL, \gamma_\mu \gau{R}^c)_\cS =  f^{abc} (J_M\gau{L}^a, \gamma^\mu \gau{R}^c)_\cS (J_M\eL, \gamma_\mu \gau{R}^b)_\cS	. 
\end{align*}
Since $f^{abc}$ is fully antisymmetric in its indices, this expression equals zero.
\end{proof}


\begin{lem}\label{lem:bb1-group2}
The term \eqref{eq:bb1-group2} equals zero if and only if
\begin{align}
	2ic' &= - c\K \label{bb1-constr2}.
\end{align}
\end{lem}
\begin{proof}
If we use that the spin connection is Hermitian and employ \eqref{eq:idnNablaS}, this yields:
\begin{align*}
	\partial_\mu \delta A_{\nu\, +} &= c (J_M\eR, \gamma_\nu \nabla^S_\mu\gau{L}).
\end{align*}
Here we have used that $[\nabla^S_\mu, J_M] = 0$, that we have a flat metric and that $\nabla^S\eLR = 0$. Now using that $A_\mu(J_M\eR, \gamma_\nu\gau{L})_\cS = (J_M\eR, A_\mu \gamma_\nu \gau{L})_\cS$ and inserting these results into the second part of \eqref{eq:bb1-group2} gives
\begin{align*}
 	&  c\frac{\mathcal{K}}{2}\int_M \tr_{N} F^{\mu\nu}(J_M\eR, D_{[\mu}\gamma_{\nu]}\gau{L})_{\cS},\quad D_\mu = \nabla^S_\mu - ig \ad(A_\mu).
\end{align*}
Using Lemma \ref{lem:pullScalar} and employing the antisymmetry of $F_{\mu\nu}$ we get
\begin{align*}
 	&  c\mathcal{K}\int_M (J_MF^{\mu\nu}\eR, D_{\mu}\gamma_{\nu}\gau{L})_{\H}. 
\end{align*}
We take the first term of \eqref{eq:bb1-group2}
and write out the expression $\can_AF = i\gamma^\mu D_\mu \gamma^\nu\gamma^\lambda F_{\nu\lambda}$. We can commute the $D_\mu$ through the $\gamma^\nu\gamma^\lambda$-combination since the metric is flat. Employing the identity
\begin{align}
	\gamma^\mu\gamma^\nu\gamma^\lambda = g^{\mu\nu}\gamma^\lambda + g^{\nu\lambda}\gamma^\mu - g^{\mu\lambda}\gamma^\nu +\epsilon^{\sigma\mu\nu\lambda}\gamma^5\gamma_\sigma
\end{align}
yields
\begin{align*}
\can_A F = i\big(2g^{\mu\nu}\gamma^\lambda + \epsilon^{\sigma\mu\nu\lambda}\gamma^5\gamma_\sigma) D_\mu F_{\nu\lambda}.
\end{align*}
Applying this operator to $\eR$ gives
\begin{align*}
\can_A F \eR = 2ig^{\mu\nu}\gamma^\lambda D_\mu F_{\nu\lambda} \eR = 2i\gamma_\lambda D_\mu F^{\mu\lambda} \eR, 
\end{align*}
for the other term cancels via the Bianchi identity and the fact that $\nabla^S\eR = 0$. 
With the above results, \eqref{eq:bb1-group2} is seen to be equal to 
\begin{align}
2ic'\langle J\gau{L}, \gamma_\nu D_\mu F^{\mu\nu} \eR\rangle + c\mathcal{K}\int_M (J_MF^{\mu\nu}\eR, D_{\mu}\gamma_{\nu}\gau{L})_{\H}.
\end{align}
Using the symmetry of the inner product, the result follows.
\end{proof}


\begin{lem}\label{lem:bb1-group3}
The term \eqref{eq:bb1-group3} equals zero iff 
\begin{align}
	c_G = - c_G'.\label{eq:bb1-constr3}
\end{align}
\end{lem}
\begin{proof}
Using the cyclicity of the trace, the symmetry property \eqref{eq:identitySymJ2} of the inner product and Lemma \ref{lem:pullScalar}, the second term of \eqref{eq:bb1-group3} can be rewritten to 
\begin{align*}
 c_{G} \int_M (J_M\gau{L}, \can_A G\eR)_{\H}
\end{align*}
from which the result immediately follows. 
\end{proof}


By combining the above three lemmas we can prove Theorem \ref{thm:bb1}:
\begin{prop}\label{prop:bb1}
	A spectral triple whose finite part consists of a building block of the first type (Def.~\ref{def:bb1}) has a supersymmetric action \eqref{eq:SYM} under the transformations \eqref{eq:susytransforms2} iff
	\begin{align*}
			2ic' &= - c\K, & c_G = - c_G'.
	\end{align*}
\end{prop}

\subsection{Second building block}\label{sec:bb2-proof}

We apply the transformations \eqref{eq:bb1-transforms2}, \eqref{eq:susytransforms4} and \eqref{eq:susytransforms5} to the terms in the action \emph{that appear for the first time}\footnote{We add this explicitly since we do not need the terms in the Yang-Mills action for together they were already supersymmetric.} as a result of the new content of the spectral triple, i.e.~\eqref{eq:bb2-action-offshell}. In the fermionic part of the action, the second and fourth terms transform under \eqref{eq:susytransforms4} to
\begin{align}
	 \langle J_M\afer{R}, \gamma^5\gau{iR}\Cw{i,j}\sfer \rangle 
	&\to 
	\langle J_M c_{ij}'^*\gamma^5[\can_A,\asfer]\epsilon_L, \gamma^5\gau{iR}\Cw{i,j}\sfer \rangle +
	\langle J_M d_{ij}'^*F_{ij}^*\epsilon_R, \gamma^5\gau{iR}\Cw{i,j}\sfer \rangle \nn\\
	&\qquad\qquad + c_i' \langle J_M\afer{R}, \gamma^5F_i\Cw{i,j}\sfer\epsilon_R \rangle 
+ c_{G_i}' \langle J_M\afer{R}, \gamma^5G_i\Cw{i,j}\sfer\epsilon_R \rangle\nn \\
		&\qquad\qquad  + \langle J_M\afer{R}, \gamma^5\gau{iR}\Cw{i,j}c_{ij}(J_M\epsilon_L, \gamma^5 \fer{L})\rangle \label{eq:bb2-transf1}
\end{align}
and
\begin{align}
	 \langle J_M\fer{L}, \gamma^5\asfer \Cw{i,j}^*\gau{iL}\rangle
	&\to 
c_{ij}'	\langle J_M\gamma^5[\can_A, \sfer]\epsilon_R, \gamma^5\asfer \Cw{i,j}^*\gau{iL}\rangle +
d_{ij}'\langle J_M F_{ij}\epsilon_L, \gamma^5\asfer \Cw{i,j}^*\gau{iL}\rangle \nn\\
&\qquad\qquad  
+ c_i'\langle J_M\fer{L}, \gamma^5\gamma^\mu\gamma^\nu\asfer \Cw{i,j}^*F_{i\mu\nu}\epsilon_L\rangle + c_{G_i}'\langle J_M\fer{L}, \gamma^5\asfer \Cw{i,j}^*G_i\epsilon_L\rangle\nn\\
	&\qquad\qquad+ \langle J_M\fer{L}, \gamma^5c_{ij}^*(J_M\epsilon_R, \gamma^5\afer{R}) \Cw{i,j}^* \gau{iL}\rangle\label{eq:bb2-transf2}
\end{align}
respectively. 
We omit the terms with $\gau{jL,R}$ instead of $\gau{iL,R}$; transformation of these yield essentially the same terms. For the kinetic term of the $R = 1$ fermions (the first term of \eqref{eq:bb2-action-ferm}) we have under the same transformations:
\begin{align}
	\langle J_M \afer{R}, \can_A \fer{L}\rangle &\to 
	\langle J_M c_{ij}'^*\gamma^5[\can_A, \asfer]\epsilon_L, \can_A \fer{L}\rangle + 
	g_ic_i\langle J_M \afer{R}, \gamma^\mu [ (J_M\epsilon_L, \gamma_\mu \gau{iR}) + (J_M\epsilon_R, \gamma_\mu \gau{iL})] \fer{L}\rangle \nn\\
	&\qquad + \langle J_M \afer{R}, \can_A c_{ij}'\gamma^5[\can_A \sfer]\eR\rangle 
 + \langle J_Md_{ij}'^*F_{ij}^*\epsilon_R, \can_A \fer{L}\rangle 
 + \langle J_M \afer{R}, \can_A d_{ij}'F_{ij}\epsilon_L\rangle. 
\label{eq:bb2-transf3}
\end{align}
As with the previous contributions to the action, we omit the terms $\delta A_j$ (instead of $\delta A_i$) for brevity. In the bosonic action, we have the kinetic terms of the sfermions, transforming to
\begin{align}
	\tr_{N_j} D^\mu\asfer  D_\mu\sfer & \to +i g_ic_i\tr_{N_j}\big( \asfer[ (J_M\epsilon_L, \gamma_\mu \gau{iR}) + (J_M\epsilon_R, \gamma_\mu \gau{iL})] D^\mu \sfer\big)\nn\\
		&\qquad -i g_ic_i\tr_{N_j}\big( D_\mu \asfer[ (J_M\epsilon_L, \gamma^\mu \gau{iR}) + (J_M\epsilon_R, \gamma^\mu \gau{iL})] \sfer\big)\nn\\
		&\qquad + \tr_{N_j}\big( D_\mu c_{ij}^*(J_M\epsilon_R, \gamma^5 \afer{R}) D^\mu \sfer\big)  
		+ \tr_{N_j}\big( D_\mu \asfer  D^\mu c_{ij}(J_M\epsilon_L, \gamma^5 \fer{L})\big)\label{eq:bb2-transf4}
\end{align}
(and terms with $\gau{j}$ instead of $\gau{i}$) and from the terms with the auxiliary fields we have 
\begin{align}
	\tr_{N_i} \P_{i}\sfer\asfer G_i & \to \tr_{N_i} \P_{i}c_{ij}(J_M\epsilon_L, \gamma^5\fer{L})\asfer G_i + \tr_{N_i} \P_{i}\sfer c_{ij}^*(J_M\epsilon_R, \gamma^5\afer{R}) G_i\nn \\
	&\qquad + c_{G_i}\tr_{N_i} \P_{i}\sfer\asfer[ (J_M\epsilon_L, \can_A\gau{iR}) + (J_M\epsilon_R, \can_A\gau{iL})]\label{eq:bb2-transf5}.
\end{align}
And finally we have the kinetic terms of the auxiliary fields $F_{ij}$, $F_{ij}^*$ that transform to 
\begin{align}
	\tr F_{ij}^*F_{ij} &\to \tr F_{ij}^*\Big[ d_{ij}(J_M\epsilon_R, \can_A\fer{L})_{S} + d_{ij,i}(J_M\epsilon_R, \gamma^5\gau{iR}\sfer)_{\cS} - d_{ij,j}(J_M\epsilon_R, \gamma^5\sfer\gau{jR})_{\cS}\Big] \nn\\
		&\qquad + \tr\Big[d_{ij}^*(J_M\epsilon_L, \can_A\afer{R})_{S}  + d_{ij,i}^*(J_M\epsilon_L, \gamma^5\asfer\gau{iL})_{\cS} - d_{ij,j}^*(J_M\epsilon_L, \gamma^5\gau{jL}\asfer)_{\cS}\Big]F_{ij},\label{eq:bb2-transf6}
\end{align}
where the traces are over $\mathbf{N}_j^{\oplus M}$. Analyzing the result of this, we can put them in groups of terms featuring the very same fields. Each of these groups should separately give zero in order to have a supersymmetric action. We have:
\begin{itemize}
\item Terms with four fermionic fields; the fifth term of \eqref{eq:bb2-transf1}, and part of the second term of \eqref{eq:bb2-transf3}: 
\begin{align}
	& \langle J_M\afer{R}, \gamma^5\gau{iR}\Cw{i,j}c_{ij}(J_M\epsilon_L, \gamma^5 \fer{L})\rangle + 
	 g_ic_i\langle J_M \afer{R}, \gamma^\mu(J_M\epsilon_L, \gamma_\mu \gau{iR}) \fer{L}\rangle \label{eq:bb2-group1}.
\end{align}
 The third term of \eqref{eq:bb2-transf2} and the other part of the second term of \eqref{eq:bb2-transf3} give a similar contribution but with $\epsilon_L \to \epsilon_R$, $\gau{iL} \to \gau{iR}$.

\item Terms with one gaugino and two sfermions, consisting of the first term of \eqref{eq:bb2-transf1}, part of the first and second terms of \eqref{eq:bb2-transf4}, and part of the third term of \eqref{eq:bb2-transf5}:
\begin{align}
	& \langle J_Mc_{ij}'^*\gamma^5[\can_A,\asfer]\epsilon_L, \gamma^5\gau{iR}\Cw{i,j}\sfer \rangle 
  +i g_ic_i\int \tr_{N_j}\big( \asfer (J_M\epsilon_L, \gamma_\mu \gau{iR}) D^\mu \sfer\big)\nn\\
		&\qquad -i g_ic_i \int \tr_{N_j}\big( D_\mu \asfer (J_M\epsilon_L, \gamma^\mu \gau{iR}) \sfer\big)
 - c_{G_i}\int \tr_{N_i} \P_i\sfer\asfer (J_M\epsilon_L, \can_A\gau{iR}) \label{eq:bb2-group2}.
\end{align}
The first term of \eqref{eq:bb2-transf2}, the other parts of the first and second terms of \eqref{eq:bb2-transf4} and the other part of the third term of \eqref{eq:bb2-transf5} give similar terms but with $\eL \to \eR$, $\gau{iR} \to \gau{iL}$.

\item Terms with two gauge fields, a fermion and a sfermion, consisting of the third term of \eqref{eq:bb2-transf1}, the third term of \eqref{eq:bb2-transf4} and the third term of \eqref{eq:bb2-transf3}:
\begin{align}
& c_i' \langle J_M\afer{R}, \gamma^5F_i\Cw{i,j}\sfer\epsilon_R \rangle 
		 + \int \tr_{N_j}\big( D_\mu c_{ij}^*(J_M\epsilon_R, \gamma^5 \afer{R}) D^\mu \sfer\big) \nn\\ 
	&\qquad	+  \langle J_M \afer{R}, \can_A \gamma^5c_{ij}'[\can_A, \sfer]\epsilon_R\rangle \label{eq:bb2-group3}
\end{align}
The fourth term of \eqref{eq:bb2-transf2}, the first term of \eqref{eq:bb2-transf3} and the fourth term of \eqref{eq:bb2-transf4} make up a similar group but with $\epsilon_R \to \epsilon_L$ and $\afer{R} \to \fer{L}$.

\item Terms with the auxiliary field $G_i$, consisting of the fourth term of \eqref{eq:bb2-transf1} and the second term of \eqref{eq:bb2-transf5}:
\begin{align}
 & c_{G_i}' \langle J_M\afer{R}, \gamma^5G_i\Cw{i,j}\sfer\epsilon_R \rangle - \int \tr_{N_i} \P_i\sfer c_{ij}^*(J_M\epsilon_R, \gamma^5\afer{R}) G_i \label{eq:bb2-group4}
\end{align}
The fifth term of \eqref{eq:bb2-transf2} and the first term of \eqref{eq:bb2-transf5} make up another such group but with $\epsilon_R \to \epsilon_L$ and $\afer{R} \to \fer{L}$.

\item And finally all terms with either $F_{ij}$ or $F_{ij}^*$, consisting of the second term of \eqref{eq:bb2-transf1}, the second term of \eqref{eq:bb2-transf2}, the fourth and fifth terms of \eqref{eq:bb2-transf3} and the terms of \eqref{eq:bb2-transf6} (of which we have omitted the terms with $\gau{j}$ for now):
\begin{align}
	& \langle J_M d_{ij}'^*F_{ij}^*\epsilon_R, \gamma^5\gau{iR}\Cw{i,j}\sfer \rangle + \langle J_Md_{ij}'^*F_{ij}^*\epsilon_R, \can_A \fer{L}\rangle  \nn\\
 &\qquad - \int \tr_{N_j} F_{ij}^*\big[ d_{ij}(J_M\epsilon_R, \can_A\fer{L})_{S} + d_{ij,i}(J_M\epsilon_R, \gamma^5\gau{iR}\sfer)_{\cS}\big] \label{eq:bb2-group5}
\end{align}
and 
\begin{align*}
& \langle J_M F_{ij}d_{ij}' \epsilon_L, \gamma^5\asfer \Cw{i,j}^*\gau{iL}\rangle + \langle J_M \afer{R}, \can_A d_{ij}'F_{ij}\epsilon_L\rangle   \\
		&\qquad - \int \tr_{N_j}\big[d_{ij}^*(J_M\epsilon_L, \can_A\afer{R})_{S}  + d_{ij,i}^*(J_M\epsilon_L, \gamma^5\asfer\gau{iL})_{\cS} \big]F_{ij}.
\end{align*}

\end{itemize}


We will tackle each of these five groups in the next five lemmas. For the first group we have:
\begin{lem}\label{lem:bb2-group1}
	The expression \eqref{eq:bb2-group1} vanishes, provided that
	\begin{align}
\frac{1}{2}\Cw{i,j} c_{ij} &= - c_{i}g_i\label{eq:bb2-group1-constr}
	\end{align}
\end{lem}
\begin{proof}
	Since the expression contains only fermionic terms, we need to prove this via a Fierz transformation, which is valid only point-wise. We will write
\begin{align*}
	\gau{i} &= \gau{}^a \otimes T^a \in L^2(S_- \otimes su(N_i)_R),\nn\\
	 \fer{L} &= \fer{mn} \otimes e_{i,m}\otimes \bar e_{j,n} \in L^2(S_+ \otimes \rep{i}{j}),\nn\\
	 \afer{R} &= \afer{rs} \otimes e_{j,r} \otimes \bar e_{i,s} \in L^2(S_- \otimes \rep{j}{i}),
\end{align*}
where a sum over $a$, $m$, $n$, $r$ and $s$ is implied, to avoid a clash of notation. Here the $T^a$ are the generators of $su(N_i)$. Using this notation, \eqref{eq:bb2-group1} is point-wise seen to be equivalent to
\begin{align*}
	& (J_M\afer{jk}, \gamma^5\gau{}^a)(J_M\epsilon_L, \gamma^5 \Cw{i,j}c_{ij}\fer{ij}) T^a_{ki}
	 + g_ic_i(J_M \afer{jk}, \gamma^\mu \fer{ij})(J_M\epsilon_L, \gamma_\mu \gau{}^a)T^a_{ki}.
\end{align*}
Since it appears in both expressions, we may simply omit $T^a_{ki}$ from our considerations. For brevity we will omit the subscripts of the fermions from here on. We then apply a Fierz transformation (see Appendix \ref{sec:fierz}) for the first term, giving:
\begin{align*}
&	( J_M\afer{}, \gamma^5\gau{}^a)(J_M\epsilon_L, \gamma^5 \fer{}) \nn\\
&= - \frac{C_{40}}{4}(J_M\afer{}, \fer{})(J_M\epsilon_L, \gau{}^a) - \frac{C_{41}}{4}(J_M\afer{}, \gamma^\mu\fer{})(J_M\epsilon_L, \gamma_\mu \gau{}^a)\\
	&\qquad - \frac{C_{42}}{4}(J_M\afer{}, \gamma^\mu\gamma^\nu\fer{})(J_M\epsilon_L, \gamma_\mu\gamma_\nu \gau{}^a) - \frac{C_{43}}{4}(J_M\afer{}, \gamma^\mu\gamma^5\fer{})(J_M\epsilon_L, \gamma_\mu\gamma^5 \gau{}^a)\\
	&\qquad\qquad - \frac{C_{44}}{4}(J_M\afer{}, \gamma^5\fer{})(J_M\epsilon_L, \gamma^5 \gau{}^a). 
\end{align*}
(Note that the sum in the third term on the RHS runs over $\mu < \nu$, see Example \ref{exmpl:dim4}.) We calculate: $C_{40} = C_{43} = C_{44} = -C_{41} = -C_{42}  = 1$ and use that $\fer{}$ and $\afer{}$ are of opposite parity, as are $\fer{}$ and $\gau{}^a$, to arrive at
\begin{align*}
 ( J_M\afer{}, \gamma^5\gau{}^a)(J_M\epsilon_L, \gamma^5 \fer{})
 &= \frac{1}{4}(J_M\afer{}, \gamma^\mu\fer{})(J_M\epsilon_L, \gamma_\mu \gau{}^a) - \frac{1}{4}(J_M\afer{}, \gamma^\mu\gamma^5\fer{})(J_M\epsilon_L, \gamma_\mu\gamma^5 \gau{}^a)\\
	&= \frac{1}{2}(J_M\afer{}, \gamma^\mu \fer{})(J_M\epsilon_L, \gamma_\mu\gau{}^a).
\end{align*}

\end{proof}

\begin{rmk}\label{rmk:bb2-group1-constr1.1}
	From the action there in fact arises also a similar group of terms as \eqref{eq:bb2-group1}, that reads
\ba
	& \langle J_M\afer{R}, \gamma^5\Cw{j,i}c_{ij}(J_M\epsilon_L, \gamma^5 \fer{L})\gau{jR}\rangle 
	 - g_jc_j\langle J_M \afer{R}, \gamma^\mu \fer{L}(J_M\epsilon_L, \gamma_\mu \gau{jR})\rangle \label{eq:bb2-group1.1},
\ea
where the minus sign comes from the one in \eqref{eq:fluctDfull}. Performing the same calculations, we find 
 \ba
\frac{1}{2}\Cw{j,i}c_{ij} &= c_{j}g_j\label{eq:bb2-group1-constr1.1}
\ea
here.
\end{rmk}


\begin{lem}\label{lem:bb2-group2}
	The term \eqref{eq:bb2-group2} vanishes provided that
	\begin{align}
		 \frac{1}{2}c_{ij}'^*\Cw{i,j} &= - g_ic_i = \mathcal{P}_i c_{G_i}\label{eq:bb2-group2-constr}.
	\end{align}
\end{lem}
\begin{proof}
	Using that $[J_M, \gamma^5] = 0$, $(\gamma^5)^* = \gamma^5$ and $(\gamma^5)^2 = 1$, the first term of \eqref{eq:bb2-group2} can be rewritten as
	\begin{align*}
		c_{ij}'^* \langle J_M[\can_A,\asfer]\epsilon_L, \gau{iR}\Cw{i,j}\sfer \rangle = c_{ij}'^*\inpr{J_M \asfer \epsilon_L}{\can_A\gau{iR}\Cw{i,j}\sfer},
	\end{align*}
	where we have used the self-adjointness of $\can_A$. The third term of \eqref{eq:bb2-group2} can be written as
\begin{align}
	& g_ic_i\langle J_M\asfer\epsilon_L, \can_A\gau{iR}\sfer \rangle \label{eq:bb2-group2.1}
\end{align}
where we have used that $\slashed{\partial}\epsilon_L = 0$. On the other hand, the second and fourth terms of \eqref{eq:bb2-group2} can be rewritten to yield
\begin{align}
	&  +i g_ic_i\int \tr_{N_j}\big( \asfer (J_M\epsilon_L, \gamma_\mu \gau{iR}) D^\mu \sfer\big)
 -  c_{G_i} \tr_{N_i} \P_i\sfer\asfer (J_M\epsilon_L, \can_A\gau{iR})_{\cS}\nn \\
 & =  g_ic_i \langle J_M\asfer \epsilon_L, \can_A\gau{iR}\sfer \rangle \label{eq:bb2-group2.2}
\end{align} 
provided that $g_ic_i = - \P_i c_{G_i}$. Then the two terms \eqref{eq:bb2-group2.1} and \eqref{eq:bb2-group2.2} cancel, provided that
\begin{align*}
c_{ij}'^*\Cw{i,j} + 2g_ic_i &= 0.
\end{align*}
\end{proof}


\begin{lem}\label{lem:bb2-group3}
	The expression \eqref{eq:bb2-group3} vanishes, provided that
 	\begin{align}
 			c_{ij}^* =  c_{ij}' = - 2ic_i'\Cw{i,j}g_i^{-1} = 2ic_j'\Cw{j,i}g_j^{-1}\label{eq:bb2-group3-constr}.
	\end{align}
\end{lem}
\begin{proof}
We start with \eqref{eq:bb2-group3}:	
\begin{align*}
 c_i' \langle J_M\afer{R}, \gamma^5F_i\Cw{i,j}\sfer\epsilon_R \rangle 
		 + c_{ij}^* \int \tr_{N_j}\big( D_\mu (J_M\epsilon_R, \gamma^5 \afer{R}) D^\mu \sfer\big)  
		-   c_{ij}'\langle J_M \afer{R}, \gamma^5\can_A[\can_A, \sfer]\epsilon_R\rangle,
\end{align*}
where we have used that $\{\gamma^5, \can_A\} = 0$. Note that the second term in this expression can be rewritten as
\begin{align*}
	- c_{ij}^*\langle J_M \afer{R}, \gamma^5D_\mu D^\mu  \sfer\epsilon_R\rangle 
\end{align*}	
by using the cyclicity of the trace, the Leibniz rule for the partial derivative and Lemma \ref{lem:symmJ}. (We have discarded a boundary term here.) Together, the three terms can thus be written as
\begin{align*}
	\langle J_M \afer{R}, \gamma^5 \mathcal{O} \sfer\epsilon_R), \quad\mathcal{O} = c_i' \Cw{i,j} F_i -  c_{ij}^*D_\mu D^\mu - c_{ij}' \can_A^2,
\end{align*}
where we have used that $\slashed{\partial}\epsilon_{R} = 0$. We must show that the above expression can equal zero. Using Lemma \ref{lem:Dsquared} we have, on a flat background:
\begin{align*}
	\can_A^2 + D_\mu D^\mu &= - \frac{1}{2}\gamma^\mu\gamma^\nu \mathbb{F}_{\mu\nu} = \frac{i}{2} \gamma^\mu\gamma^\nu (g_i F^i_{\mu\nu} - g_j F^{j\,o}_{\mu\nu})
\end{align*}
since $\mathbb{A}_\mu = -ig_i\mathrm{ad}(A_\mu)$. Comparing the above equation with the expression for $\mathcal{O}$ we see that if $- c_{ij}^* = - c_{ij}' = 2ic_i' \Cw{i,j}g_i^{-1}$, the operator $\mathcal{O}$ ---applied to $\sfer\epsilon_R$--- indeed equals zero. From transforming the fermionic action we also obtain the term
\bas
	c_j'\inpr{J_M\afer{R}}{\gamma^5 \Cw{i,j}\sfer_{}F_j\eR}
\eas
	from which we infer the last equality of \eqref{eq:bb2-group3-constr}
\end{proof}


\begin{lem}\label{lem:bb2-group4}
	The expression \eqref{eq:bb2-group4} vanishes, provided that
	\begin{align}
		c_{ij}^*\P_i &= c_{G_i}'\Cw{i,j}\label{eq:bb2-group4-constr}
	\end{align}
\end{lem}
\begin{proof}
	The second term of \eqref{eq:bb2-group4} is rewritten using Lemmas \ref{lem:symmJ}, \ref{lem:pullScalar} and \ref{lem:moveScalar} to give
	\begin{align*}
  - c_{ij}^*\langle J_M\afer{R}, \gamma^5G_i\P_i\sfer \eR\rangle 
	\end{align*}
	establishing the result.
\end{proof}


Then finally for the last group of terms we have:
\begin{lem}\label{lem:bb2-group5}
	The expression \eqref{eq:bb2-group5} vanishes, provided that
	\begin{align}
		d_{ij} &= d_{ij}'^*, & d_{ij,i} &= d_{ij}'^*\Cw{i,j}, & d_{ij,j} &= - d_{ij}'^*\Cw{j,i}.\label{eq:bb2-group5-constr}
	\end{align}
\end{lem}
\begin{proof}
	The first two identities of \eqref{eq:bb2-group5-constr} are immediate. The third follows from the term that we have omitted in {\eqref{eq:bb2-group5}}, which is equal to the other term except that $\gau{iR}\sfer_{} \to \sfer_{}\gau{jR}$, $\Cw{i,j} \to \Cw{j,i}$ and $d_{ij,i} \to - d_{ij,j}$.
\end{proof}


Combining the five lemmas above, we complete the proof of Theorem \ref{prop:bb2} with the following proposition: 


\begin{prop}\label{thm:bb2}
	A supersymmetric action remains supersymmetric $\mathcal{O}(\Lambda^{0})$ after adding a `building block of the second type' to the spectral triple if the scaled parameters in the finite Dirac operator are given by
\begin{align}
\Cw{i,j} &= \sgnc_{i,j} \sqrt{\frac{2}{\K_i}} g_i \id_M, &
\Cw{j,i} &= \sgnc_{j,i} \sqrt{\frac{2}{\K_j}} g_j \id_M \label{eq:bb2-results1}
\end{align}
and if 
\begin{subequations}\label{eq:bb2-results2}
	\begin{align}
		c_{ij}' &=  c_{ij}^* = \sgnc_{i,j} \sqrt{2\K_i}c_i = - \sgnc_{j,i} \sqrt{2\K_j}c_j, \label{eq:bb2-results2a}\\
		 d_{ij} &= d_{ij}'^* = \sgnc_{i,j} \sqrt{\frac{\K_i}{2}} \frac{d_{ij,i}}{g_i}  = - \sgnc_{j,i} \sqrt{\frac{\K_j}{2}} \frac{d_{ij,j}}{g_j},\label{eq:bb2-results2b}\\
\mathcal{P}_i^2 &= g_i^2\mathcal{K}_i^{-1}, \label{eq:bb2-results2d} \\
		c_{G_i} &= \sgnc_{i} \sqrt{\K_i} c_i, \label{eq:bb2-results2c}
	\end{align}
\end{subequations}
	with $\sgnc_{ij}, \sgnc_{ji}, \sgnc_{i} \in \{\pm\}$.
\end{prop}
\begin{proof}
	Using Lemmas \ref{lem:bb2-group1}, \ref{lem:bb2-group2}, \ref{lem:bb2-group3}, \ref{lem:bb2-group4} and \ref{lem:bb2-group5}, the action is seen to be fully supersymmetric if the relations \eqref{eq:bb2-group1-constr}, \eqref{eq:bb2-group2-constr}, \eqref{eq:bb2-group3-constr}, \eqref{eq:bb2-group4-constr} and \eqref{eq:bb2-group5-constr} can simultaneously be met. We can combine \eqref{eq:bb2-group1-constr} and the second equality of \eqref{eq:bb2-group3-constr} to yield 
	\begin{align*}
		ic_i'\Cw{i,j}^*\Cw{i,j} = g_i^2c_i^*\qquad\Longrightarrow \qquad  \Cw{i,j}^*\Cw{i,j}c_i = - \frac{2g_i^2}{\K_i}c_i^*,
	\end{align*}
	where in the last step we have used the relation \eqref{eq:bb1-constr-final} between $c_i$ and $c_i'$. Inserting the expression for $\Cw{i,j}$ from \eqref{eq:bb2-scalingG} and assuming that $c_i \in i \mathbb{R}$ to ensure the reality of $\Cw{i,j}$, we find the first relation of \eqref{eq:bb2-results1}. The other parameter, $\Cw{j,i}$, can be obtained by invoking Remark \ref{rmk:bb2-group1-constr1.1} and using \eqref{eq:bb2-group3-constr}, leading to the second relation of \eqref{eq:bb2-results1}. Plugging the former result into \eqref{eq:bb2-group3-constr} and \eqref{eq:bb2-group5-constr} (and invoking \eqref{eq:bb1-constr-final}) gives the second equality in \eqref{eq:bb2-results2a} and those of \eqref{eq:bb2-results2b} respectively. Combining \eqref{eq:bb2-group4-constr}, \eqref{eq:bb2-results1} and the second equality of \eqref{eq:bb2-results2a}, we find 
	\ba \label{eq:bb2-proof-interm} c_{G_i} = - g_i^{-1} \K_i \P_i c_i. \ea The combination of the second equality of \eqref{eq:bb2-group2-constr} with \eqref{eq:bb2-proof-interm} yields \eqref{eq:bb2-results2d}. Finally, plugging this result back into \eqref{eq:bb2-proof-interm} gives \eqref{eq:bb2-results2c}.
\end{proof}

Note that upon setting $\K_i \equiv 1$ (as should be done in the end) we recover the well known results for both the supersymmetry transformation constants and the parameters of the fermion--sfermion--gaugino interaction.

\subsection{Third building block}\label{sec:bb3-proof}

The off shell counterparts of the \emph{new interactions} that we get in the four-scalar action, are of the form (c.f.~\eqref{eq:bb3-action-1})  
\ba
S_{123,B}[\zeta, \szeta, F] &= \int_M \Big[\tr F_{ij}^*(\beta_{ij,k}\sfer_{ik}\asfer_{jk})+ \tr(\sfer_{jk}\asfer_{ik}\beta_{ij,k}^*)F_{ij} + 
\tr F_{ik}^*(\beta_{ik,j}^*\sfer_{ij}\sfer_{jk})+ \nn \\
&\qquad\qquad+\tr(\asfer_{jk}\asfer_{ij}\beta_{ik,j})F_{ik}  + \tr(\beta_{jk,i}\asfer_{ij}\sfer_{ik})F_{jk}^* + \tr(\asfer_{ik}\sfer_{ij}\beta_{jk,i}^*)F_{jk}\Big]\nn\\
&\equiv \int_M \Big[\tr F_{1}^*(\beta_{1}\sfer_{2}\asfer_{3}) 
+ \tr(\asfer_{3}\asfer_{1}\beta_{2})F_{2} + \tr(\beta_{3}\asfer_{1}\sfer_{2})F_{3}^* + h.c.\Big] \nn\\
&\to \int_M \Big[\tr F_{1}^*(\beta_{1}'\sfer_{2}\asfer_{3}) 
+ \tr(\asfer_{3}\asfer_{1}\bp_{2})F_{2} + \tr(\bp_{3}\asfer_{1}\sfer_{2})F_{3}^* + h.c.\Big] \label{eq:bb3-action}.
\ea
Here we have already scaled the fields according to \eqref{eq:bb3-scalingfields} and have written 
\ba\label{eq:def-scale-beta} \bp_1 &:= \n_3^{-1}\beta_1\n_2^{-1},& \bp_2 &:= \n_3^{-1} \beta_2 \n_1^{-1},& \bp_3 &:= \n_1^{-1}\beta_3\n_2^{-1}.\ea We apply the transformations \eqref{eq:susytransforms4} and \eqref{eq:susytransforms5} to the first term of \eqref{eq:bb3-action} above, giving:
\ba
 \tr F_{1}^*(\bp_{1}\sfer_{2}\asfer_{3}) &\to 
 \tr \Big[\Big(d_1^*(J_M\eL, \can_A\afer{1}) + d_{1,i}^*(J_M\eL,\gamma^5\asfer_{1}\gau{iL}) - d_{1,j}^*(J_M\eL, \gamma^5\gau{jL}\asfer_{1})\Big)(\bp_{1}\sfer_{2}\asfer_{3})\nn\\
	&\qquad  +  \tr F_{1}^*\bp_{1}c_2(J_M\eR, \gamma^5\fer{2})\asfer_{3} + \tr F_{1}^*\bp_{1}\sfer_{2}(J_M\eR, \gamma^5\afer{3})c_3^*\Big]\label{eq:bb3-transform1},
\ea
where $c_{1,2,3}$ should not be confused with the transformation parameter $c_i$ of the building blocks of the first type. We have two more terms that can be obtained from the above ones by interchanging the indices $1$, $2$ and $3$:
\ba
 \tr(\asfer_{3}\asfer_{1}\bp_2)F_{2} &\to \tr\Big[(\asfer_{3}\asfer_{1}\bp_2)\Big(d_{2}(J_M\eL, \can_A\fer{2}) + d_{2,i}(J_M\eL, \gamma^5\gau{iL}\sfer_{2})_{\cS} - d_{2,k}(J_M\eL, \gamma^5\sfer_{2}\gau{kL})\Big)\nn\\
			&\qquad + \tr c_3^*(J_M\eR, \gamma^5\afer{3})\asfer_{1}\bp_2F_2 + \tr \asfer_3c_1^*(J_M\eR, \gamma^5\afer{1})\bp_2F_2\Big]\label{eq:bb3-transform2}
\ea
and
\ba
  \tr F_{3}^*(\bp_{3}\asfer_{1}\sfer_{2}) &\to  \tr \Big[\Big(d_{3}^*(J_M\eL, \can_A\afer{3})_{S} + d_{3,j}^*(J_M\eL, \gamma^5\asfer_{3}\gau{jL})_{\cS} - d_{3,k}^*(J_M\eL, \gamma^5\gau{kL}\asfer_{3})\Big)(\bp_{3}\asfer_1\sfer_2)\nn \\
	&\qquad + \tr F_{3}^*\bp_{3}c_1^*(J_M\eR, \gamma^5\afer{1})\sfer_{2} +  \tr F_{3}^*\bp_{3}\asfer_{1}c_2(J_M\eR, \gamma^5\fer{2})\Big]\label{eq:bb3-transform3}.
\ea	
We can omit the other half of the terms in \eqref{eq:bb3-action} from our considerations. \\

We introduce the notation 
\ba\label{eq:def-scale-yuk}
	\yukp{1}{} &:= \yuk{1}{}\n_1^{-1},& 	
	\yukp{2}{} &:= \n_2^{-1}\yuk{2}{},& 	
	\yukp{3}{} &:= \yuk{3}{}\n_3^{-1},& 	
\ea
for the scaled version of the parameters. Then for three of the fermionic terms of \eqref{eq:bb3-action-ferm-detail}, after scaling the fields, we get:
\ba
\inpr{J_M \afer{1}}{\gamma^5\fer{2}\asfer_{3}\yukps{3}{}} &\to 
\inpr{J_M ( c_{1}'^*\gamma^5 [\can_A, \asfer_{1}]\eL +d_{1}'^* F_{1}^*\eR)}{\gamma^5\fer{2}\asfer_{3}\yukps{3}{}}\nn\\ 
&\qquad	 +\inpr{J_M \afer{1}}{\gamma^5\fer{2} c_{3}^*(J_M\eR, \gamma^5 \afer{3})\yukps{3}{}} \nn\\
&\qquad  +\inpr{J_M \afer{1}}{\gamma^5(c_{2}' \gamma^5 [\can_A, \sfer_{2}]\eL + d_{2}'F_{2}\eR)\asfer_{3}\yukps{3}{}} ,\label{eq:bb3-transform4}\\
\inpr{J_M \afer{1}}{\gamma^5\yukps{2}{}\sfer_{2}\afer{3}} &\to 
\inpr{J_M ( c_{1}'^*\gamma^5 [\can_A, \asfer_{1}]\eL + F_{1}^*d_{1}'^*\eR)}{\gamma^5\yukps{2}{}\sfer_{2}\afer{3}} \nn\\
&\qquad + \inpr{J_M \afer{1}}{\gamma^5\yukps{2}{}c_{2}(J_M\eR, \gamma^5 \fer{2})\afer{3}}\nn\\
&\qquad + \inpr{J_M \afer{1}}{\gamma^5\yukps{2}{}\sfer_{2}( c_{3}'^*\gamma^5 [\can_A, \asfer_{3}]\eL + d_{3}'^*F_{3}^*\eR)},\label{eq:bb3-transform5} \\
& \text{and} \nn\\
\inpr{J_M\afer{3}}{\gamma^5\asfer_{1}\yukps{1}{}\fer{2}} &\to 
\inpr{J_M( c_{3}'^*\gamma^5 [\can_A, \asfer_{3}]\eL + d_{3}'^*F_{3}^*\eR)}{\gamma^5\asfer_{1}\yukps{1}{}\fer{2}}\nn\\
&\qquad + \inpr{J_M\afer{3}}{\gamma^5 c_{1}^*(J_M\eR, \gamma^5 \afer{1})\yukps{1}{}\fer{2}}\nn\\
&\qquad + \inpr{J_M\afer{3}}{\gamma^5\asfer_{1}\yukps{1}{}( \gamma^5 [\can_A, c_{2}'\sfer_{2}]\eL + d_{2}'F_{2}\eR)}.\label{eq:bb3-transform6}
\ea
We can safely omit the other terms of the fermionic action \eqref{eq:bb3-action-ferm-detail}.\\

Collecting the terms from \eqref{eq:bb3-transform1} -- \eqref{eq:bb3-transform6} containing the same variables, we obtain the following groups of terms:
\begin{itemize}

	\item a group with three fermionic terms:
	\ba
	& \inpr{J_M \afer{1}}{\gamma^5\fer{2} c_{3}^*(J_M\eR, \gamma^5 \afer{3})\yukps{3}{}} + \inpr{J_M \afer{1}}{\gamma^5\yukps{2}{}c_{2}(J_M\eR, \gamma^5 \fer{2})\afer{3}} \nn\\
	&\qquad + \inpr{J_M\afer{3}}{\gamma^5 c_{1}^*(J_M\eR, \gamma^5 \afer{1})\yukps{1}{}\fer{2}}\nn\\
	&= \inpr{J_M \afer{1}}{\gamma^5\fer{2a} c_{3}^*(J_M\eR, \gamma^5 \afer{3b})}(\yukps{3}{})_{ba} + \inpr{J_M \afer{1}}{\gamma^5c_{2}(J_M\eR, \gamma^5 \fer{2a})\afer{3b}}(\yukps{2}{})_{ba} \nn\\
	&\qquad + \inpr{J_M\afer{3b}}{\gamma^5 c_{1}^*(J_M\eR, \gamma^5 \afer{1})\fer{2a}}(\yukps{1}{})_{ba}\label{eq:bb3-group1},
	\ea
	consisting of part of the second term of \eqref{eq:bb3-transform4}, the second term of \eqref{eq:bb3-transform5} and the second term of \eqref{eq:bb3-transform6}. Here we have explicitly written possible family indices and have assumed that it is $\sfer_{ij}$ and $\fer{ij}$ that lack these.

	\item Three similar groups containing all terms with the auxiliary fields $F_1^*$, $F_2$ and $F_3^*$ respectively:
	\begin{subequations}\label{eq:bb3-group2}
	\ba
			 &\inpr{J_M d_{1}'^*F_{1}^*\eR}{\gamma^5\fer{2}\asfer_{3}\yukps{3}{}} +  \inpr{J_M d_{1}'^*F_{1}^*\eR}{\gamma^5\yukps{2}{}\sfer_{2}\afer{3}}\nn \\ 
			 &\qquad\qquad + \int_M  \tr F_{1}^*\bp_{1}c_2(J_M\eR, \gamma^5\fer{2})\asfer_{3} + \tr F_{1}^*\bp_{1}\sfer_{2}c_3^*(J_M\eR, \gamma^5\afer{3}) \label{eq:bb3-group2.1},\\
 &  \inpr{J_M \afer{1}}{\gamma^5d_{2}'F_{2}\asfer_{3}\yukps{3}{}\eR} + \inpr{J_M\afer{3}}{\gamma^5\asfer_{1}\yukps{1}{}d_{2}'F_{2}\eR)}\nn\\
 &\qquad\qquad + \int_M  \tr \asfer_3c_1^*(J_M\eR, \gamma^5\afer{1})\bp_{2}F_2 + \tr c_3^*(J_M\eR, \gamma^5\afer{3})\asfer_{1}\bp_{2}F_2 \label{eq:bb3-group2.2}\\
	&\qquad \text{and}\nn\\
		& \inpr{J_M \afer{1}}{\gamma^5\yukps{2}{}\sfer_{2} d_{3}'^*F_{3}^*\eR} + \inpr{J_Md_{3}'^*F_{3}^*\eR}{\gamma^5\asfer_{1}\yukps{1}{}\fer{2}}\nn\\
		&\qquad\qquad + \int_M \tr F_{3}^*\bp_{3}c_1^*(J_M\eR, \gamma^5\afer{1})\sfer_{2} +  \tr F_{3}^*\bp_{3}\asfer_{1}c_2(J_M\eR, \gamma^5\fer{2}) \label{eq:bb3-group2.3},
	\ea
	\end{subequations}
	where, for example, the first group comes from parts of the first terms of \eqref{eq:bb3-transform4} and of \eqref{eq:bb3-transform5} and from the last two terms of \eqref{eq:bb3-transform1}.

	\item A group with the gauginos $\gau{iL}$, $\gau{jL}$:
	\ba
		 \int_M &\tr\big[ d_{1,i}^*(J_M\eL,\gamma^5\asfer_{1}\gau{iL}) - d_{1,j}^*(J_M\eL, \gamma^5\gau{jL}\asfer_{1})\big](\bp_1\sfer_{2}\asfer_{3})\nn\\
				&\qquad + \tr(\asfer_{3}\asfer_{1}\bp_2)\big[d_{2,i}(J_M\eL, \gamma^5\gau{iL}\sfer_{2}) - d_{2,k}(J_M\eL, \gamma^5\sfer_{2}\gau{kL})\big]\nn\\
 		&\qquad 	+  \tr\big[ d_{3,j}^*(J_M\eL, \gamma^5\asfer_{3}\gau{jL}) - d_{3,k}^*(J_M\eL, \gamma^5\gau{kL}\asfer_{3})\big](\bp_3\asfer_1\sfer_2)\label{eq:bb3-group3},
	\ea
coming from the second and third terms of \eqref{eq:bb3-transform1}, \eqref{eq:bb3-transform2} and \eqref{eq:bb3-transform3} respectively.

	\item And finally three groups of terms containing the Dirac operator $\can_A$:
\begin{subequations}\label{eq:bb3-group4}  
	\ba
&\inpr{J_M \afer{1}}{c_{2}'[\can_A, \sfer_{2}]\asfer_{3}\yukps{3}{}\eL} +
\inpr{J_M \afer{1}}{ \yukps{2}{}\sfer_{2} c_{3}'^*[\can_A, \asfer_{3}]\eL}
	+	\int_M \tr  d_1^*(J_M\eL, \can_A\afer{1})\bp_1\sfer_{2}\asfer_{3}\label{eq:bb3-group4.1}  
	,\\
&\inpr{J_M c_{1}'^*[\can_A, \asfer_{1}] \eL}{\fer{2}\asfer_{3}\yukps{3}{}}
+ \inpr{J_Mc_{3}'^*[\can_A, \asfer_{3}] \eL}{\asfer_{1}\yukps{1}{}\fer{2}}
 + \int_M\tr\asfer_{3}\asfer_{1}\bp_2d_{2}(J_M\eL, \can_A\fer{2}), \label{eq:bb3-group4.2}  
	\\&\qquad\qquad\text{and}\nn\\
&\inpr{J_M c_{1}'^* [\can_A, \asfer_{1}]\eL}{ \yukps{2}{}\sfer_{2}\afer{3}} + 
\inpr{J_M\afer{3}}{\asfer_{1}  \yukps{1}{} c_{2}'[\can_A, \sfer_{2}]\eL}
+ \int_M \tr  d_{3}^* (J_M\eL, \can_A\afer{3})\bp_3 \asfer_1\sfer_2 \label{eq:bb3-group4.3}, 
	\ea
	\end{subequations}
coming from parts of the first and third terms of \eqref{eq:bb3-transform4} -- \eqref{eq:bb3-transform6} and from the first terms of \eqref{eq:bb3-transform1} -- \eqref{eq:bb3-transform3}.
\end{itemize}

\begin{lem}\label{lem:bb3-lem1}
	The group \eqref{eq:bb3-group1} vanishes, provided that
	\ba
		c_3^*\yukps{3}{} = c_2\yukps{2}{} = c_1^*\yukps{1}{}	\label{eq:bb3-constr1}
	\ea
\end{lem}
\begin{proof}
	Since the terms contain four fermions, we must employ a Fierz transformation (Appendix \ref{sec:fierz}). Point-wise, we have for the first term of \eqref{eq:bb3-group1} (omitting its pre-factor for now):
	\bas
		&\rinpr{J_M\afer{1}}{\gamma^5\fer{2}}\rinpr{J_M\eR}{\gamma^5\afer{3}}\nn\\ &= - \frac{C_{40}}{4}\rinpr{J_M\afer{1}}{\afer{3}}\rinpr{J_M\eR}{\fer{2}}
- \frac{C_{41}}{4}\rinpr{J_M\afer{1}}{\gamma^\mu\afer{3}}\rinpr{J_M\eR}{\gamma_\mu\fer{2}}\\
&\qquad - \frac{C_{42}}{4}\rinpr{J_M\afer{1}}{\gamma^\mu\gamma^\nu\afer{3}}\rinpr{J_M\eR}{\gamma_\mu\gamma_\nu\fer{2}}
- \frac{C_{43}}{4}\rinpr{J_M\afer{1}}{\gamma^\mu\gamma^5\afer{3}}\rinpr{J_M\eR}{\gamma_\mu\gamma^5\fer{2}}\\
&\qquad - \frac{C_{44}}{4}\rinpr{J_M\afer{1}}{\gamma^5\afer{3}}\rinpr{J_M\eR}{\gamma^5\fer{2}}\\
	&= - \frac{1}{2}\rinpr{J_M\afer{1}}{\gamma^5\afer{3}}\rinpr{J_M\eR}{\gamma^5\fer{2}}
+ \frac{1}{4}\rinpr{J_M\afer{1}}{\gamma^\mu\gamma^\nu\afer{3}}\rinpr{J_M\eR}{\gamma_\mu\gamma_\nu\fer{2}},
	\eas 
where we have used that $C_{40} = C_{44} = -C_{42} = 1$ and that all fermions are of the same chirality. (Note that the sum in the last term runs over $\mu < \nu$, see Example \ref{exmpl:dim4}.) Similarly, we can take the third term of \eqref{eq:bb3-group1}, use the symmetries of the inner product for both terms, and apply the same transformation. This yields
\ba
&\rinpr{J_M\afer{3}}{\gamma^5 (J_M\eR, \gamma^5 \afer{1})\fer{2}} \nn\\
&= \rinpr{J_M\fer{2}}{\gamma^5 \afer{3}}\rinpr{J_M\afer{1}}{\gamma^5\eR}\nn\\
	&= 	- \frac{1}{2}\rinpr{J_M\fer{2}}{\gamma^5\eR}\rinpr{J_M\afer{1}}{\gamma^5 \afer{3}}
	 	+ \frac{1}{4}\rinpr{J_M\fer{2}}{\gamma^\mu\gamma^\nu\eR}\rinpr{J_M\afer{1}}{\gamma_\mu\gamma_\nu \afer{3}}\nn\\
	&= 	-  \frac{1}{2}\rinpr{J_M\eR}{\gamma^5\fer{2}}\rinpr{J_M\afer{1}}{\gamma^5 \afer{3}}
	 	- \frac{1}{4}\rinpr{J_M\eR}{\gamma^\mu\gamma^\nu\fer{2}}\rinpr{J_M\afer{1}}{\gamma_\mu\gamma_\nu \afer{3}}\label{eq:bb3-fierz1},
\ea
where we have used the symmetries \eqref{eq:identitySymJ} for the second inner product in each of the two terms of \eqref{eq:bb3-fierz1}. We can add the two results, yielding
\bas
&\rinpr{J_M\afer{1}}{\gamma^5\fer{2}}\rinpr{J_M\eR}{\gamma^5\afer{3}c_3^*\yukps{3}{}} + \rinpr{J_M\afer{3}}{\gamma^5 (J_M\eR, \gamma^5 \afer{1})c_1^*\yukps{1}{}\fer{2}}\nn\\
&= - \frac{1}{2}(c_1^*\yukps{1}{}
 +  c_3^*\yukps{3}{} )_{ba} \rinpr{J_M\afer{1}}{\gamma^5\afer{3b}}\rinpr{J_M\eR}{\gamma^5\fer{2a}} \nn\\
&\qquad + \frac{1}{4} (c_3^*\yukps{3}{}
 - c_1^*\yukps{1}{} )_{ba} \rinpr{J_M\eR}{\gamma^\mu\gamma^\nu\fer{2a}}\rinpr{J_M\afer{1}}{\gamma^\mu\gamma^\nu \afer{3b}}
\eas
When $c_3^*\yukps{3}{} = c_1^*\yukps{1}{} = c_2\yukps{2}{}$, this result is seen to cancel the remaining term in \eqref{eq:bb3-group1}.
\end{proof}


\begin{lem}
	The groups of terms \eqref{eq:bb3-group2} vanish, provided that
	\ba
		c_2\bp_1  &= - d_1'^*\yukps{3}{}, &c_3^*\bp_1  &= - d_1'^*\yukps{2}{},& c_3^*\bp_2  &= - d_2'\yukps{1}{},\nn\\
		c_1^*\bp_2  &= - d_2'\yukps{3}{},& c_1^*\bp_3  &= - d_3'^*\yukps{2}{}, & c_2 \bp_3&= - d_3'^*\yukps{1}{}.\label{eq:bb3-constr2}
	\ea
\end{lem}
\begin{proof}
	This can readily be seen upon using Lemma \ref{lem:symmJ}, the cyclicity of the trace and Lemma \ref{lem:pullScalar}.
\end{proof}

\begin{lem}\label{lem:bb3-group3}
	The group of terms \eqref{eq:bb3-group3} vanishes, provided that
	\ba
		d_{1,i}^*\bp_1 &= -d_{2,i}\bp_2, & d_{1,j}^*\bp_1	&= d_{3,j}^*\bp_3, & d_{2,k}\bp_2 	&= -d_{3,k}^*\bp_3.\label{eq:bb3-constr3}
	\ea
\end{lem}
\begin{proof}
	This can readily be seen upon using the cyclicity of the trace and Lemma \ref{lem:pullScalar}.
\end{proof}

\begin{lem}\label{lem:bb3-group4}
	The three groups of terms \eqref{eq:bb3-group4} vanish, provided that
	\ba
		\yukps{3}{}c_2' &= c_3'^*\yukps{2}{} = - d_1^* \bp_1,&
		\yukps{3}{}c_1'^* &= c_3'^*\yukps{1}{} = - \bp_2d_2,&
		c_1'^*\yukps{2}{} &= \yukps{1}{}c_2' = - d_3^*\bp_3.\label{eq:bb3-constr4}
	\ea
\end{lem}
\begin{proof}
	This can be checked quite easily using the symmetry \eqref{eq:identitySymJ}, the Leibniz rule for $\can_A$ and the fact that it is self-adjoint, that $\epsilon_{L,R}$ vanish covariantly and Lemmas \ref{lem:pullScalar} and \ref{lem:moveScalar}.
\end{proof}

Combining the above lemmas, we get:

\begin{prop}\label{prop:bb3}
	The extra action as a result of adding a building block \B{ijk} of the third type is supersymmetric if and only if the coefficients $\yuk{i}{j}$, $\yuk{i}{k}$ and $\yuk{j}{k}$ are related to each other via 	
	\ba
	\yuk{i}{j}C_{iij}^{-1} &= -(C_{iik}^*)^{-1}\yuk{i}{k}, &
	\yuk{i}{j}C_{ijj}^{-1} &= - \yuk{j}{k}C_{jjk}^{-1},& 
	(C_{ikk}^*)^{-1}\yuk{i}{k} &= - \yuk{j}{k}C_{jkk}^{-1}, \label{eq:bb3-result3}
	\ea
 the constants of the transformations satisfy
	\ba
	|d_1|^2 &= 
	|d_2|^2 = 
	|d_3|^2 = |c_1|^2 =  |c_2|^2 = |c_3|^2\label{eq:bb3-result2}
	\ea
and the coefficients $\bp_{ij}$ are given by
	\ba
\bps_1\bp_1 &= \bps_2\bp_2 = \bps_3\bp_3 = \yukp{1}{}\yukps{1}{}=  \yukp{2}{}\yukps{2}{} = \yukp{3}{}\yukps{3}{}.\label{eq:bb3-result1}
	\ea
\end{prop}
\begin{proof}
	First of all, we plug the intermediate result \eqref{eq:bb2-group1-constr} for $\Cw{i,j}$ as given by \eqref{eq:bb2-scalingG} (but keeping in mind the results of Remark \ref{rmk:bb2-group1-constr1.1}) into the Hermitian conjugate of the result \eqref{eq:bb3-constr1} such that pairwise the same combination $c_ig_i$ appears on both sides. This yields 
\bas
	\yuk{i}{j}(- 2c_ig_i)C_{iij}^{-1} &= (-2c_ig_iC_{iik}^{-1})^*\yuk{i}{k}, &
	\yuk{i}{j}(2c_jg_jC_{ijj}^{-1}) &= \yuk{j}{k}(-2c_jg_j)C_{jjk}^{-1},\nn\\
	(2c_kg_kC_{ikk}^{-1})^*\yuk{i}{k} &= \yuk{j}{k}(2c_kg_k)C_{jkk}^{-1}. &&\nn
\eas
Using that the $c_{i,j,k}$ are purely imaginary (cf.~Theorem \ref{thm:bb2}), we obtain \eqref{eq:bb3-result3}. Secondly, comparing the relations \eqref{eq:bb3-constr2} with \eqref{eq:bb3-constr4} gives
	\bas
	d_1d_1' &= (c_2c_2')^* = c_3c_3', &
	(d_2d_2')^* &= c_1c_1' = c_3c_3', &
	d_3d_3' &= c_1c_1' = (c_2c_2')^*.
	\eas
Using the relations \eqref{eq:bb2-results2a} and \eqref{eq:bb2-results2b} between the constraints, \eqref{eq:bb3-result2} follows. Plugging the relations from \eqref{eq:bb3-result2} into those of \eqref{eq:bb3-constr2}, we obtain
\bas
\bps_1\bp_1 &= \yukp{3}{}\yukps{3}{}= \yukp{2}{}\yukps{2}{}, &
\bps_2\bp_2 &= \yukp{1}{}\yukps{1}{}= \yukp{3}{}\yukps{3}{}, &
\bps_3\bp_3 &= \yukp{2}{}\yukps{2}{}= \yukp{1}{}\yukps{1}{},
\eas
from which \eqref{eq:bb3-result1} directly follows.
\end{proof}

N.B.~Using \eqref{eq:def-scale-beta} and \eqref{eq:def-scale-yuk} we can phrase the identities \eqref{eq:bb3-result1} in terms of the unscaled quantities $\beta_{1,2,3}$ and $\yuk{1,2,3}{}$ as
\bas
	\n_3^{-1}\beta_2 &= \beta_3\n_2^{-1} = \yuks{1}{},&
	\n_3^{-1}\beta_1 &= \beta_3\n_1^{-1} = \yuks{2}{},&
	\n_1^{-1}\beta_2 &= \beta_1\n_2^{-1} = \yuks{3}{},
\eas
where we have used that $\n_{1} \in \mathbb{R}$ since $\sfer_{1}$ has $R = 1$ (and consequently multiplicity 1).

\subsection{Fourth building block}\label{sec:bb4-proof}

Phrased in terms of the auxiliary field $F_{11'} =: F$, a building block of the fourth type induces the following action:
\bas
&	\frac{1}{2}\inpr{J_M\fer{}}{\gamma^5 \maj^* \fer{}} + \frac{1}{2}\inpr{J_M\afer{}}{\gamma^5 \maj \afer{}} - \tr \Big(F^* \gamma \asfer + h.c.\Big). 
\eas
Here we have written $\fer{} := \fer{11'L}$, $\afer{} := \afer{11'R}$ and $\sfer := \sfer_{11'}$ for conciseness. Transforming the fields that appear in the above action, we have the following.
\begin{itemize}
\item From the first term:
\bas
	\frac{1}{2}\inpr{J_M(c^*\gamma^5[\can_A, \sfer_{}]\eR + d^* F\eL)}{\gamma^5 \maj^* \fer{}} + 
	\frac{1}{2}\inpr{J_M\fer{}}{\gamma^5 \maj^* (c^*\gamma^5[\can_A, \sfer_{}]\eR + d^* F\eL)}.
\eas

\item From the second term:
\bas
	\frac{1}{2}\inpr{J_M(c\gamma^5[\can_A, \asfer_{}]\eL + d F^*\eR)}{\gamma^5 \maj \afer{}} + 
	\frac{1}{2}\inpr{J_M\afer{}}{\gamma^5 \maj (c\gamma^5[\can_A, \asfer_{}]\eL + d F^*\eR)}.
\eas

\item From the terms with the auxiliary fields:
\bas
	& - \tr \Big[d^*(J_M\eL, \can_A\afer{}) + d'^*(J_M\eL, \gamma^5\asfer\gau{1L}) - d''^*(J_M\eL, \gamma^5\gau{1'L}\asfer)\Big] \gamma \asfer - c^* \tr F^* \gamma \rinpr{J_M\eR}{\gamma^5 \afer{}} 
\eas
and 
\bas 
	& - \tr  \sfer{}\gamma^*\Big[d(J_M\eR, \can_A\fer{}) + d'(J_M\eR, \gamma^5\gau{1R}\sfer) - d''(J_M\eR, \gamma^5\sfer\gau{1'R})\Big]
	- c\tr (J_M\eL, \gamma^5\fer{})\gamma^* F.	
\eas
\end{itemize}
Here we have written $c := c_{ij}$, $d := d_{ij}$ (where we have expressed $c_{ij}'^*$ as $c_{ij}$ and $d_{ij}'^*$ as $d_{ij}$ using \eqref{eq:bb2-results2a} and \eqref{eq:bb2-results2b}) and $d' := d_{11',1}$, $d'' := d_{11',1'}$. We group all terms according to the fields that appear in them, leaving essentially the following three.
\begin{itemize}

\item The group consisting of all terms with $F^*$ and $\afer{}$:
\bas
	&\frac{1}{2}\inpr{J_M d F^*\eR}{\gamma^5 \maj \afer{}} + 
	\frac{1}{2}\inpr{J_M\afer{}}{\gamma^5 \maj  d F^*\eR}
	 - c^*\int_M \tr F^* \gamma \rinpr{J_M\eR}{\gamma^5 \afer{}} \\
	&= \inpr{J_M F^*\eR}{\gamma^5(d \maj - c^*\gamma)\afer{}} 
\eas
where we have used the symmetry of the inner product from Lemma \ref{lem:symmJ} and Lemma \ref{lem:pullScalar}. This group thus only vanishes if
\ba\label{eq:bb4-group1-constr}
	d \maj &= c^*\gamma.
\ea
There is also a group of terms featuring $F$ and $\fer{}$, but this is of the same form as the one above.

\item A group of three terms with $\fer{}$ and $\sfer{}$:
\bas
	& \frac{1}{2}\inpr{J_Mc^*\gamma^5[\can_A, \sfer_{}]\eR }{\gamma^5 \maj^* \fer{}} + 
	\frac{1}{2}\inpr{J_M\fer{}}{\gamma^5 \maj^* c^*\gamma^5[\can_A, \sfer_{}]\eR}
	 - \int_M \tr  \sfer{}\gamma^*d(J_M\eR, \can_A\fer{}) \\
	&= \inpr{J_Mc^*\gamma^5[\can_A, \sfer_{}]\eR }{\gamma^5 \maj^* \fer{}}
	 -   \inpr{J_M\sfer{}\eR}{\can_A\gamma^*d\fer{}},
\eas
where also here we have used Lemmas \ref{lem:symmJ} and \ref{lem:pullScalar}. Using the self-adjointness of $\can_A$ this is only seen to vanish if
\ba\label{eq:bb4-group2-constr}
	c^*\maj^* &= \gamma^* d.
\ea
There is also a group of terms featuring $\afer{}$ and $\asfer{}$ but these are seen to be of the same form as the terms above.
\item Finally, there are terms that feature gauginos:
\bas
	& -\int_M \Big[ \tr d'^*(J_M\eL, \gamma^5\asfer\gau{1L}) - d''^*(J_M\eL, \gamma^5\gau{1'L}\asfer)\Big] \gamma \asfer \nn\\
	&\qquad  - \int_M \tr \sfer{}\gamma^*\Big[d'(J_M\eR, \gamma^5\gau{1R}\sfer) - d''(J_M\eR, \gamma^5\sfer\gau{1'R})\Big].
\eas
This expression is immediately seen to vanish when 
\bas
d'^*\gau{1L}  &= d''^*\gau{1'L}, &
d'\gau{1R} &= d''\gau{1'R}.
\eas
For this to happen we need that the gauginos are associated to each other and that $d' = d''$.

\end{itemize}

Combining the demands \eqref{eq:bb4-group1-constr} and \eqref{eq:bb4-group2-constr} we obtain
\bas
	\maj^*\maj = \frac{|c|^2}{|d|^2} \gamma^*\gamma = \frac{|d|^2}{|c|^2}\gamma^*\gamma
\eas
i.e.
\bas
	\maj^*\maj &=  \gamma^*\gamma, & |d|^2 &= |c|^2.
\eas
\subsection{Fifth building block}\label{sec:bb5-proof}

We transform the fields that appear in the action according to \eqref{eq:susytransforms4} and \eqref{eq:susytransforms5}. We suppress the indices $i$ and $j$ as much as possible, writing $c \equiv c_{ij}, d \equiv d_{ij}$ for the transformation coefficients \eqref{eq:susytransforms5} of the building block \Bc{ij}{+} of the second type. We eliminate $c_{ij}'$ and $d_{ij}'$ in these transformations using the first relations of \eqref{eq:bb2-results2a} and \eqref{eq:bb2-results2b} so that we can write $c', d'$ for those associated to \Bc{ij}{-}.

The first fermionic term of \eqref{eq:bb5-action-scaled} transforms as
\bas
 \inpr{J_M \afer{R}}{\gamma^5 \mu\fer{R}'} &\to
 \inpr{J_M (\gamma^5 c[\can_A, \asfer]\eL + dF^*\eR)}{\gamma^5 \mu\fer{R}'}\\
	&\qquad + \inpr{J_M \afer{R}}{\gamma^5 \mu(c'^* \gamma^5 [\can_A, \sfer']\eL + d'^*F'\eR)} 
\eas

The second fermionic term of \eqref{eq:bb5-action-scaled} transforms as
\bas
 \inpr{J_M\afer{L}'}{\gamma^5 \mu^*\fer{L}} &\to 
 \inpr{J_M(c'\gamma^5 [\can_A, \asfer']\eR + d'F'^*\eL)}{\gamma^5 \mu^*\fer{L}} \\ 
 &\qquad + \inpr{J_M\afer{L}'}{\gamma^5 \mu^*(c^* \gamma^5 [\can_A, \sfer]\eR + d^*F\eL)}
\eas

The four terms in \eqref{eq:bb5-auxfields} transform as
\bas
	 - \int_M \tr F'^* \delta \sfer  &\to - \int_M \Big(\tr \big[d'^*(J_M\eR, \can_A\afer{L}') + d_{ij,i}'^*(J_M\eR, \gamma^5\asfer'\gau{iR}) \nn\\
		&\qquad\qquad - d_{ij,j}'^*(J_M\eR, \gamma^5\gau{jR}\asfer')\big] \delta \sfer  + \tr F'^* \delta c(J_M\eL, \gamma^5 \fer{L})\Big) \nn,\\
	 - \int_M  \tr F^* \delta' \sfer' &\to - \int_M \Big(	\tr \big[d^*(J_M\eL, \can_A\afer{R}) + d_{ij,i}^*(J_M\eL, \gamma^5\asfer\gau{iL})\nn\\
			&\qquad\qquad  - d_{ij,j}^*(J_M\eL, \gamma^5\gau{jL}\asfer)\big] \delta' \sfer'  + \tr F^* \delta' c'(J_M\eR, \gamma^5 \fer{R}') \Big)\nn,\\
	 - \int_M  \tr \asfer \delta^* F' &\to - \int_M \Big( \tr c^*(J_M\eR, \gamma^5 \afer{R}) \delta^* F' + \tr \asfer \delta^* \big[d'(J_M\eL, \can_A\fer{R}') \nn\\
	&\qquad\qquad+ d_{ij,i}'(J_M\eL, \gamma^5\gau{iL}\sfer') - d_{ij,j}'(J_M\eL, \gamma^5\sfer'\gau{jL})\big]\Big)\nn
\intertext{and }
	 - \int_M \tr \asfer' \delta'^* F &\to - \int_M \Big(\tr c'^*(J_M\eL, \gamma^5 \afer{L}') \delta'^* F  + \tr \asfer' \delta'^* \big[ d(J_M\eR, \can_A\fer{L}) \nn\\
		&\qquad\qquad + d_{ij,i}(J_M\eR, \gamma^5\gau{iR}\sfer) - d_{ij,j}(J_M\eR, \gamma^5\sfer\gau{jR})\big] \Big).
\eas

We group all terms that feature the same fields, which gives
\begin{itemize}
\item  a group with $F$ and $F'$:
\bas
 & d'^*\inpr{J_M \afer{R}}{\gamma^5 \mu F'\eR} +  d^*\inpr{J_M\afer{L}'}{\gamma^5 \mu^*F\eL}\nn\\
	&\qquad\qquad - \int_M\Big( \tr c^*(J_M\eR, \gamma^5 \afer{R}) \delta^* F' + \tr c'^*(J_M\eL, \gamma^5 \afer{L}') \delta'^* F\Big). \nn
\eas
Using Lemmas \ref{lem:moveScalar} and \ref{lem:pullScalar} and employing the symmetries of the inner product (Lemma \ref{lem:symmJ}), this is seen to equal
\bas
 & d'^*\inpr{J_M \afer{R}}{\gamma^5 \mu F'\eR} +  d^*\inpr{J_M\afer{L}'}{\gamma^5 \mu^*F\eL}\nn\\
	&\qquad\qquad -  c^*\inpr{J_M\afer{R}}{\gamma^5  \delta^* F'\eR} -  c'^*\inpr{J_M \afer{L}'}{\gamma^5 \delta'^* F\eL} \nn\\
 &= \inpr{J_M \afer{R}}{\gamma^5 \big[ d'^*\mu -  c^* \delta^*\big] F'\eR} +  \inpr{J_M\afer{L}'}{\gamma^5 \big[ d^*\mu^* - c'^*\delta'^*\big]F\eL}\nn
\eas
This only vanishes if
\ba\label{eq:bb5-constr1}
 d'^*\mu &=  c^* \delta^*, & d^*\mu^* &= c'^*\delta'^*.
\ea

\item a group with $F^*$ and $F'^*$, that vanishes automatically if and only if \eqref{eq:bb5-constr1} is satisfied.

\item a group featuring $\fer{R}'$ and $\fer{L}$:
\bas
 & \inpr{J_M  c[\can_A, \asfer]\eL}{ \mu\fer{R}'} + c'\inpr{J_M [\can_A, \asfer']\eR }{ \mu^*\fer{L}} \nn\\
	&\qquad\qquad - \int_M \Big( \tr \asfer \delta^*d'(J_M\eL, \can_A\fer{R}') + \tr \asfer' \delta'^* d(J_M\eR, \can_A\fer{L})  \Big).
\eas
Employing Lemmas \ref{lem:pullScalar} and \ref{lem:moveScalar} this is seen to equal
\bas
 & \inpr{J_M  c[\can_A, \asfer]\eL}{ \mu\fer{R}'} +  c'\inpr{J_M [\can_A, \asfer']\eR }{ \mu^*\fer{L}} \nn\\
	&\qquad\qquad - d'\inpr{J_M\asfer \delta^*\eL}{\can_A\fer{R}'} -  d\inpr{J_M\asfer' \delta'^*\eR}{\can_A\fer{L}} 
\eas
Using the self-adjointness of $\can_A$, that $[\mu, \can_A] = 0$ and the symmetries of the inner product, this reads
\bas
 & \inpr{J_M  \asfer\eL}{ \big[c\mu - d'\delta^*\big]\can_A\fer{R}'} +  \inpr{J_M \asfer'\eR }{\big[c' \mu^* - d \delta'^*\big]\can_A\fer{L}} \nn.
\eas
We thus require that 
\ba\label{eq:bb5-constr2}
c\mu &= d'\delta^*, & c' \mu^* &= d \delta'^*
\ea
for this to vanish.

\item a group with $\afer{R}$ and $\afer{L}'$ that vanishes if and only if \eqref{eq:bb5-constr2} is satisfied. 

	\item a group with the left-handed gauginos: 
\bas
	&- \int_M \Big( \tr \big[d_{ij,i}^*(J_M\eL, \gamma^5\asfer\gau{iL}) - d_{ij,j}^*(J_M\eL, \gamma^5\gau{jL}\asfer)\big] \delta' \sfer' \nn\\
	&\qquad\qquad + \tr \asfer \delta^* \big[ d_{ij,i}'(J_M\eL, \gamma^5\gau{iL}\sfer') - d_{ij,j}'(J_M\eL, \gamma^5\sfer'\gau{jL})\big]\Big)\nn\\
	& =  - \inpr{J_M \big(d_{ij,i}^* \delta'\sfer'\asfer  + d_{ij,i}'\sfer'\asfer \delta^*\big)\eL}{\gamma^5\gau{iL}} + \inpr{J_M\big(d_{ij,j}^*\asfer \delta' \sfer' + d_{ij,j}'\asfer \delta^* \sfer'\big)\eL}{\gamma^5\gau{jL}},  
\eas
where we have used Lemmas \ref{lem:moveScalar} and \ref{lem:pullScalar}. For this to vanish, we require that 
\bas
	 d_{ij,i}^* \delta' &= - d_{ij,i}' \delta^*, & d_{ij,j}^* \delta' &= - d_{ij,j}' \delta^* . 
\eas
Inserting \eqref{eq:bb5-constr2} above this is equivalent to
\bas
d_{ij,i}^*\frac{c'^*}{d^*} &= - d_{ij,i}'\frac{c}{d'}, &
d_{ij,j}^*\frac{c'^*}{d^*} &= - d_{ij,j}'\frac{c}{d'}.
\eas

\item A group with the right-handed gauginos 
\bas
	  &- \int_M \tr \big[ d_{ij,i}'^*(J_M\eR, \gamma^5\asfer'\gau{iR})  - d_{ij,j}'^*(J_M\eR, \gamma^5\gau{jR}\asfer')\big] \delta \sfer \nn\\
	 &\qquad - \int_M \tr  \asfer'\delta'^*\big[  d_{ij,i}(J_M\eR, \gamma^5\gau{iR}\sfer) - d_{ij,j}(J_M\eR, \gamma^5\sfer\gau{jR})\big] \nn\\
	  &= -  \inpr{J_M\big(d_{ij,i}'^* \delta \sfer \asfer' + d_{ij,i}\sfer\asfer'\delta'^*\big)\eR}{ \gamma^5\gau{iR}}  + \inpr{J_M\big(d_{ij,j}'^*\asfer'\delta \sfer + d_{ij,j}\asfer' \delta'^*\sfer\big)\eR}{ \gamma^5\gau{jR}},
\eas
which vanishes iff
\bas
	  d_{ij,i}'^* \delta &= - d_{ij,i}\delta'^* ,& d_{ij,j}'^*\delta&= - d_{ij,j} \delta'^*.
\eas

\end{itemize}

Combining all relations, above, we require that 
\bas
	|c|^2 &= |d'|^2,& |c'|^2 &= |d|^2,& |d_{ij,i}|^2 &= |d_{ij,i}'|^2,& |d_{ij,j}|^2 &= |d_{ij,j}'|^2,
\eas
for the transformation constants and 
\bas
	\delta\delta^* &= \mu^*\mu,  & \delta'\delta'^* &= \mu\mu^*
\eas
for the parameters in the off shell action.

\section{Auxiliary lemmas and identities}

In this section we provide some auxiliary lemmas and identities that are used in and throughout the previous proofs.

\begin{lem}
For the spin-connection $\nabla^S :\Gamma(S) \to \mathcal{A}^1(M) \otimes_{C^\infty(M)} \Gamma(S)$ on a flat manifold we have:
\begin{align}
	[\nabla^S, \gamma^\mu]	&= 0.\label{eq:idnNablaS}
\end{align}
\end{lem}
\begin{proof}
	The spin-connection is the unique connection compatible with the Levi-Civita connection on $T^*M$, which means that it satisfies
	\begin{align*}
		\nabla^Sc(\alpha) = c(\alpha)\nabla^S + c(\nabla^g \alpha)	
	\end{align*}
	for any $\alpha \in \Gamma^{\infty}(\mathbb{C}l(M))$. Here $c : \Gamma(\mathbb{C}l(M)) \to \Gamma(\End(S))$ is the \emph{spin homomorphism}. Taking in particular $\alpha = \mathrm{d}x^\mu$, writing $\gamma^\mu = c(\mathrm{d}x^\mu)$ and using that $\nabla^g\mathrm{d}x^\mu = - \Gamma^\mu_{\nu\lambda} \mathrm{d}x^\nu \otimes \mathrm{d}x^\lambda = 0$ for a flat manifold, we have
	\begin{align*}
		\nabla^S\gamma^\mu &= \gamma^\mu\nabla^S. 
	\end{align*}
	Here we have used that $c(\mathrm{d}x^\mu \otimes \mathrm{d}x^\nu) = c(\mathrm{d}x^\mu \wedge \mathrm{d}x^\nu) + c(\{\mathrm{d}x^\mu,\mathrm{d}x^\nu\})$.
\end{proof}

\begin{lem}\label{lem:Dsquared}
Let
$
	\can_A = -i c \circ (\nabla^S + \mathbb{A})
$
 and $D_\mu = (\nabla^S + \mathbb{A})_\mu$. For a flat manifold, we have locally:
\begin{align*}
	\can_A^2 + D_\mu D^\mu = - \frac{1}{2} \gamma^\mu\gamma^\nu \mathbb{F}_{\mu\nu}.
\end{align*}
\end{lem}
\begin{proof}
Locally we write 
\begin{align*}
	\can_A = -i c(\mathrm{d}x^\mu) (\nabla^S_\mu + \mathbb{A}_\mu)
\end{align*}
where $\mathbb{A}_\mu$ is skew-Hermitian in order for $\can_A$ to be self-adjoint. Now for the square of this, we have
\begin{align*}
	 \can_A^2  &= - c(\mathrm{d}x^\mu) (\nabla^S_\mu + \mathbb{A}_\mu)c(\mathrm{d}x^\nu) (\nabla^S_\nu + \mathbb{A}_\nu)\\
					&= - c(\mathrm{d}x^\mu) c(\mathrm{d}x^\nu)(\nabla^S_\mu + \mathbb{A}_\mu)(\nabla^S_\nu + \mathbb{A}_\nu) 
					- c(\mathrm{d}x^\mu) c(\nabla^g_{\partial_\mu}\mathrm{d}x^\nu)(\nabla^S_\nu + \mathbb{A}_\nu), 
\end{align*}
of which the last term vanishes for a flat manifold. Here we have employed that the spin connection is the unique connection compatible with the Levi-Civita connection. We write:
\begin{align*}
	c(\mathrm{d}x^\mu) c(\mathrm{d}x^\nu) &= \frac{1}{2}\{c(\mathrm{d}x^\mu), c(\mathrm{d}x^\nu)\} + \frac{1}{2}[c(\mathrm{d}x^\mu), c(\mathrm{d}x^\nu)]= g^{\mu\nu} + \frac{1}{2}[c(\mathrm{d}x^\mu), c(\mathrm{d}x^\nu)]
\end{align*}
to arrive at
\begin{align*}
 \can_A^2 &= - (\nabla^{S} + \mathbb{A})^\mu(\nabla^S + \mathbb{A})_\mu - \frac{1}{2}c(\mathrm{d}x^\mu)c(\mathrm{d}x^\nu)[\nabla^S_\mu + \mathbb{A}_\mu, \nabla^S_\nu + \mathbb{A}_\nu] 
\end{align*}
obtaining the result.
\end{proof}

\begin{cor}\label{cor:4bcurvs}
By applying the previous result, we have for $\szeta_{ik} \in C^{\infty}(M, \mathbf{N_i}\otimes \mathbf{N_k})$, $\epsilon \in L^2(M, S)$
\begin{align*}
	(\can_A [\can_A, \szeta_{ik}]\epsilon + D_\mu [D^\mu, \szeta_{ik}])\epsilon = \frac{1}{2}[\mathbb{F}, \szeta_{ik}]\epsilon + [D^\mu, \szeta_{ik}]\nabla^S_\mu\epsilon + [\can_A, \szeta_{ik}]\slashed{\partial}\epsilon,
\end{align*}
where the term with $R$ vanished due to the commutator.
\end{cor}

\begin{lem}\label{lem:symmJ}
Let $M$ be a four-dimensional Riemannian spin manifold and $\inpr{\,.\,}{\,.\,} : L^2(S) \times L^2(S) \to \mathbb{C}$ the inner product on sections of the spinor bundle. For $\mathcal{P}$ a basis element of $\Gamma(\com l(M))$, we have the following identities:
\bas
	\inpr{J_M\zeta_1}{\mathcal{P}\zeta_2} = \pi_{\mathcal{P}}\inpr{J_M\zeta_2}{\mathcal{P}\zeta_1},\qquad \pi_{\mathcal{P}} \in \{\pm\},
\eas 
for any $\zeta_{1,2}$, the Grassmann variables corresponding to $\zeta_{1,2}' \in L^2(S)$. The signs $\pi_{\P}$ are given by 
\begin{align}
	\pi_{\id} &= 1, & \pi_{\gamma^\mu} &= - 1, & \pi_{\gamma^\mu\gamma^\nu} &= - 1\quad (\mu < \nu), \nn\\
 	\pi_{\gamma^\mu\gamma^5} &=  1, & \pi_{\gamma^5} &=  1.\label{eq:identitySymJ}
\end{align} 
\end{lem}
\begin{proof}
Using that $J_M^2 = -1$ and $\inpr{J_M\zeta_1'}{J_M\zeta_2'} = \inpr{\zeta_2'}{\zeta_1'}$, we have
\bas
	\inpr{J_M\zeta_1'}{\mathcal{P}\zeta_2'} = - \inpr{J_M\zeta_1'}{J_M^2\mathcal{P}\zeta_2'} = - \inpr{J_M\mathcal{P}\zeta_2'}{\zeta_1'}.
\eas
When considering Grassmann variables, we obtain an extra minus sign (see the discussion in \cite[\S 4.2.6]{DS12}). From $J_M\gamma^\mu = - \gamma^\mu J_M$, $(\gamma^\mu)^* = \gamma^\mu$ and $\gamma^\mu\gamma^\nu = - \gamma^\nu\gamma^\mu$ for $\mu \ne \nu$, we obtain the result.
\end{proof}

\begin{cor}\label{cor:symmInnerProd}
	Similarly (\cite[\S 4]{CCM07}) we find by using that $\dirac^* = \dirac$ and $J_M\dirac = \dirac J_M$, that
	\ba
		\inpr{J_M\zeta_1}{\dirac\zeta_2} =  \inpr{J_M\zeta_2}{\dirac\zeta_1}\label{eq:identitySymJ2}
	\ea
	for the Grassmann variables corresponding to any two $\zeta_{1, 2}' \in L^2(S)$.
\end{cor}

\begin{lem}\label{lem:pullScalar}
	For any $\sfer \in C^{\infty}(M, \rep{i}{j})$, $\fer{} \in L^2(S \otimes \rep{j}{i})$ and $\epsilon \in L^2(S)$ we have
	\begin{align*}
		\tr_{N_i} \sfer(J_M\epsilon, \fer{})_{\cS} = (J\sfer\epsilon, \fer{})_{\H}.
	\end{align*}
\end{lem}
\begin{proof}
This can be seen easily by writing out the elements in full detail:
	 \begin{align*}
		\szeta &= f \otimes e \otimes \bar e', & \fer{} &= \zeta \otimes \eta \otimes \bar \eta',\qquad f \in C^{\infty}(M, \com), \zeta \in L^2(S).
		\end{align*}
\end{proof}

\begin{lem}\label{lem:moveScalar}
	Let $\fer{1} \in L^2(S\otimes \rep{i}{j})$, $\fer{2} \in L^2(S \otimes \rep{k}{i})$, $\afer{2} \in L^2(S \otimes \rep{j}{k})$, $\sfer \in C^{\infty}(M, \rep{j}{k})$ and $\sfer' \in C^{\infty}(M, \rep{k}{i})$, then
\ba
	\inpr{J\fer{1}\sfer}{\fer{2}} &= \inpr{J\fer{1}}{\sfer\fer{2}}& &\text{and}& \inpr{J\fer{1}}{\fer{2}\sfer'} &= \inpr{J\sfer'\fer{1}}{\fer{2}}.\label{eq:moveScalar}
\ea
\end{lem}
\begin{proof}
	This can simply be proven by using that the right action is implemented via $J$ and that $J$ is an anti-isometry with $J^2 = \pm $.
\end{proof}

\subsection{Fierz transformations}\label{sec:fierz}

Details for the Fierz transformation in this context can be found in the Appendix of \cite{BS10} but we list the main result here.

\begin{defin}[Orthonormal Clifford basis] Let $Cl(V)$ be the Clifford algebra over a vector space $V$ of dimension $n$. Then $\gamma_K := \gamma_{k_1}\cdots\gamma_{k_r}$ for all strictly ordered sets $K = \{k_1 < \ldots < k_r\} \subseteq \{1, \ldots, n\}$ form a basis for $Cl(V)$. If $\gamma_K$ is as above, we denote with $\gamma^K$ the element $\gamma^{k_1}\cdots\gamma^{k_r}$. The basis spanned by the $\gamma_K$ is said to be \emph{orthonormal} if $\tr\gamma_K\gamma_L = nn_K\delta_{KL}\ \forall\ K, L$. Here $n_K := (-1)^{r(r-1)/2}$, where $r$ denotes the cardinality of the set $K$ and with $\delta_{KL}$ we mean 
\begin{align}
  \delta_{KL} = \left\{\begin{array}{ll} 1\quad \text{if}\ K = L\\
                         0 \quad \text{else}\\
                       \end{array}.
                \right.
\end{align}
\end{defin}

\begin{exmpl}\label{exmpl:dim4} Take $V = \mathbb{R}^4$ and let $Cl(4, 0)$ be the Euclidean Clifford algebra [i.e. with signature \mbox{($+$\ $+$\ $+$\ $+$)}]. Its basis are the sixteen matrices
\begin{align*}
 & 1                                                      & & \nonumber\\
 & \gamma_\mu                                             & & \text{(4 elements)} \nonumber\\
 & \gamma_\mu\gamma_\nu\quad          \quad (\mu < \nu)   &        & \text{(6 elements)}\nonumber\\
 & \gamma_\mu\gamma_\nu\gamma_\lambda \quad (\mu < \nu < \lambda)& & \text{(4 elements)} \nonumber\\
 &\gamma_1\gamma_2\gamma_3\gamma_4=:\gamma_5.             &  &\nonumber
\end{align*}
We can identify
\begin{align}\label{eq:mink_eucl}
 \gamma_1\gamma_2\gamma_3 &= \gamma_4\gamma_5, & \gamma_1\gamma_3\gamma_4 &= \gamma_2\gamma_5 &
 \gamma_1\gamma_2\gamma_4 &= - \gamma_3\gamma_5, & \gamma_2\gamma_3\gamma_4 &= -\gamma_1\gamma_5,
\end{align}
establishing a connection with the basis most commonly used by physicists.
\end{exmpl}

We then have the following result:
\begin{prop}[(Generalized) Fierz identity]\label{prop:fierz} If for any two strictly ordered sets $K, L$ there exists a third strictly ordered set $M$ and $c \in \mathbb{N}$ such that $\gamma_K\gamma_L = c\,\gamma_M$, we have for any $\psi_1, \ldots, \psi_4$ in the $n$-dimensional spin representation of the Clifford algebra 
\begin{align}
   \inpr{\psi_1}{\gamma^K\psi_2}\inpr{\psi_3}{\gamma_K\psi_4} &= -\frac{1}{n}\sum_L C_{KL}\inpr{\psi_3}{\gamma^L\psi_2}\inpr{\psi_1}{\gamma_{L}\psi_4}\label{eq:fierzf},
\end{align}
where the constants $C_{LK} \equiv n_Lf_{LK}$, $f_{LK}\in \mathbb{N}$ are
defined via 
$
	  \gamma^K\gamma^L\gamma_K = f_{KL}\gamma^L$ (no sum over $L$).
Here we have denoted by $\inpr{.}{.}$ the inner product on the spinor representation.
\end{prop}

\bibliographystyle{plain}

\begin{thebibliography}{10}

\bibitem{Bhowmick2011}
J.~Bhowmick, F.~D'Andrea, B.~Das, and L.~D\k{a}browski.
\newblock Quantum gauge symmetries in noncommutative geometry.
\newblock {\em arXiv:1112.3622}, 2011.

\bibitem{BS10}
T.~{\noopsort{broek}}van~den Broek and W.D. van Suijlekom.
\newblock Supersymmetric {QCD} and noncommutative geometry.
\newblock {\em Comm. Math. Phys.}, 303(1):149--173, 2010.

\bibitem{Cha94}
A.H. Chamseddine.
\newblock Connection between space-time supersymmetry and noncommutative
  geometry.
\newblock {\em Phys. Lett. B}, B332:349--357, 1994.

\bibitem{CC96}
A.H. Chamseddine and A.~Connes.
\newblock Universal formula for noncommutative geometry actions: {U}nifications
  of gravity and the standard model.
\newblock {\em Phys. Rev. Lett.}, 77:4868--4871, 1996.

\bibitem{CC97}
A.H. Chamseddine and A.~Connes.
\newblock The spectral action principle.
\newblock {\em Comm. Math. Phys.}, 186:731--750, 1997.

\bibitem{CC08}
A.H. Chamseddine and A.~Connes.
\newblock Why the {S}tandard {M}odel.
\newblock {\em J. Geom. Phys.}, 58:38--47, 2008.

\bibitem{CC12}
A.H. Chamseddine and A.~Connes.
\newblock Resilience of the {S}pectral {S}tandard {M}odel.
\newblock {\em J. High Energy Phys.}, 1209:104, 2012.

\bibitem{CCM07}
A.H. Chamseddine, A.~Connes, and M.~Marcolli.
\newblock Gravity and the standard model with neutrino mixing.
\newblock {\em Adv. Theor. Math. Phys.}, 11:991--1089, 2007.

\bibitem{Colafrancesco2010}
S.~Colafrancesco.
\newblock Dark {M}atter in {M}odern {C}osmology.
\newblock {\em arXiv:1004.3869}, 2010.

\bibitem{C94}
A.~Connes.
\newblock {\em Noncommutative geometry}.
\newblock Academic Press, 1994.

\bibitem{C96}
A.~Connes.
\newblock Gravity coupled with matter and the foundation of noncommutative
  geometry.
\newblock {\em Commun. Math. Phys.}, 182:155--176, 1996.

\bibitem{C00}
A.~Connes.
\newblock Noncommutative geometry year 2000.
\newblock {\em math/0011193}, 2007.

\bibitem{CL89}
A.~Connes and J.~Lott.
\newblock Particle models and noncommutative geometry.
\newblock {\em Nuclear Phys. B Proc. Suppl.}, 18:29--47, 1991.

\bibitem{CM07}
A.~Connes and M.~Marcolli.
\newblock {\em Noncommutative Geometry, Quantum Fields and Motives}.
\newblock American Mathematical Society, 2007.

\bibitem{Dabrowski2010}
L.~D\k{a}browski and G.~Dossena.
\newblock Product of real spectral triples.
\newblock {\em Int. J. Geom. Methods Mod. Phys.}, 8(8):1833--1848, 2010.

\bibitem{DGR04}
M.~Drees, R.~Godbole, and P.~Roy.
\newblock {\em Theory and phenomenology of Sparticles}.
\newblock World Scientific Publishing Co., 2004.

\bibitem{DS12}
K.~{\noopsort{dungen}}van~den Dungen and W.D. van Suijlekom.
\newblock Particle physics from almost-commutative spacetimes.
\newblock {\em Rev. Math. Phys.}, 24:1230004, 2012.

\bibitem{Consonni2013}
S.M.~Consonni for~the {ATLAS}~Collaboration.
\newblock Higgs search at {ATLAS}.
\newblock {\em arXiv:1305.3315}, 2013.

\bibitem{Gil84}
P.B. Gilkey.
\newblock {\em Invariance theory, the heat equation and the {A}tiyah-{S}inger
  index theorem}, volume~11 of {\em Mathematics Lecture Series}.
\newblock Publish or Perish, Wilmington, DE, 1984.

\bibitem{GVF00}
J.M. Gracia-Bond\'ia, J.C. V\'arilly, and H.~Figueroa.
\newblock {\em Elements of Noncommutative Geometry}.
\newblock Birkh\"auser {A}dvanced {T}exts, 2000.

\bibitem{ISS03}
B.~Iochum, T.~Sch\"ucker, and C.~Stephan.
\newblock On a {C}lassification of {I}rreducible {A}lmost {C}ommutative
  {G}eometries.
\newblock {\em J. Math. Phys.}, 45:5003--5041, 2004.

\bibitem{KR97}
T.~Krajewski.
\newblock Classification of finite spectral triples.
\newblock {\em J. Geom. Phys.}, 28:1--30, 1998.

\bibitem{LMMS97}
F.~Lizzi, G.~Mangano, G.~Miele, and G.~Sparano.
\newblock Fermion {H}ilbert space and fermion doubling in the noncommutative
  geometry approach to gauge theories.
\newblock {\em Phys. Rev. D}, 55:6357--6366, 1997.

\bibitem{NW96}
P.~{\noopsort{nieuwenhuizen}}van Nieuwenhuizen and A.~Waldron.
\newblock On euclidean spinors and wick rotations.
\newblock {\em Phys. Lett. B}, 389:29--36, 1996.
\newblock arXiv:hep-th/9608174.

\bibitem{S12}
W.D. \noopsort{suijlekom}van Suijlekom.
\newblock Renormalizability conditions for almost-commutative geometries.
\newblock {\em arXiv:1204.4070}, 2012.

\bibitem{OS1}
K.~Osterwalder and R.~Schrader.
\newblock Axioms for {E}uclidean {G}reen's functions \rnum{1}.
\newblock {\em Comm. Math. Phys.}, 31:83--112, 1973.

\bibitem{OS2}
K.~Osterwalder and R.~Schrader.
\newblock Axioms for {E}uclidean {G}reen's functions \rnum{2}.
\newblock {\em Commun. Math. Phys.}, 42:281--305, 1975.

\bibitem{PS96}
M.~Paschke and A.~Sitarz.
\newblock Discrete spectral triples and their symmetries.
\newblock {\em J. Math. Phys.}, 39:6191, 1998.

\bibitem{Rie74}
M.A. Rieffel.
\newblock Morita equivalence for {$C\sp{\ast} $}-algebras and
  {$W\sp{\ast}$}-algebras.
\newblock {\em J. Pure Appl. Algebra}, 5:51--96, 1974.

\bibitem{SS74}
A.~Salam and J.~Strathdee.
\newblock Super-gauge transformations.
\newblock {\em Nuclear Phys. B Proc. Suppl.}, 76:477--201, 1974.

\bibitem{Vanhecke2007}
F.J. Vanhecke.
\newblock On the product of real spectral triples.
\newblock {\em Lett. Math. Phys.}, 50:157--162, 2007.

\bibitem{V06}
J.C. V\'arilly.
\newblock {\em An Introduction to Noncommutative Geometry}.
\newblock European Mathematical Society, 2006.

\bibitem{wessbagger1992}
J.~Wess and J.~Bagger.
\newblock {\em Supersymmetry and Supergravity}.
\newblock Princeton University Press, 1992.

\end{thebibliography}
\providecommand{\noopsort}[1]{}

\end{document}